\documentclass[9pt]{article}

\usepackage{color}
\usepackage{comment}

\newif\ifusepackageA
\usepackageAtrue      
\usepackage[utf8]{inputenc}
\usepackage[T1]{fontenc}
\usepackage{lmodern}
\usepackage{float} 
\usepackage{booktabs}
\usepackage{xcolor}
\usepackage{stmaryrd} 
\usepackage[flushmargin,ragged,symbol*]{footmisc}
\usepackage{enumitem} 
\usepackage{multirow}

\usepackage{stmaryrd}
\usepackage{algorithm}
\usepackage{algpseudocode}

\usepackage[dvipsnames]{xcolor}
\usepackage{authblk}
\usepackage{amsmath}
\usepackage{amssymb}
\usepackage{amsthm}

\usepackage{mathtools}

\newtheorem{theorem}{Theorem}[section]
\newtheorem{lemma}[theorem]{Lemma}
\newtheorem{corollary}[theorem]{Corollary}
\newtheorem{definition}[theorem]{Definition}
\newtheorem{proposition}[theorem]{Proposition}

\theoremstyle{remark}
\newtheorem*{example}{Example}
\newtheorem{remark}{Remark}
\usepackage{etoc}

\theoremstyle{definition}
\newtheorem{assumption}[theorem]{Assumptions} 

\newcommand{\Z}{\textbf{\textit{Z}} }

\makeatletter
\newcommand{\Bitem}[1]{\item[\textnormal{(#1)}]\def\@currentlabel{\textnormal{#1}}}
\makeatother

\ifusepackageA
    \usepackage[left=2.4cm, right=2.4cm, marginparwidth=3.2cm]{geometry}
    \usepackage[colorlinks=true,allcolors=NavyBlue]{hyperref}
    \usepackage{threeparttable}
    \usepackage{tablefootnote}
    \usepackage{parskip} 
    \usepackage{natbib}
    \usepackage[capitalize,noabbrev]{cleveref}
\else
    \usepackage{geometry}
\fi

\usepackage{comment}
\usepackage{lipsum}

\usepackage{comment}
\usepackage{lipsum}

\newcommand{\algname}{DExtrI}
\usepackage{tikz-cd}

\newcommand{\C}{\mathbb{C}}
\newcommand{\eps}{\varepsilon}
\newcommand{\Th}{\Theta}
\renewcommand{\th}{\tilde h}

\newcommand{\tf}{\tilde f}
\newcommand{\tg}{\tilde g}
\newcommand{\tb}{\tilde\beta}
\newcommand{\tB}{\widetilde{\mathcal B}}
\newcommand{\tK}{\widetilde K}
\newcommand{\Btr}{\mathcal B}
\newcommand{\Xtr}{\mathcal{X}_{\mathrm{tr}}}

\DeclareMathOperator{\supp}{supp}

\providecommand{\R}{\mathbb{R}}
\providecommand{\tS}{\widetilde S}
\providecommand{\cP}{\mathcal P}
\providecommand{\cR}{\mathcal R}

\usepackage{subcaption}   
\usepackage{caption}      

\title{Distributional Extrapolation for Interactions}
\author[1]{Marin \v{S}ola\thanks{Co-corresponding author: \texttt{marin.sola@stat.math.ethz.ch}}}
\author[2]{Xinwei Shen}
\author[1]{Peter B\"uhlmann\thanks{Co-corresponding author: \texttt{buhlmann@stat.math.ethz.ch}}}
\date{}

\affil[1]{Seminar for Statistics, ETH Z\"urich}
\affil[2]{Department of Statistics, University of Washington}

\begin{document}

\maketitle

\begin{abstract}
Predicting combinatorial effects from limited-range observations is a fundamental challenge in many scientific domains, including drug discovery and hyperparameter optimization. We study combinatorial extrapolation, where training data consists of axis-aligned samples with only one active covariate, while test-time inputs involve multiple simultaneously active covariates.
We introduce DExtrI, a method for extrapolating interaction effects beyond the support of the training data. We provide theoretical guarantees characterizing when such extrapolation is possible. Empirical results on synthetic and real-world datasets demonstrate that \algname\ successfully generalizes to unseen combinations of covariates. Our approach enables applications such as predicting previously untested drug combinations and improving the efficiency of hyperparameter optimization.
\end{abstract}


\ifusepackageA
    \section{Introduction}
\else
\fi

Extrapolation beyond the training distribution is a central challenge in science and engineering, arising in applications such as predicting effective drug combinations \cite{liu2020drugcombdb} or silencing genes in genomics. While many modern models achieve strong in-distribution performance, their behavior outside the support of the training data is often unstable and poorly understood \cite{shen2024engression}. This limitation is particularly severe in settings where acquiring fully combinatorial data is prohibitively expensive or infeasible.

In this work, we focus on combinatorial extrapolation, a regression setting in which training data are axis-aligned, meaning that each sample contains only a single active covariate, while test-time inputs involve multiple simultaneously active covariates. This scenario naturally arises in several domains: in drug discovery, one aims to predict combination effects from single-drug responses; in genomics, multi-gene interactions must be inferred from single-gene perturbations; and in hyperparameter optimization of complex models, joint configurations must be predicted from marginal parameter sweeps (see Section~\ref{sec:empiricalvalidation}). Despite their diversity, these problems share a common structure: models are trained on isolated covariate activations but are required to generalize to unseen combinations at test time.


In this paper, we introduce \textbf{\algname}\ (\textbf{D}istributional \textbf{Extr}apolation for \textbf{I}nteractions), a method that enables combinatorial extrapolation over unseen covariate combinations within the Cartesian product of observed marginal supports. By modeling a random variable $Y\in\mathbb{R}$ as a sum of additive components and projection-pursuit regression functions for interactions, and using a pre-additive error structure,
\algname\ enables generative-style extrapolation to unseen covariate combinations.
We establish identification results showing that, within an identifiable structural model class, the entire response surface can be recovered from the conditional distributions on the axis-aligned training support.
As we will illustrate in Section \ref{sec:empiricalvalidation}, this represents a significant advance over standard machine learning and statistical algorithms based on neural networks, such as regularized empirical risk minimization or the Engression approach \citep{shen2024engression} for (non-combinatorial) extrapolation, which struggle to generalize to combinatorial extrapolation.  
\ifusepackageA    
    \begin{figure}[htbp]
        \centering
        \includegraphics[width=0.75\linewidth]{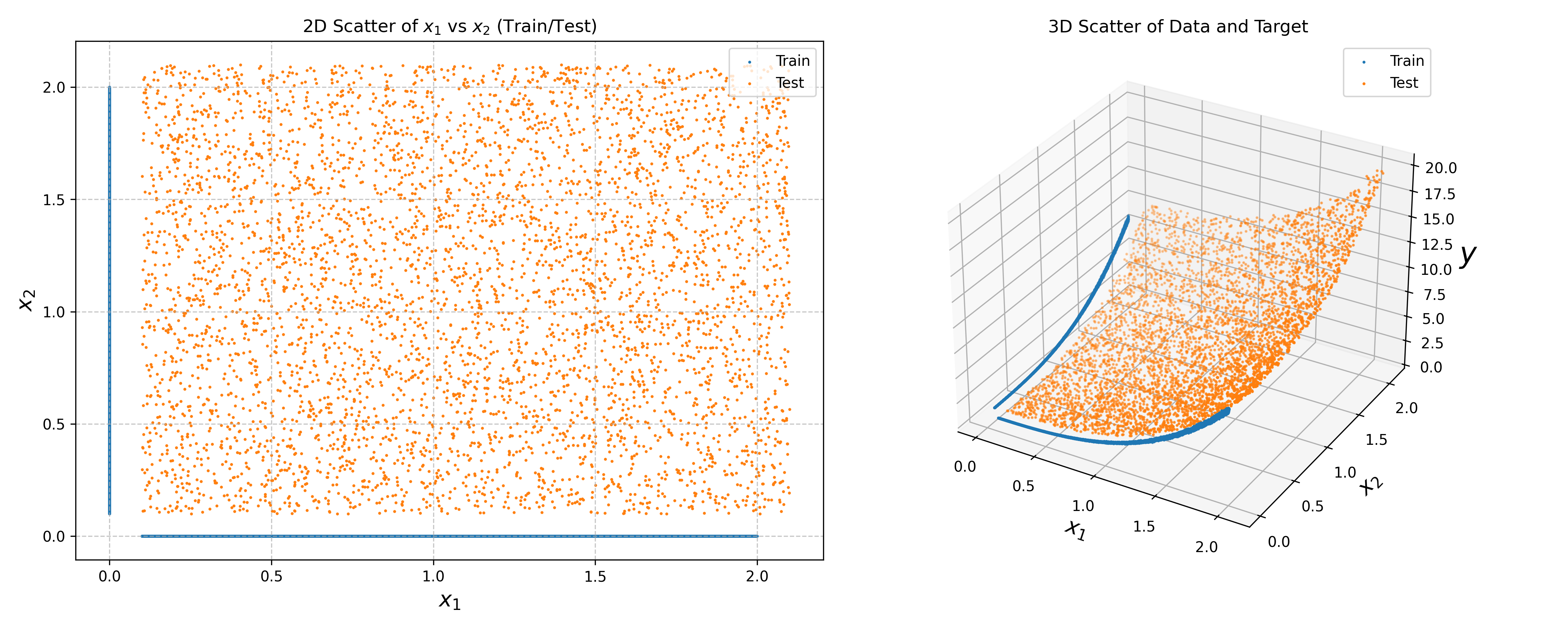}
        \caption{
        Combinatorial extrapolation (conceptual): the data available during training is aligned with the axes (blue) and test data (orange) consists of their combinations. 
        The left subplot displays the distribution of input features $x_1, x_2$ for both training (blue) and test (orange) data. The right subplot extends this to 3-dimensional visualization, showing the relationship between $x_1, x_2$, and the target variable $y$. }\label{fig:first_scatterplots}
    \end{figure}
\fi

\ifusepackageA \subsection{Related Work}\label{sec:relatedwork}
\else 
\fi

In the next subsection, we discuss closely related work while more general connections are given in the following.

\ifusepackageA \subsubsection{Combinatorial and Compositional Extrapolation}
\fi

Our work is directly related to combinatorial and compositional generalization, where the goal is to predict outcomes for novel combinations of features that were not jointly observed during training. Matrix completion and missing-not-at-random formulations address related problems under low-rank or spectral decay assumptions \citep{simchowitz2023combinatorial}, while transductive approaches reformulate extrapolation as an imputation problem within an expanded support \citep{netanyahu2023learning}. 

Theoretical work on compositional generalization studies when learned components can be recombined to yield correct predictions on unseen combinations. In particular, recent analyses characterize structural conditions under which such extrapolation is possible, as well as fundamental failure modes when these conditions are violated \citep{fu2024general,lippl2024does}. 
In \cite{ma2026deeplearningmissingdata}, the authors derive optimal statistical performance for regression when modeling the data missingness mechanism with neural networks. 

\ifusepackageA \subsubsection{General relations}
\fi
Empirical and theoretical work has shown that extrapolation behavior strongly depends on model structure or inductive bias. For instance, empirical studies in protein engineering demonstrate that different model classes, ranging from simple parametric models to deep architectures, can vary substantially in their ability to generalize far beyond the training distribution \citep{freschlin2024neural}. Similarly, several works highlight the role of model scale and architecture in determining generalization outside the training support \citep{kawaguchi2017generalization,xu2020neural,chatterjee2022generalization}. 

Another line of research addresses extrapolation through robustness to distributional shifts. Approaches based on risk extrapolation and invariant learning introduce regularization principles designed to reduce sensitivity to domain shifts \citep{nagarajan2020understanding,krueger2021out}. Related ideas include anchor regression and invariant risk minimization, which aim to enforce stability across environments \citep{rothenhausler2021anchor, arjovsky2019invariant}. A complementary perspective emphasizes the curse of dimensionality: in high-dimensional settings, most test points lie outside the convex hull of the training data, effectively turning generalization into an extrapolation problem \citep{balestriero2021learning,nagarajan2020understanding,freschlin2024neural}. This viewpoint motivates methods that explicitly target out-of-support prediction, see also \citep{pfister2024extrapolation}.

A prominent approach to extrapolation is to model the full conditional distribution rather than only the conditional mean. The engression framework \citep{shen2024engression} and related distributional regression methods learn conditional distributions via proper scoring rules and provide guarantees under structural assumptions on the data-generating process, especially by using pre-additive noise instead of the usual post-additive version. These methods can extrapolate beyond the observed training support due to the pre-additive noise structure, something we crucially exploit in our approach as well. 
Related approaches include latent-variable and autoencoder-based models \citep{shen2024distributional}, which aim to learn representations that generalize beyond observed inputs. While these methods are flexible, they do not impose structure tailored to combinatorial extrapolation from highly constrained training supports with axis-aligned data.

\ifusepackageA \subsubsection{Positioning of Our Work}
\fi
Our contribution is a model class combining the pre-additive noise structure of engression
\citep{shen2024engression} with the ridge representation of interactions used in
projection pursuit regression \citep{Friedman1981projectionpursuit,
huber1985projection}. The pre-additive noise enlarges the range of the index
arguments beyond that induced by the covariates alone, which is what makes the
axis-aligned design informative about the interaction components.
Section~\ref{sec:identifiabilityExtrapolation} shows that the model is thereby
identified from axis-aligned data and that the conditional law of $Y$ given $X=x$
is determined on the Cartesian product of the marginal supports; neither
conclusion holds for the classical models with post-additive noise, and we briefly discuss some theoretical challenges at the beginning of that section.
Section~\ref{sec:empiricalvalidation} compares the resulting estimator with
engression, neural additive models, and empirical risk minimization on synthetic
designs, drug combination screens, and hyperparameter optimization benchmarks.

\subsection{Outline and Notation}
The paper is organized as follows: Section \ref{sec:Methodology} presents the \algname\ framework and Section \ref{sec:identifiabilityExtrapolation} its identification and extrapolation theory; Section \ref{sec:empiricalvalidation} 
reports empirical results; and Section \ref{sec:discussion} concludes.  

Regarding notation: For a random variable $X$, we write $X\overset{d}{=}\tilde{X}$
to denote equality in distribution, meaning $P[X\leq t]=P[\tilde{X}\leq t]$ for a
probability measure $P$ and all $t$. We denote the $\alpha$-quantile of a random
variable $X$ by $Q_\alpha(X)=\inf\{x:P[X\leq x]\geq\alpha\}$ for $X\sim P$. The
index support set of $\beta_k\in\mathbb{R}^p$ is denoted by $\mathcal{B}_k$, in
symbols, $\mathcal{B}_k=\{i:\beta_{ki}\neq 0\}$. $\mathcal{X}_{\mathrm{tr}}$ and
$\mathcal{X}_i$ denote the training support and marginal coverage of $X_i$,
respectively. For $x\in\mathbb{R}^p$, $\bar{x}_k\in\mathbb{R}$ denotes the
argument of the $k$-th interaction function $g_k$, with
$\bar{x}_k=\beta_k^\top x+\eta_k, \eta_k = h_k(\varepsilon)$. $\varepsilon$ is a random variable distributed
as $\mathcal{N}(0,1)$, with cumulative distribution function
$\Phi(t)=\frac{1}{\sqrt{2\pi}}\int_{-\infty}^t e^{-s^2/2}\,ds$ and quantile
function $\Phi^{-1}:(0,1)\to\mathbb{R}$. For a positive integer $K$, $[K]$
denotes the set $\{1,2,\dots,K\}$. $P_{\mathrm{tr}}$ denotes the distribution of
the training data. $U[a,b]$ denotes the uniform distribution on $[a,b]$ and
$\mathcal{N}(\mu,\sigma^2)$ the Gaussian distribution with mean $\mu$ and variance
$\sigma^2$. $P_{Tx}$ and $P_{\tilde{T}x}$ denote the pointwise laws of
$T(x,\varepsilon)$ and $\tilde{T}(x,\varepsilon)$, respectively. We denote by $D$
the square root of the energy distance as defined in \cite{szekely2003estatistics}. $\mathbb{C}[u]:=\{\text{polynomials in }u\text{ with coefficients in }\mathbb{C}\}$.

\section{The DExtrI Model and Algorithm}\label{sec:Methodology}

We model the response variable $Y \in \mathbb{R}$ using an additive structure combined with $K$ interaction components, each represented by a ridge function. For covariates $X = (X_j)_{j=1}^p \in \mathbb{R}^p$, the model is given by
\begin{align}\label{adeng_truemodel}
    Y \;=\; \sum_{j=1}^p f_j(X_j) \;+\; \sum_{k=1}^K g_k(\beta_k^\top X + h_k(\varepsilon)),
\end{align}
where $f_j, g_k , h_k : \mathbb{R} \to \mathbb{R}$ are differentiable univariate functions, $\beta_k$ are coefficient vectors, 
and $\varepsilon \sim \mathcal{N}(0,1)$ is independent of $X$. The same $\varepsilon$ is used for all $k$, which induces dependence among the $h_k(\varepsilon)$. We always assume that each $g_k$ is nonlinear; otherwise, the model reduces to an additive model with the usual post-additive noise, for which combinatorial extrapolation is straightforward but does not permit interactions. An important additional assumption is that $\sum_{k=1}^K g_k(\beta_k^\top X + h_k(\varepsilon))$ is strictly monotone in $\varepsilon$. We will discuss this assumption below in several places.

We refer to $h_k(\varepsilon)$ (sometimes also written $\eta_k$) as \emph{pre-additive noise} as in \cite{shen2024engression}, since it is added to the argument of the interaction function prior to applying $g_k$. This feature is crucial for extrapolation: it enlarges the effective support of the argument of each $g_k$. For instance, if $h_k(\varepsilon)$ has full support on $\mathbb{R}$, then the argument $\beta_k^\top X + h_k(\varepsilon)$ also spans $\mathbb{R}$, even when $\beta_k^\top X$ itself has limited support \cite{shen2024engression}.

Model \eqref{adeng_truemodel} is related to projection pursuit regression \citep{Friedman1981projectionpursuit, hall1989onprojectionpursuit} and single-index models \citep{ICHIMURA199371}. Projection pursuit regression is known to provide flexible function approximation as $K$ increases; thus, moderate values of $K > 1$ already yield a substantially richer function class than single-index models. There are key distinctions though: due to the fact that existing work mentioned above does not consider combinatorial extrapolation, the classical approaches do not need to make a monotonicity assumption. Also, in contrast to \eqref{adeng_truemodel}, these models assume post-additive errors. And finally, we use adjustment with additive functions which enables combinatorial extrapolation with corresponding monotonicity assumptions only for the interaction component: this is a natural separation into parts with no (additive) and with combinatorial interaction, something which also gave better empirical results.
We also note that pre-additive noise models have been studied in the context of causal inference \citep{zhang2012identifiability}.

Identifiability results for model \eqref{adeng_truemodel} are provided in Section~\ref{sec:identifiabilityExtrapolation}. To ensure identifiability, constraints must be imposed on the functions $f_j$ and $g_k$, for example, by requiring zero mean and including a global intercept term. We also assume that the representation in \eqref{adeng_truemodel} is \emph{minimal} in the sense that no two interaction components have the same covariate support. This minimality condition is a structural assumption required for identifiability.
Furthermore, as formalized in Assumption~\hyperref[assumptions]{A4}, each interaction function $g_k$ is assumed to depend on at least two covariates, ensuring that it represents a genuine interaction rather than a univariate effect; and finally, we will assume that $\sum_{k=1}^K g_k(\beta_k^\top X + h_k(\eps))$ is strictly monotone in $\varepsilon$ (through $h_k(\varepsilon)$), which holds, for example, if each $g_k$ and $h_k$ is strictly increasing. 

We note that model \eqref{adeng_truemodel} with $K=1$ is closely related to Engression \citep{shen2024engression}, with an additive function component here replacing the linear component used in Engression. Engression develops the theory of how distributional information permits non-combinatorial extrapolation of a nonlinear function in the univariate setting $p=1$. In contrast, we consider combinatorial extrapolation from axis-aligned data in the multivariate setting $p \ge 2$, with $K \ge 1$, and establish identification of the multivariate index directions.

\ifusepackageA\paragraph{Some intuition for why model \eqref{adeng_truemodel} enables combinatorial extrapolation.}
\fi 

A crucial component  for identifiability from axis-aligned training distributions is given by strict monotonicity of $\sum_{k=1}^K g_k(\beta_k^\top X + h_k(\varepsilon))$ in $\varepsilon$. Consider for simplicity the following special case of \eqref{adeng_truemodel} with no additive functions and $K=1$:
\begin{eqnarray*}
    Y = g(\beta^\top X + h(\varepsilon))
\end{eqnarray*}
One can learn the conditional training distribution $P_{\text{tr}}(Y|X=x)$ for all axis-aligned $x$. For fixed axis-aligned $x$, we can do distributional matching of $P_{\text{tr}}(Y|X=x)$ and $g(\beta^\top x + h(\varepsilon))$, e.g.\ matching of all quantiles . When varying across all axis-aligned $x$ within $g(\beta^\top x + h(\varepsilon))$ and matching with $P_{\text{tr}}(Y|X=x)$, due to the linear argument $\beta^\top x$ and monotonicity one can then identify $\beta$ and $g$ (and $h$) from axis-aligned data. The mathematical details are given in Section \ref{sec:identifiabilityExtrapolation}. 
In summary, we see that two ingredients are crucial: learning the conditional training distribution $P_{\text{tr}}(Y|X=x)$ for the axis-aligned $x$ and the monotonicity of $g(\beta^\top X + h(\varepsilon))$. We will ensure the first by optimizing a distributional loss (Section \ref{subsec:energyscore}) and the second by assumption. 

Obviously, as with projection pursuit regression or single-index models, due to non-linearity of $g$, we obtain interactions in $X$, e.g.\ with $g(u) = \exp(u)$.

\ifusepackageA \paragraph{Generative interpretation.}
\fi
Model \eqref{adeng_truemodel} admits a natural generative interpretation of the conditional distribution of $Y$ given $X$. For any fixed $x \in \mathbb{R}^p$, define
\[
F(x) + G(x,\varepsilon)
\;:=\;
\sum_{j=1}^p f_j(x_j)
\;+\;
\sum_{k=1}^K g_k\big(\beta_k^\top x + h_k(\varepsilon)\big),
\]
where $\varepsilon \sim \mathcal{N}(0,1)$ is independent of $X$. Then
\[
Y \mid X = x \;\overset{d}{=}\; F(x) + G(x,\varepsilon),
\]
that is, the model specifies the full conditional distribution of $Y$ given $X=x$ via a deterministic transformation of a low-dimensional noise variable.

This representation makes the model directly generative: given $x$, one can sample $\varepsilon$ and obtain draws from the conditional distribution $Y \mid X = x$. While such a generative view is not unique to our setting, it is particularly convenient here because the noise distribution is fixed and known, and the stochasticity enters through a one-dimensional latent variable passed through smooth transformations.

This perspective is central to our estimation approach (see Section~\ref{subsec:energyscore}), where we compare conditional distributions rather than point predictions. Moreover, it provides a useful lens for combinatorial extrapolation: once the structural components $f_j$, $g_k$, and $h_k$ are learned, the model can generate plausible responses at covariate values outside the training support by propagating the same latent noise $\varepsilon$ through the learned transformations.

\subsection{Distributional regression via energy score optimization}\label{subsec:energyscore}

Our goal is to estimate the true model components using a flexible parameterization, such as a neural network. The parameters are learned by minimizing the empirical counterpart of the expected energy score, which is a strictly proper scoring rule for probability distributions \citep{gneiting2007strictly}. In the following, we present the population formulation, corresponding to the limit of infinitely many data samples.

Let $Z \sim P$ be a random vector with realization $z$. The energy score is defined as
\begin{align*}
    \mathrm{ES}(P, z) = \frac{1}{2}\mathbb{E}_P \| Z - Z' \| - \mathbb{E}_P \| Z - z \|,
\end{align*}
where $Z, Z' \stackrel{\text{i.i.d.}}{\sim} P$, and $\|\cdot\|$ denotes the Euclidean norm. A higher energy score indicates a better agreement between the observation $z$ and the distribution $P$. Moreover, the energy score is strictly proper \citep{szekely2003estatistics, SzekelyRizzo2023}, meaning that for any distribution $P'$,
\begin{align*}
    \mathbb{E}_{Z \sim P}[\mathrm{ES}(P, Z)] \geq \mathbb{E}_{Z \sim P}[\mathrm{ES}(P', Z)],
\end{align*}
with equality if and only if $P' = P$.
In the univariate case, which is the setting considered here, the energy score coincides with the negative continuous ranked probability score (CRPS).
To estimate the conditional distribution, we minimize the negative energy score using samples from a generative model. Specifically, consider the model
\[
T(X,\varepsilon) = F(X) + G(X,\varepsilon) 
= \sum_{j=1}^p f_j(X_j) + \sum_{k=1}^K g_k\big(\beta_k^\top X + h_k(\varepsilon)\big),
\]
where $\varepsilon \sim \mathcal{N}(0,1)$. The population objective function is given by
\begin{align}\label{obj_fct} 
    \hspace*{-0.5cm}\mathcal{L}(T) 
    = \mathbb{E} \bigg[ \|Y - T(X,\varepsilon) \| 
    - \frac{1}{2}\| T(X,\varepsilon) - T(X,\varepsilon') \| \bigg],
\end{align}
where $\varepsilon, \varepsilon' \stackrel{\text{i.i.d.}}{\sim} \mathcal{N}(0,1)$, and $(X,Y) \sim P_{\mathrm{tr}}$ denotes the training distribution (e.g., obtained from axis-aligned data).
We define the minimizer of \eqref{obj_fct} over a model class $\mathcal{M}$ as
\[
\tilde{T} \in \arg \min_{T \in \mathcal{M}} \mathcal{L}(T),
\]
and refer to $\tilde{T}$ as the \emph{engression solution}. In the population limit, this solution recovers the conditional distribution $P_{\mathrm{tr}}(y \mid x)$ in the sense of equality in distribution. This is formalized in the following result.

\begin{lemma}[\cite{shen2024engression}, Proposition 1]\label{lemma_engression_prop1}
Let $\mathcal{X}_{\mathrm{tr}}$ denote the support of the training distribution $P_{\mathrm{tr}}$. Assume that there exists a function $T \in \mathcal{M}$ such that, for all $x \in \mathcal{X}_{\mathrm{tr}}$, the random variable $T(x,\varepsilon)$ (with $\varepsilon \sim \mathcal{N}(0,1)$) follows the conditional distribution $P_{\mathrm{tr}}(y \mid x)$. Then, any population minimizer $\tilde{T}$ of \eqref{obj_fct} satisfies
\[
\tilde{T}(x,\varepsilon) \sim P_{\mathrm{tr}}(y \mid x)
\]
for $P_X-$a.e. $x\in \mathcal{X}_{\mathrm{tr}}$.
\end{lemma}

In other words, minimizing $\mathcal{L}$ enforces that $\tilde{T}$ matches the true conditional distribution in law. This property is crucial, as it enables the identification of the structural components of $T$ within the model class $\mathcal{M}$.

In the sequel, we assume that $\tilde{T}$ is well-approximated through training and thus effectively attains the minimum of $\mathcal{L}$. By Lemma~\ref{lemma_engression_prop1}, it therefore reproduces the true conditional distribution $P_{\mathrm{tr}}(y \mid x)$ for $P_X-$almost every $x\in \Xtr$. Finally, $\mathcal{M}$ represents a flexible class of generative models; in practice, we choose 
$\mathcal{M}$ as a sufficiently expressive class of neural networks, see "Architecture and implementation" in Section \ref{sec:algo}.

\subsection{The DExtrI algorithm}\label{sec:algo}

We now describe the practical implementation of the proposed DExtrI framework, including architectural choices, optimization, and hyperparameter settings.

\emph{Architecture and implementation.}
All component functions $f_j$, $g_k$, and $h_k$ are parameterized using fully connected neural networks. 
Unless stated otherwise, each hidden layer consists of 100 neurons. The networks representing $f_j$ have four layers, 
while those corresponding to $g_k$ are similar, with four layers and a residual connection.

\emph{Optimization.}
Model parameters are optimized using the Adam algorithm~\cite{kingma2014adam} with a fixed learning rate of $10^{-4}$ 
across all experiments. The batch size is typically set to approximately 1000, which in many settings corresponds 
to full-batch training. The number of training epochs depends on the experiment: in synthetic settings, models are trained for roughly 
300 epochs, whereas in real-data applications training may extend to several thousand epochs. In most cases, 
we train for approximately 1000 epochs. Empirically, provided the batch size is sufficiently large, we do not 
observe noticeable overfitting even when training for longer.

\emph{Hyperparameters.}
Hyperparameter selection is guided by simplicity and empirical robustness. In particular, the number of interaction 
components $K$ is typically fixed to $K=1$, as increasing $K$ rarely leads to substantial performance gains. 
While domain knowledge may motivate alternative choices in specific applications, the method appears largely 
insensitive to moderate misspecification of this parameter. 
Importantly, such insensitivity is not easily detectable from axis-aligned training data alone and is better 
assessed through out-of-sample performance, for example via mean squared error on a test set. This observation 
is consistent with findings in the engression literature~\citep{shen2024engression}, where performance is often 
robust to hyperparameter choices. One possible explanation is the use of the energy score loss, which is 
substantially different from classical objectives such as squared error.

Finally, the proposed framework naturally induces a full conditional distribution through its dependence on the 
noise variable $\boldsymbol{\varepsilon}$. This enables straightforward estimation of functionals such as the 
conditional mean and conditional quantiles:
\[
\tilde{\mu}(x) = \mathbb{E}_{\boldsymbol{\varepsilon}}[\tilde{T}(x, \boldsymbol{\varepsilon})], 
\qquad 
\tilde{q}_\alpha(x) = Q_\alpha(\tilde{T}(x, \boldsymbol{\varepsilon})).
\]
In the finite-sample setting, these quantities are approximated empirically using samples from the learned 
conditional distribution. 

Note that the population result in Lemma~\ref{lemma_engression_prop1} concerns recovery of the conditional distribution on the training support for a minimizer over a chosen model class $\mathcal{M}$. The identifiability and extrapolation results in Section~\ref{sec:identifiabilityExtrapolation} apply when the model components in $\mathcal{M}$ satisfy the corresponding assumptions specified there. The neural-network implementation described above does not explicitly enforce all of these assumptions; the same is true for the implementation of the original engression algorithm \cite{shen2024engression}.

\section{Model identifiability and extrapolation properties}\label{sec:identifiabilityExtrapolation}
We develop here precise results on identifiability and extrapolation properties, from training data with axis-aligned support. 

The theoretical problem differs from classical single-index and projection-pursuit models. The covariate support is only a union of coordinate axes and has empty interior, so interaction directions cannot be identified using ordinary multivariate variation or mixed derivatives. In addition, the noise enters inside the nonlinear ridge functions, requiring the ridge directions, profiles, and noise transformations to be disentangled jointly from conditional distributions on the restricted support.

\subsection{Definition of axis-aligned support}  
We consider the \emph{combinatorial extrapolation} setting, where the training data support of the covariates is
\begin{align}\label{defn:train_supports}
\mathcal{X}_{\text{tr}} = \bigcup_{j=1}^p \mathcal{X}_j, \enspace \text{with} \enspace \mathcal{X}_j = \{t e_j : t \in [a_j, A_j]\} \subseteq \mathbb{R}^p,
\end{align}
where $e_j = (0,0,\ldots , 0, 1, 0 \ldots , 0)^\top \in \R^p$ denotes the $j$-th standard basis vector, with the $j$-th component being 1 and 0 elsewhere; also, $[a_j, A_j] \subset \mathbb{R}$ are (nonempty) intervals containing zero. That is, the training data lies along the coordinate axis only, see figure \ref{fig:first_scatterplots}. The test data support for the covariates is the full product space \begin{align}\label{defn:test_support}\mathcal{X}_{\text{test}} = \prod_{j=1}^p [a_j, A_j].\end{align}

\subsection{Identifiability}\label{subsec:identif}
Denote by $T(X,\varepsilon)$ the right-hand side of \eqref{adeng_truemodel}. We define the model class as 
\ifusepackageA
    \begin{align*}
        \mathcal{M} &=\big\{ T(x,\varepsilon) = F(x) + G(x,\varepsilon) : F\in\mathcal{F},\ G\in\mathcal{G} \big\},\ \mbox{where} \\[4pt]
        \mathcal{F} &= \big\{F : F(x)=\sum_{j=1}^p f_j(x_j)\ \text{for differentiable } f_1,\dots,f_p\big\},\\[4pt]
        \mathcal{G} &= \Big\{G : G(x,\varepsilon)=\sum_{k=1}^K g_k(\beta_k^\top x + \eta_k),\ \varepsilon\mapsto G(x,\varepsilon) \ \text{strictly increasing for every }x, \\ & \hspace{2cm}\text{for } g_k, h_k \text{ differentiable}  
         \Big\} ,
    \end{align*}
\fi
Requiring $G(x,\cdot)$ to be increasing fixes the global orientation of the latent variable. Since $\varepsilon\overset{d}{=}-\varepsilon$, the opposite orientation is observationally equivalent under the reflection $\varepsilon\mapsto -\varepsilon$. We additionally adopt the following normalization for the individual warps;
without loss of generality we assume each $h_k$ is strictly increasing, meaning, if
$h_k$ is decreasing, replace $(\beta_k,g_k,h_k)$ by
$(-\beta_k,\,g_k(-\cdot),\,-h_k)$, which leaves $g_k(\beta_k^\top
x+h_k(\varepsilon))$ unchanged. This ensures that each warp $h_k$ preserves the
ordering of $\varepsilon$.

We first give a general definition of identifiability that naturally arises in our model class.
\begin{definition}\label{DEF_model_is_identifiable}
    We say the model is \emph{identifiable in the model class ${\cal M} = ({\cal F},{\cal G})$} if the following holds. Suppose two parameterizations
\ifusepackageA
    \[
      \big((\beta_k)_{k=1}^K,\ (f_j)_{j=1}^p,\ (g_k)_{k=1}^K,\ (h_k)_{k=1}^K\big)
      \quad\text{and} 
      \quad
      \big((\tilde{\beta}_k)_{k=1}^{\tilde K},\ (\tilde f_j)_{j=1}^p,\ (\tilde g_k)_{k=1}^{\tilde K},\ (\tilde h_k)_{k=1}^{\tilde K}\big),
      \]
\fi
where $(f_j,\tilde{f}_j)_{j=1}^p \in {\cal F}, (g_k,h_k)_{k=1}^K \in {\cal G}, (\tilde{g}_k,\tilde{h}_k)_{k=1}^{\tilde{K}} \in {\cal G}$ yield the same observational conditional distribution \(P_{\mathrm{tr}}(y\mid x)\) for all \(x\in\mathcal X_{\text{tr}}\). Then \(K=\tilde K\), and there exist constants \(b_1,\dots,b_p\), \(\gamma_1,\dots,\gamma_K\) and nonzero scalars \(r_1,\dots,r_K\) such that for all \(j\in[p]\) and \(k\in[K]\) after relabeling
    \ifusepackageA
        \begin{align*}
          f_j(x_j) &= \tilde f_j(x_j) + b_j, \enspace
          g_k(\bar x_k) = \tilde g_k\!\big(\tfrac{1}{r_k}\bar x_k\big) + \gamma_k, \enspace 
          h_k(\varepsilon) = r_k\,\tilde h_k(\varepsilon),\enspace 
          \beta_k = r_k\,\tilde\beta_k,
        \end{align*}
    \fi
where \(\bar x_k=\beta_k^\top x + h_k(\varepsilon)\). Moreover, $\sum_{k=1}^K \gamma_k + \sum_{j=1}^p b_j=0$.
\end{definition}
It follows that the support of equivalent models remains the same. That is, denote the $k$-th support of the true coefficients by
\begin{eqnarray*}
{\cal B}_k = \mathrm{supp}(\beta_k) = \{j;\ (\beta_k)_j \neq 0\},\ \ k =1, \ldots ,K.
\end{eqnarray*} 
For a $\tilde{\beta}_k$ from an equivalent model, identifiability implies that 
\begin{eqnarray*}
    \mathrm{supp}(\tilde{\beta}_k) = {\cal B}_k\ \mbox{for all}\ k=1,\ldots K,
\end{eqnarray*}
and for $r_k$ in the Definition \ref{DEF_model_is_identifiable} we have that \(r_k=\beta_{k j}/\tilde\beta_{k j}\) for any \(j\in\mathcal B_k\). 
In words, the model is identifiable up to the obvious affine or scale indeterminacies (additive constants for the \(f_j\), and an additive redistribution among the component functions and per-term rescaling of $(\beta_k, h_k)$ with the corresponding inverse argument rescaling of $g_k$). Under mild initial or normalization-type constraints, the additive constants vanish and the scaling degrees of freedom are fixed. 
Identifiability is in terms of an alternative model $\tilde{T}(\cdot,\cdot)$, which can be interpreted as arising from the minimization of the energy score in~\eqref{obj_fct}. Moreover, it is assumed that the models specified by the two parametrizations, $(f_j)_{j=1}^p, (\beta_k,g_k,h_k)_{k=1}^K$ and $(\tilde{f}_j)_{j=1}^p, (\tilde{\beta}_k,\tilde{g}_k,\tilde{h}_k)_{l=1}^{\tilde{K}}$, are observationally equivalent, or $T(x,\varepsilon)\overset{d}{=}\tilde{T}(x,\varepsilon)$ for every $x\in \mathcal{X}_{\mathrm{tr}}$, and $\varepsilon \sim \mathcal{N}(0,1)$.

\begin{assumption}\label{assumptions}
We impose the following conditions on every parameterization in the model
class $\mathcal{M}$. 

\begin{enumerate}[label=\textnormal{(A\arabic*)},ref=\textnormal{A\arabic*},leftmargin=3.2em]
    \item\label{as:A1} \textbf{Functions.}
    Each $g_k$ is continuous and not a polynomial; each $f_j$ is continuous.

    \item\label{as:A2} \textbf{Warp normalization.}
    Each $h_k$ is strictly monotone with median zero, so that $h_k(0)=0$,
    and
    \[
        h_k(\varepsilon)\to\pm\infty
        \qquad\text{as}\qquad
        \varepsilon\to\pm\infty.
    \]

    \item\label{as:AR} \textbf{Analyticity.}
    Each $g_k$, $f_j$, and $h_k$ is real analytic.

    \item\label{as:Supp} \textbf{Supports.}
    $|\mathcal{B}_k|\geq 2$ for every $k$, and the supports
    $\{\mathcal{B}_k\}_{k=1}^K$ are pairwise distinct.

    \item\label{as:C} \textbf{Non-proportional warps.} The warps $\{h_k\}_{k=1}^K$ are pairwise non--proportional.

    \item\label{as:G} \textbf{Distinct growth orders.}
    Any two non-proportional warp functions that occur in parameterizations
belonging to $\mathcal M$ have distinct asymptotic growth orders:
for any such $h,\widetilde h$,
\[
    \left|\frac{h(\varepsilon)}
    {\widetilde h(\varepsilon)}\right|
    \longrightarrow 0
    \qquad\text{or}\qquad
    \left|\frac{h(\varepsilon)}
    {\widetilde h(\varepsilon)}\right|
    \longrightarrow \infty
    \qquad\text{as }\varepsilon\to+\infty.
\]

    \item\label{as:E} \textbf{Exponential--polynomial functions.}
    Every $g_k$ is an exponential polynomial,
    \[
        g_k(u)=\sum_i p_i(u)e^{\mu_i u},
    \]
    with $p_i\in\mathbb{C}[u]$, where $\mathbb{C}[u]=\{\text{polynomials in }u\text{ with coefficients in }\mathbb{C}\}$, and distinct $\mu_i\in\mathbb{R}$.
\end{enumerate}
\end{assumption}

\ref{as:A1}--\ref{as:Supp} are structural and regularity assumptions. Assumption \ref{as:A1} excludes polynomial profiles $g_k$ which are known to be unidentifiable without further assumptions \cite{pinkus2013ridge}. Assumption \ref{as:A2} requires each warp strictly monotone, median-zero, and unbounded. Median-zero is a normalization choice. Assumption \ref{as:AR} is standard for identifiability results of this kind. Assumption \ref{as:Supp} requires $|\Btr_k|\ge2$, since a single-coordinate interaction is a main effect, and $\Btr_k\neq\Btr_l$. 

Assumption \ref{as:C} rules out parameterizations containing two proportional warps. In particular, two median-zero affine warps $h_k(\eps)=a_k\eps$ are 
proportional and hence
cannot occur together under \ref{as:C}. The affine-warp regime is treated separately in Section \ifusepackageA~\ref{ss:affine-id}, \fi where the profiles rather than distinct warp growth rates facilitate identification.
Assumption \ref{as:G} is a model-class condition, and in a comparison of two parameterizations it therefore applies to their pooled warps. We note that one can relax the assumption to pairs of warps which belong to two parameterizations from distributionally equivalent models. The identifiability rests on the noise $\varepsilon$ entering distinct interactions differently. \ifusepackageA Remark~\ref{res:fail} \fi gives an unidentifiable example satisfying \ref{as:E} but not \ref{as:G}, so \ref{as:E} alone is insufficient. \ifusepackageA Remark~\ref{rem:polyG-A6-fails} \fi discusses a counterexample for this setting when functions $g_k$ are polynomial.  

We state the main identifiability result in a general setting.

\begin{theorem}[{Identifiability}]\label{THM:identify_main} 
Suppose that the model class $\mathcal{M}$ satisfies Assumptions~\ref{assumptions}. Then, the model is identifiable in the sense of Definition \ref{DEF_model_is_identifiable}. 
\end{theorem}

The proof of Theorem~\ref{THM:identify_main} is provided in \ifusepackageA Appendix~\ref{proofs}. \fi 
Pinkus \cite{pinkus2013ridge} studies a similar identification problem under the more permissive assumption that $x$ ranges over a set with nonempty interior, so that mixed partial derivatives of arbitrary order are directly available for recovering the ridge directions $\beta_k$.
When data are observed only on $\mathcal{X}_{\mathrm{tr}}$, however, it is not immediately clear whether the ridge directions can still be disentangled.

Some form of the nonlinearity condition~\ref{as:G} is necessary.
If the functions $h_k$ are affine, two models with distinct supports may induce
the same conditional law on $\mathcal{X}_{\mathrm{tr}}$, as discussed in Appendix \ifusepackageA ~\ref{ss:affine-id}\fi.
Theorem \ref{THM:identify_main} analyses the exponential--polynomial profiles. However, identifiability results are not confined to this case; \ifusepackageA Appendix \ref{proofs} \fi shows that they hold for a broader class of interaction profiles. 
Allowing the interaction functions $g_k$ to be polynomial (but nonlinear) requires stronger assumptions, even if the full $\mathcal{X}_{\text{test}}$ were observed, and is related to uniqueness questions for Waring decompositions of symmetric tensors \citep{comon2008symmetric}. An example of such a (strong) assumption is discussed in \ifusepackageA Section \ref{sec:polynomial-identification} \fi, and a counterexample for our setting is given in Remark \ifusepackageA \ref{rem:polyG-A6-fails}. \fi

Identifiability results form the basis for combinatorial extrapolation, as described next.
\begin{theorem}[Extrapolability]\label{thm:extrapolation}
Consider the model \eqref{adeng_truemodel} with corresponding true $T(x,\varepsilon) {\in {\cal M}}$, where $x \in {\cal X}_{\mathrm{test}} \supset {\cal X}_{\mathrm{tr}}$, see Equations \eqref{defn:train_supports}, \eqref{defn:test_support} and $\varepsilon \in \mathbb{R}$.
    Consider $\tilde{T} \in \mathcal{M}$ such that $\forall x\in \mathcal{X}_{\mathrm{tr}}, \enspace T(x, \varepsilon) \overset{d}{=} \tilde{T}(x, \varepsilon)$.

    Such a $\tilde{T}$ is obtained, for example, by a population minimizer of \eqref{obj_fct} under the conditions of Lemma \ref{lemma_engression_prop1} and Corollary \ref{cor:ae2all} (whose continuity conditions are automatic in $\cal M$ by Remark \ref{rem:continuity-laws}).
    Furthermore, suppose that the model ${\cal M}$ is identifiable in the sense of Definition \ref{DEF_model_is_identifiable}. 
    Then, \begin{enumerate}
        \item for all $x \in {\cal X}_{\mathrm{tr}} \cup{\cal X}_{\mathrm{test}}$ and for all fixed real $\varepsilon$, 
        \begin{eqnarray*}
            T(x,\varepsilon) = \tilde{T}(x, \varepsilon)\ \ \mbox{(deterministic equality)}.
            \end{eqnarray*}
        \item Moreover, when extrapolating further than ${\cal X}_{\mathrm{test}}$, for $x\in \mathbb{R}^p$ and for all fixed real $\varepsilon$, 
        \ifusepackageA
            \begin{eqnarray}
                & &G(x,\varepsilon) = \tilde{G}(x,\varepsilon) + \sum_{k=1}^K \gamma_k\ \ (\mbox{deterministic equality}) \label{eq:interaction} %
            \end{eqnarray}
        \fi
        where $G(.,.)$ and $\tilde{G}(.,.)$ are the interaction part of $T(.,.)$ and $\tilde{T}(.,.)$, respectively
    \end{enumerate} 
\end{theorem}


We collect several concise clarifications and consequences of Theorem~\ref{thm:extrapolation}.

\emph{Identification and the role of unbounded pre-additive noise.}  
    The assumption that each $h_k(\varepsilon)$ is unbounded above and below (as $\varepsilon$ ranges over $\mathbb{R}$) ensures that the argument $\beta_k^\top x + h_k(\varepsilon)$ of each interaction function spans an effectively unbounded range as $x$ varies along coordinate axes and $\varepsilon$ varies. Combined with the identifiability condition (Definition~\ref{DEF_model_is_identifiable}) and the monotonicity of $G(x,\cdot)$, this yields~\eqref{eq:interaction}.

\emph{Interpretation of deterministic equality.}
The deterministic equality in statement~1 has a natural interpretation in generative modeling. When sampling independent noise variables $\varepsilon,\varepsilon' \sim \mathcal{N}(0,1)$, it implies that for every $x \in \mathcal{X}_{\mathrm{test}}$,
\[T(x,\varepsilon) \;\overset{d}{=}\; \tilde T(x,\varepsilon').
\]
Thus, the two models generate identically distributed outputs.

\emph{Placement of the axes.}  
The one-dimensional axes $\mathcal{X}_j$ (for the $j$th component) used for training need not lie at the boundary of the joint support. The result continues to hold as long as each axis provides marginal coverage of the corresponding coordinate, even if it lies in the interior of the support. The key requirement is that the training data span each coordinate marginally, enabling identification of both the univariate components and the interaction functions along these axes.

\section{Empirical results across application domains}\label{sec:empiricalvalidation}

We present empirical results on both synthetic and real-world datasets to evaluate the performance of \algname. The section is organized as follows; Section~\ref{subsec:baselines_architectures} describes the baseline models used in our experiments, Sections~\ref{subsec:almanac} and~\ref{sec:hpobench} report results on real datasets exhibiting axis-aligned structure, Section~\ref{subsec:almostaxisaligned} extends the analysis to datasets that are not strictly axis-aligned but instead concentrate around coordinate axes. Experiments on synthetic datasets with known data-generating mechanisms are deferred to \ifusepackageA Appendix~\ref{app:synthetic_generation}. \fi

\subsection{Baseline methods and architectures}\label{subsec:baselines_architectures}

We compare \textbf{\algname} against three baseline approaches:
\begin{enumerate}
    \item \textbf{Neural Additive Model (NAM):} an additive model that learns separate neural networks for individual features and combines their outputs additively \cite{agarwal2021neural};
    \item \textbf{Empirical Risk Minimization (ERM):} a standard feedforward multilayer perceptron trained using mean squared error loss;
    \item \textbf{Engression}~\cite{shen2024engression}: a generative regression method that models the conditional distribution using an ensemble of two-layer perceptrons.
\end{enumerate}

For a fair comparison, the architectures of ERM and Engression are designed such that their number of parameters approximately matches that of \algname. Full architectural specifications and hyperparameter settings are provided in \ifusepackageA Appendix~\ref{sec:models_hyperparameters}. \fi 

\paragraph{Low sample size and optimizing the energy score.}
NAM and ERM are trained by optimizing the Mean Squared Error (MSE), while \algname\ and Engression are trained by minimizing the negative energy score. Empirically, this objective appears surprisingly robust in moderate- and low-sample-size settings. A potential reason is that the energy score matches distributions through pairwise distances, avoiding direct optimization of pointwise prediction or likelihood-based losses. In contrast, neural networks trained via empirical risk minimization with standard squared-error or logistic losses may be more prone to overfitting when data are scarce. The results in Table~\ref{tab:hpobench} are consistent with this phenomenon: even with only 39 training observations, the neural-network-based methods trained using the energy score exhibit strong performance.

\subsection{Axis-aligned real data from drug combination screens}\label{subsec:almanac}

\paragraph{O'Neil data}
We evaluate our method on the Merck OncoPolyPharmacology Screen \cite{twarog2021data}, 
a large-scale drug combination screen measuring the viability of $39$ human cancer cell 
lines under $38$ anticancer compounds, spanning $583$ pairwise drug combinations across 
a $4 \times 4$ concentration grid. Single-agent measurements are available for each 
compound at approximately $8$ concentrations per cell line, yielding $11{,}856$ 
single-drug observations, while the full combination screen produces $358{,}560$ 
dose-response measurements. Each observation reports cell viability relative to an 
untreated control, where a value of $1$ indicates no growth effect and values below $1$ 
indicate cytotoxicity. 

\paragraph{ALMANAC data}
We consider an experimental setup based on a subset of the ALMANAC collection from DrugCombDB \citep{liu2020drugcombdb} to evaluate combinatorial extrapolation in corresponding drug combination screens.
ALMANAC comprises high-throughput measurements of cell viability under drug perturbations across multiple values of concentration. 
The response variable is percent cell viability normalized to an untreated control (100 = no inhibition; $<100$ = reduced viability; $>100$ = overgrowth). 

For both datasets, to assess combinatorial extrapolation, we adopt an axis-aligned data split. The training set consists of axis-aligned experiments, i.e., single-drug concentration sweeps in which exactly one drug is perturbed at a time, yielding clean marginal response curves. For O'Neil, this yields $N = 11{,}856$ single-agent observations across 38 drugs; for ALMANAC, $N = 1{,}363{,}939$ single-agent observations across 101 drugs. The test set comprises true combinatorial experiments with multiple simultaneously nonzero drug concentrations, thereby probing interaction effects; for O'Neil the test set contains $M = 358{,}560$ combination observations, and for ALMANAC $M = 2{,}045{,}907$. Missing responses are excluded and targets are normalized by the training standard deviation.
\ifusepackageA
    \begin{table}[t]
    \centering
    \caption{Test \textsc{mse} (mean $\pm$ std over 10 seeds) on drug combination screens.
             Bold indicates the method is within
             $2\,\hat\sigma/\sqrt{10}$ of the best mean.
             LM is a single regularized linear regression (no std).  }
    \label{tab:biodata}
    \begin{tabular}{lccccc}
    \toprule
    \textbf{Dataset} & \textbf{DExtrI} & \textbf{NAM} & \textbf{Engression} & \textbf{ERM} & \textbf{LM} \\
    \midrule
    O'Neil         & $\mathbf{1.045} \pm 0.042$ & $1.675 \pm 0.101$ & $1.477 \pm 0.059$ & $2.519 \pm 0.031$ & 1.7592 \\
    ALMANAC        & $ 0.820 \pm 0.031 $ & $\mathbf{0.731} \pm 0.002$ & $ 0.863 \pm 0.015$ & $ 0.807 \pm 0.013 $ & $0.8691$ \\
    \bottomrule
    \end{tabular}%
    \end{table}
\fi
\subsection{MLP hyperparameter optimization on tabular benchmarks}
\label{sec:hpobench}

We further evaluate combinatorial extrapolation on the 
\textsc{mlp}
benchmark family from \textsc{HPOBench} for tabular data \citep{eggensperger2021hpobench}, a
standardized collection of hyperparameter optimization problems with
pre-computed, exhaustively evaluated configuration grids.
The benchmark covers eight binary classification tasks drawn from the
\textsc{OpenML} AutoML benchmark \citep{vanschoren2013openml}:
\texttt{australian}, \texttt{car}, \texttt{phoneme}, \texttt{vehicle},
\texttt{kc1}, \texttt{segment}, \texttt{blood\_transfusion}, and
\texttt{credit\_g}.
Each task defines a five-dimensional hyperparameter space; all five
hyperparameters are continuous or integer-valued and are fully described
in \ifusepackageA Appendix~\ref{app:hpobench}\fi.
The response variable is the validation error rate (one minus validation
accuracy), so that lower values indicate better-performing configurations.

Each hyperparameter is varied individually across its full grid while the
remaining four are held fixed at their respective midpoint values, yielding
$39$ training observations per benchmark problem. 
The test set consists of joint configurations, in which all five
hyperparameters vary simultaneously; by construction, every test point lies
outside the convex hull of the training support (see \ifusepackageA Appendix~\ref{app:hpobench} \fi for details).

Results are reported in Table~\ref{tab:hpobench}.
\textsc{DExtrI} achieves the lowest test \textsc{mse} on seven of the eight
tasks, with particularly pronounced improvements over the baselines on
\texttt{phoneme}, \texttt{kc1}, and \texttt{credit\_g}.
The single exception is \texttt{blood\_transfusion}, where Engression
performs best; we attribute this to the unusually flat response surface of
that task, which reduces the advantage of modelling interaction effects. We also considered standard linear models, which are often regarded as appropriate in low-sample-size settings. Their performance was roughly comparable to that of ERM.

\ifusepackageA
    \begin{table}[t]
    \centering
    \caption{Test \textsc{mse} (mean $\pm$ std over 10 seeds) on the
             \textsc{HPOBench} tabular \textsc{mlp} benchmark family.
             Bold indicates the method is within
             $2\,\hat\sigma/\sqrt{10}$ of the best mean.
             LM is a single regularized linear regression (no std).  }
    \label{tab:hpobench}
    \begin{tabular}{lccccc}
    \toprule
    \textbf{Dataset} & \textbf{DExtrI} & \textbf{NAM} & \textbf{Engression} & \textbf{ERM} & \textbf{LM} \\
    \midrule
    australian         & $\mathbf{3.4795} \pm 0.1728$ & $3.5753 \pm 0.0018$ & $5.3283 \pm 0.5796$ & $8.5440 \pm 1.1389$ & $7.8936$ \\
    car                & $\mathbf{1.2641} \pm 0.1026$ & $1.3539 \pm 0.0067$ & $2.3852 \pm 0.4472$ & $4.0730 \pm 0.3778$ & $5.0950$ \\
    phoneme            & $\mathbf{1.5743} \pm 0.1404$ & $1.9858 \pm 0.0161$ & $1.9753 \pm 0.5779$ & $3.4248 \pm 0.5403$ & $3.1109$ \\
    vehicle            & $\mathbf{1.2634} \pm 0.0621$ & $1.4437 \pm 0.0072$ & $1.6756 \pm 0.4406$ & $3.0933 \pm 0.5342$ & $2.9708$ \\
    kc1                & $\mathbf{1.4146} \pm 0.1075$ & $1.8258 \pm 0.0112$ & $1.9640 \pm 0.2982$ & $2.9794 \pm 0.4868$ & $3.1330$ \\
    segment            & $\mathbf{1.3022} \pm 0.0560$ & $1.3697 \pm 0.0041$ & $\mathbf{1.3038} \pm 0.4503$ & $2.6107 \pm 0.0590$ & $2.0371$ \\
    blood\_transfusion & $3.3805 \pm 0.3141$ & $3.1196 \pm 0.0182$ & $\mathbf{1.8262} \pm 0.3681$ & $2.6514 \pm 0.2144$ & $2.2050$ \\
    credit\_g          & $\mathbf{1.6573} \pm 0.1264$ & $2.2589 \pm 0.0093$ & $2.1006 \pm 0.3471$ & $4.0761 \pm 0.3632$ & $3.9223$ \\
    \bottomrule
    \end{tabular}
    \end{table}
\fi

\subsection{Approximately axis-aligned real data}\label{subsec:almostaxisaligned}

We consider settings in which covariates are not strictly axis-aligned, but lie close to coordinate axes. That is, the covariates are not required to have exactly nonzero components provided that at most one deviates substantially from a baseline.
Across all datasets, we construct a train--test split reflecting this structure. Covariates are first centered and standardized relative to a baseline (e.g., mean or median). The training set consists of \emph{approximately axis-aligned} samples, defined as those for which all but at most one coordinate lie within a small tolerance of not being axis-aligned. The remaining samples, which exhibit simultaneous deviations in multiple coordinates, form the test set and capture combinatorial effects. The response is scaled using its training-set variance.

\paragraph{Seoul bike sharing demand}
The Seoul Bike Sharing dataset \cite{seoul_bike_sharing_demand_560} contains hourly bike rental counts together with weather and temporal covariates. We use hour of day, temperature, humidity, wind speed, visibility, dew point temperature, and solar radiation as covariates, and log-transform the response (rental count).

\paragraph{Combined cycle power plant}
The Combined Cycle Power Plant dataset \cite{combined_cycle_power_plant_294} contains 9{,}568 observations of plant operation. The task is to predict net hourly electrical energy output from temperature, ambient pressure, relative humidity, and exhaust vacuum. 

\paragraph{Airfoil self-noise}
The Airfoil self-noise dataset \cite{airfoil_self-noise_291} consists of approximately 1{,}500 observations from wind tunnel experiments. Covariates include frequency, angle of attack, chord length, free-stream velocity, and displacement thickness. The response is sound pressure level.

\paragraph{Gas turbine}
The gas turbine dataset \cite{gas_turbine_co_and_nox_emission_data_set_551} contains approximately 36{,}000 observations of turbine operation with multiple sensor measurements. We use all seven covariates and take turbine energy yield (TEY) as the response.

\paragraph{Abalone}
The Abalone dataset \cite{abalone_1} contains physical measurements used to predict age (number of rings) from six covariates. We discard the categorical sex covariable and log-transform the response. All remaining preprocessing follows the procedure described above.

\paragraph{California housing}
The California Housing dataset \cite{pace1997californiahousing} contains eight census-based covariates for predicting median house value; the categorical ocean proximity variable is excluded. 
Observations with missing values are removed and the remainings sample size is around 20k samples. The response is log-transformed, and skewed count-based covariates are log-transformed. 

\begin{table*}[ht]
  \centering
  \caption{Test mean squared error for approximately axis-aligned real data.
  Values are mean $\pm$ std across 10 random seeds.
  The percentage of training samples being approximately axis-aligned
  is denoted in the subcolumns.
  Bold indicates the method is within $2\,\hat\sigma/\sqrt{10}$ of the best mean per column.}
  \label{tab:almost_axis_aligned_mse}
  \resizebox{\textwidth}{!}{%
  \begin{tabular}{l|ccc|ccc|ccc}
    \toprule
    \multicolumn{1}{c|}{\textbf{Dataset}} &
      \multicolumn{3}{c|}{\textbf{Seoul Bike Sharing Demand} \cite{seoul_bike_sharing_demand_560}} &
      \multicolumn{3}{c|}{\textbf{Combined Cycle Power Plant} \cite{combined_cycle_power_plant_294}} &
      \multicolumn{3}{c}{\textbf{Airfoil Self-Noise} \cite{airfoil_self-noise_291}} \\
    \textbf{\% used for train}
      & \textbf{10\%} & \textbf{16\%} & \textbf{25\%}
      & \textbf{10\%} & \textbf{20\%} & \textbf{25\%}
      & \textbf{10\%} & \textbf{16\%} & \textbf{24\%} \\
    \midrule
    ERM
      & $1.76 \pm 0.36$ & $1.12 \pm 0.20$ & $1.13 \pm 0.29$ & $21.3 \pm 11.8$ & $12.9 \pm 6.5$ & $4.93 \pm 1.81$ & $4.54 \pm 2.45$ & $\mathbf{3.98} \pm 1.86$ & $2.51 \pm 1.28$ \\
    NAM
      & $1.19 \pm 0.45$ & $0.93 \pm 0.12$ & $0.52 \pm 0.05$ & $70.2 \pm 28.6$ & $25.1 \pm 8.2$ & $4.37 \pm 1.92$ & $15.0 \pm 5.0$ & $16.6 \pm 6.1$ & $11.4 \pm 5.2$ \\
    Engression
      & $0.70 \pm 0.03$ & $0.68 \pm 0.03$ & $0.61 \pm 0.02$ & $13.7 \pm 3.5$ & $6.61 \pm 1.60$ & $2.35 \pm 0.39$ & $12.9 \pm 2.9$ & $13.8 \pm 4.0$ & $6.32 \pm 1.87$ \\
    \algname\
      & $\mathbf{0.60} \pm 0.06$ & $\mathbf{0.53} \pm 0.03$ & $\mathbf{0.44} \pm 0.01$ & $\mathbf{5.05} \pm 1.78$ & $\mathbf{1.28} \pm 0.49$ & $\mathbf{0.52} \pm 0.37$ & $\mathbf{2.99} \pm 1.16$ & $\mathbf{3.06} \pm 0.96$ & $\mathbf{1.35} \pm 0.28$ \\
    \midrule
    \multicolumn{1}{c|}{\textbf{Dataset}} &
      \multicolumn{3}{c|}{\textbf{Gas Turbine CO and NOx Emission} \cite{gas_turbine_co_and_nox_emission_data_set_551}} &
      \multicolumn{3}{c|}{\textbf{California Housing} \cite{pace1997californiahousing}} &
      \multicolumn{3}{c}{\textbf{Abalone} \cite{abalone_1}} \\
    \textbf{\% used for train}
      & \textbf{10\%} & \textbf{15\%} & \textbf{25\%}
      & \textbf{5\%}  & \textbf{11\%} & \textbf{20\%}
      & \textbf{12\%} & \textbf{16\%} & \textbf{26\%} \\
    \midrule
    ERM
      & $7.48 \pm 2.99$ & $3.27 \pm 1.74$ & $\mathbf{0.65} \pm 0.53$ & $\mathbf{5.92} \pm 2.98$ & $4.14 \pm 1.45$ & $2.36 \pm 0.76$ & $\mathbf{3.26} \pm 1.75$ & $3.06 \pm 1.27$ & $4.62 \pm 1.81$ \\
    NAM
      & $47.0 \pm 20.3$ & $15.1 \pm 3.5$ & $4.08 \pm 1.47$ & $8.81 \pm 2.18$ & $9.01 \pm 2.26$ & $2.73 \pm 0.66$ & $10.4 \pm 2.0$ & $4.97 \pm 1.06$ & $1.76 \pm 0.35$ \\
    Engression
      & $36.0 \pm 5.5$ & $14.6 \pm 2.2$ & $4.31 \pm 0.84$ & $6.82 \pm 2.08$ & $2.72 \pm 0.51$ & $1.80 \pm 0.32$ & $6.41 \pm 1.84$ & $3.77 \pm 0.79$ & $4.66 \pm 0.97$ \\
    \algname\
      & $\mathbf{1.28} \pm 0.56$ & $\mathbf{0.89} \pm 0.33$ & $\mathbf{0.47} \pm 0.16$ & $\mathbf{5.28} \pm 0.89$ & $\mathbf{2.26} \pm 0.48$ & $\mathbf{1.23} \pm 0.19$ & $\mathbf{2.64} \pm 0.64$ & $\mathbf{1.88} \pm 0.65$ & $\mathbf{1.41} \pm 0.37$ \\
    \bottomrule
  \end{tabular}%
  }
\end{table*}

\section{Discussion}\label{sec:discussion}

We introduced DExtrI, a framework for combinatorial extrapolation from (almost) axis-aligned data. By combining distributional regression with a projection-pursuit-type interaction structure, DExtrI can recover and extrapolate interaction effects that are not directly observed during training. 
We established population identifiability and extrapolation guarantees under structural model assumptions,
and demonstrated good
empirical performance across drug combination studies, hyperparameter optimization tasks, and approximately axis-aligned real-world datasets. 

A key implication is that one-factor-at-a-time experiments may contain substantially more information about interactions than is commonly assumed. This is particularly relevant in applications such as drug discovery, perturbation genomics, and hyperparameter optimization, where exhaustive combinatorial experimentation is often infeasible.

More broadly, our results suggest that combinatorial extrapolation can be treated as an identifiable statistical problem rather than solely as a prediction challenge.
Several directions merit further investigation. First, finite-sample convergence rates would provide guidance for experimental design. Second, it would be valuable to understand the robustness of the method when the model assumptions are only approximately satisfied.

\section*{Acknowledgments}
 M. \v{S}ola was supported by the Swiss National Science Foundation, grant no. 214865. 
We thank Markus Ulmer for insightful discussions on experimental design and dataset choice, and for thoughtful feedback that helped make the manuscript clearer and more accessible. The authors used ChatGPT (OpenAI) and Claude (Anthropic) during manuscript preparation for language editing, proofreading, and critical feedback on mathematical arguments and exposition. The authors reviewed and verified all resulting suggestions and take full responsibility for the content of the manuscript.

\bibliography{library.bib}

\begin{thebibliography}{10}

\bibitem{gas_turbine_co_and_nox_emission_data_set_551}
{Gas Turbine CO and NOx Emission Data Set}.
\newblock UCI Machine Learning Repository, 2019.
\newblock {DOI}: https://doi.org/10.24432/C5WC95.

\bibitem{seoul_bike_sharing_demand_560}
{Seoul Bike Sharing Demand}.
\newblock UCI Machine Learning Repository, 2020.
\newblock {DOI}: https://doi.org/10.24432/C5F62R.

\bibitem{aczel1966}
J.~Acz{\'e}l.
\newblock {\em Lectures on Functional Equations and Their Applications}.
\newblock Academic Press, New York, 1966.

\bibitem{agarwal2021neural}
Rishabh Agarwal, Levi Melnick, Nicholas Frosst, Xuezhou Zhang, Ben Lengerich, Rich Caruana, and Geoffrey~E Hinton.
\newblock Neural additive models: Interpretable machine learning with neural nets.
\newblock {\em Advances in Neural Information Processing Systems}, 34:4699--4711, 2021.

\bibitem{arjovsky2019invariant}
Martin Arjovsky, L{\'e}on Bottou, Ishaan Gulrajani, and David Lopez-Paz.
\newblock Invariant risk minimization.
\newblock {\em arXiv preprint arXiv:1907.02893}, 2019.

\bibitem{balestriero2021learning}
Randall Balestriero, Jerome Pesenti, and Yann LeCun.
\newblock Learning in high dimension always amounts to extrapolation.
\newblock {\em arXiv preprint arXiv:2110.09485}, 2021.

\bibitem{airfoil_self-noise_291}
Thomas Brooks, D.~Pope, and Michael Marcolini.
\newblock {Airfoil Self-Noise}.
\newblock UCI Machine Learning Repository, 1989.
\newblock {DOI}: https://doi.org/10.24432/C5VW2C.

\bibitem{chatterjee2022generalization}
Satrajit Chatterjee and Piotr Zielinski.
\newblock On the generalization mystery in deep learning.
\newblock {\em arXiv preprint arXiv:2203.10036}, 2022.

\bibitem{comon2008symmetric}
Pierre Comon, Gene Golub, Lek-Heng Lim, and Bernard Mourrain.
\newblock Symmetric tensors and symmetric tensor rank.
\newblock {\em SIAM Journal on Matrix Analysis and Applications}, 30(3):1254--1279, 2008.

\bibitem{dudley2002real}
R.~M. Dudley.
\newblock {\em Real Analysis and Probability}.
\newblock Cambridge University Press, 2002.

\bibitem{eggensperger2021hpobench}
Katharina Eggensperger, Philipp M{\"u}ller, Neeratyoy Mallik, Matthias Feurer, Ren{\'e} Sass, Aaron Klein, Noor Awad, Marius Lindauer, and Frank Hutter.
\newblock {HPOB}ench: A collection of reproducible multi-fidelity benchmark problems for {HPO}.
\newblock In Joaquin Vanschoren and Sai-Kit Yeung, editors, {\em Proceedings of the Neural Information Processing Systems Track on Datasets and Benchmarks 1 (NeurIPS Datasets and Benchmarks 2021)}, 2021.

\bibitem{freschlin2024neural}
Chase~R Freschlin, Sarah~A Fahlberg, Pete Heinzelman, and Philip~A Romero.
\newblock Neural network extrapolation to distant regions of the protein fitness landscape.
\newblock {\em Nature Communications}, 15(1):6405, 2024.

\bibitem{Friedman1981projectionpursuit}
Jerome~H. Friedman and Werner Stuetzle.
\newblock Projection pursuit regression.
\newblock {\em Journal of the American Statistical Association}, 76(376):817--823, 1981.

\bibitem{fu2024general}
Jingwen Fu, Zhizheng Zhang, Yan Lu, and Nanning Zheng.
\newblock A general theory for compositional generalization.
\newblock {\em arXiv preprint arXiv:2405.11743}, 2024.

\bibitem{gneiting2007strictly}
Tilmann Gneiting and Adrian~E. Raftery.
\newblock Strictly proper scoring rules, prediction, and estimation.
\newblock {\em Journal of the American Statistical Association}, 102(477):359--378, 2007.

\bibitem{hall1989onprojectionpursuit}
Peter Hall.
\newblock On projection pursuit regression.
\newblock {\em The Annals of Statistics}, 17(2):573--588, 1989.

\bibitem{huber1985projection}
Peter~J. Huber.
\newblock {Projection Pursuit}.
\newblock {\em The Annals of Statistics}, 13(2):435 – 475, 1985.

\bibitem{ICHIMURA199371}
Hidehiko Ichimura.
\newblock Semiparametric least squares (sls) and weighted sls estimation of single-index models.
\newblock {\em Journal of Econometrics}, 58(1):71--120, 1993.

\bibitem{kawaguchi2017generalization}
Kenji Kawaguchi, Yoshua Bengio, and Leslie~Pack Kaelbling.
\newblock Generalization in deep learning.
\newblock In Philipp Grohs and Gitta Kutyniok, editors, {\em Mathematical Aspects of Deep Learning}, pages 112--148. Cambridge University Press, Cambridge, 2022.

\bibitem{kingma2014adam}
Diederik~P Kingma and Jimmy Ba.
\newblock Adam: A method for stochastic optimization.
\newblock {\em International Conference on Learning Representations}, 2015.

\bibitem{krueger2021out}
David Krueger, Ethan Caballero, Joern-Henrik Jacobsen, Amy Zhang, Jonathan Binas, Dinghuai Zhang, Remi Le~Priol, and Aaron Courville.
\newblock Out-of-distribution generalization via risk extrapolation (rex).
\newblock In {\em International Conference on Machine Learning}, pages 5815--5826. PMLR, 2021.

\bibitem{lippl2024does}
Samuel Lippl and Kim Stachenfeld.
\newblock When does compositional structure yield compositional generalization? a kernel theory.
\newblock {\em International Conference on Learning Representations}, 2025.

\bibitem{liu2020drugcombdb}
Hui Liu, Wenhao Zhang, Bo~Zou, Jinxian Wang, Yuanyuan Deng, and Lei Deng.
\newblock Drugcombdb: a comprehensive database of drug combinations toward the discovery of combinatorial therapy.
\newblock {\em Nucleic acids research}, 48(D1):D871--D881, 2020.

\bibitem{ma2026deeplearningmissingdata}
Tianyi Ma, Tengyao Wang, and Richard~J Samworth.
\newblock Deep learning with missing data.
\newblock {\em Journal of the Royal Statistical Society Series B: Statistical Methodology}, page qkag114, 07 2026.

\bibitem{nagarajan2020understanding}
Vaishnavh Nagarajan, Anders Andreassen, and Behnam Neyshabur.
\newblock Understanding the failure modes of out-of-distribution generalization.
\newblock {\em International Conference on Learning Representations}, 2021.

\bibitem{abalone_1}
Warwick Nash, Tracy Sellers, Simon Talbot, Andrew Cawthorn, and Wes Ford.
\newblock {Abalone}.
\newblock UCI Machine Learning Repository, 1994.
\newblock {DOI}: https://doi.org/10.24432/C55C7W.

\bibitem{netanyahu2023learning}
Aviv Netanyahu, Abhishek Gupta, Max Simchowitz, Kaiqing Zhang, and Pulkit Agrawal.
\newblock Learning to extrapolate: A transductive approach.
\newblock {\em International Conference on Learning Representations}, 2023.

\bibitem{pace1997californiahousing}
Kelley Pace and Ronald Barry.
\newblock Sparse spatial autoregressions.
\newblock {\em Statistics \& Probability Letters}, 33(3):291--297, 1997.

\bibitem{pfister2024extrapolation}
Niklas Pfister and Peter Bühlmann.
\newblock Extrapolation-aware nonparametric statistical inference.
\newblock {\em arXiv preprint arXiv:2402.09758}, 2024.

\bibitem{pinkus2013ridge}
A.~Pinkus.
\newblock Smoothness and uniqueness in ridge function representation.
\newblock {\em Indagationes Mathematicae}, 24(4):725–738, 2013.
\newblock In memory of N.G. (Dick) de Bruijn (1918–2012).

\bibitem{rothenhausler2021anchor}
Dominik Rothenh{\"a}usler, Nicolai Meinshausen, Peter B{\"u}hlmann, and Jonas Peters.
\newblock Anchor regression: Heterogeneous data meet causality.
\newblock {\em Journal of the Royal Statistical Society Series B: Statistical Methodology}, 83(2):215--246, 2021.

\bibitem{shen2024distributional}
Xinwei Shen and Nicolai Meinshausen.
\newblock Distributional principal autoencoders.
\newblock {\em arXiv preprint arXiv:2404.13649}, 2024.

\bibitem{shen2024engression}
Xinwei Shen and Nicolai Meinshausen.
\newblock Engression: extrapolation through the lens of distributional regression.
\newblock {\em Journal of the Royal Statistical Society Series B: Statistical Methodology}, 87(3):653--677, 07 2025.

\bibitem{simchowitz2023combinatorial}
Max Simchowitz, Abhishek Gupta, and Kaiqing Zhang.
\newblock Tackling combinatorial distribution shift: A matrix completion perspective.
\newblock In {\em Proceedings of Thirty Sixth Conference on Learning Theory}, volume 195 of {\em Proceedings of Machine Learning Research}, pages 3356--3468. PMLR, 12--15 Jul 2023.

\bibitem{szekely2003estatistics}
G{\'a}bor~J. Sz{\'e}kely.
\newblock {E}-statistics: {T}he {E}nergy of {S}tatistical {S}amples.
\newblock Technical Report 02-16, Department of Mathematics and Statistics, Bowling Green State University, 2003.

\bibitem{szekelyhidi1991}
L.~Sz{\'e}kelyhidi.
\newblock {\em Convolution Type Functional Equations on Topological Abelian Groups}.
\newblock World Scientific, Singapore, 1991.

\bibitem{SzekelyRizzo2023}
Gábor~J. Székely and Maria~L. Rizzo.
\newblock {\em The Energy of Data and Distance Correlation}.
\newblock Chapman and Hall/CRC, 1st edition, 2023.

\bibitem{combined_cycle_power_plant_294}
Pnar Tfekci and Heysem Kaya.
\newblock {Combined Cycle Power Plant}.
\newblock UCI Machine Learning Repository, 2014.
\newblock {DOI}: https://doi.org/10.24432/C5002N.

\bibitem{twarog2021data}
Nathaniel~R Twarog, Nancy~E Martinez, Jessica Gartrell, Jia Xie, Christopher~L Tinkle, and Anang~A Shelat.
\newblock Data vignettes for the application of response surface models in drug combination analysis.
\newblock {\em Data in Brief}, 38:107400, 2021.

\bibitem{vanschoren2013openml}
Joaquin Vanschoren, Jan~N. van Rijn, Bernd Bischl, and Luis Torgo.
\newblock {OpenML}: Networked science in machine learning.
\newblock {\em ACM SIGKDD Explorations Newsletter}, 15(2):49--60, 2013.

\bibitem{xu2020neural}
Keyulu Xu, Mozhi Zhang, Jingling Li, Simon~S Du, Ken-ichi Kawarabayashi, and Stefanie Jegelka.
\newblock How neural networks extrapolate: From feedforward to graph neural networks.
\newblock {\em International Conference on Learning Representations}, 2021.

\bibitem{zhang2012identifiability}
Kun Zhang and Aapo Hyvarinen.
\newblock On the identifiability of the post-nonlinear causal model.
\newblock In {\em Proceedings of the Twenty-Fifth Conference on Uncertainty in Artificial Intelligence (UAI 2009)}, pages 647--655. AUAI Press, 2009.

\end{thebibliography}

\appendix

\section{Mathematical Derivations}\label{proofs}

Throughout this appendix, $\mathcal X_{\mathrm{tr}}$ denotes the axis-aligned
training support in~\eqref{defn:train_supports}. We give proofs of theoretical results stated in Section \ref{sec:identifiabilityExtrapolation}.  Consider two arbitrary parameterizations
\[
T(x,\varepsilon)
=
F(x)+G(x,\varepsilon)
=
\sum_{j=1}^p f_j(x_j)
+
\sum_{k=1}^K g_k\!\left(\beta_k^\top x+h_k(\varepsilon)\right)
\]
and
\[
\widetilde T(x,\varepsilon)
=
\widetilde F(x)+\widetilde G(x,\varepsilon)
=
\sum_{j=1}^p \widetilde f_j(x_j)
+
\sum_{\ell=1}^{\widetilde K}
\widetilde g_\ell\!\left(\widetilde\beta_\ell^\top x+
\widetilde h_\ell(\varepsilon)\right).
\]
Tilded quantities always refer to the second parameterization.
Write its additive and interaction parts as \[
F(x) := \sum_{j=1}^p f_j(x_j), \qquad
G(x, \varepsilon) := \sum_{k=1}^K g_k\!\big(\beta_k^\top x+h_k(\varepsilon)\big).
\]

A second parametrization of the same model is a choice of parameters $(\tf_j)_{j=1}^p$, $(\tg_l,\tb_l,\th_l)_{l=1}^{\tK}$ with additive and interaction parts $\tilde F(x)=\sum_{j}\tf_j(x_j)$ and $\tilde G(x,\varepsilon)=\sum_{l=1}^{\tK}\tg_l(\tb_l^\top x+\th_l(\varepsilon))$; throughout tilde $\tilde{(\cdot)}$ quantities refer to such a second parametrization. 
The following result enables all following identifiability results.
\begin{lemma}\label{initial_step_lemma}
Let $F,\tilde{F}:\mathbb{R}^p\rightarrow\mathbb{R}$, and let
$G,\tilde{G}:\mathbb{R}^p\times\mathbb{R}\rightarrow\mathbb{R}$ be continuous, strictly
increasing in the second variable. Fix $x\in \Xtr$, if
\begin{align*}
    F(x)+G(x,\varepsilon)\overset{d}{=}\tilde{F}(x)+\tilde{G}(x,\varepsilon),
    \qquad \varepsilon\sim \mathcal{N}(0,1),
\end{align*}
then for all $\varepsilon\in\mathbb{R}$,
\begin{align*}
    F(x)+G(x,\varepsilon)=\tilde{F}(x)+\tilde{G}(x,\varepsilon).
\end{align*}
\end{lemma}

\begin{proof}
Equality in distribution implies equality of all quantiles. For
$\varepsilon\sim \mathcal{N}(0,1)$ the $\alpha$-quantile is $\Phi^{-1}(\alpha)$, where
$\Phi$ denotes the standard normal cumulative distribution function, and
$\Phi^{-1}:(0,1)\to\mathbb{R}$ is surjective. Since $G$ and $\tilde{G}$ are
strictly increasing in their second argument, the $\alpha$-quantile of
$F(x)+G(x,\varepsilon)$ is $F(x)+G(x,\Phi^{-1}(\alpha))$, and likewise for the
tilde side; hence
\begin{align*}
    F(x)+G(x,\Phi^{-1}(\alpha))=\tilde{F}(x)+\tilde{G}(x,\Phi^{-1}(\alpha))
    \quad\text{for all }\alpha\in(0,1).
\end{align*}
Since $\Phi^{-1}$ maps $(0,1)$ onto $\mathbb{R}$, substituting
$\varepsilon=\Phi^{-1}(\alpha)$ gives the identity for all
$\varepsilon\in\mathbb{R}$.
\end{proof}

Let $\mu_{\rm tr}$ denote the marginal distribution of the training
covariates and, for $x\in\mathcal X_{\rm tr}$, let
$P_x:=P_{\rm tr}(Y\mid X=x)$ denote a regular conditional distribution
of $Y$ given $X=x$. We measure weak convergence of probability laws
using the bounded--Lipschitz metric. For probability measures $P,Q$ on
$\mathbb R$, define
\[
d_{\rm BL}(P,Q)
:=
\sup_{\|f\|_{\rm BL}\leq 1}
\left|\int f\,dP-\int f\,dQ\right|,
\qquad
\|f\|_{\rm BL}:=\|f\|_\infty+\operatorname{Lip}(f),
\]
where
$\operatorname{Lip}(f):=\sup_{u\neq v}|f(u)-f(v)|/|u-v|$.
The metric $d_{\rm BL}$ metrizes weak convergence of probability
measures on $\mathbb R$; see, e.g.,
\cite[Theorem~11.3.3]{dudley2002real}. 

\begin{corollary}\label{cor:ae2all}
Let $\mu_{\rm tr}$ denote the marginal distribution of the training
covariates, with $\operatorname{supp}(\mu_{\rm tr})=\mathcal X_{\rm tr}$,
and let $\widetilde T$ be any population minimizer from
Lemma~\ref{lemma_engression_prop1}. For $x\in\mathcal X_{\rm tr}$, let
$P_x:=P_{\rm tr}(Y\mid X=x)$. Assume that
$x\mapsto P_x$ and
$x\mapsto \operatorname{Law}(\widetilde T(x,\varepsilon))$
are continuous on $\mathcal X_{\rm tr}$ with respect to the
bounded-Lipschitz metric. Then
\[
\operatorname{Law}(\widetilde T(x,\varepsilon))=P_x
\qquad\text{for every }x\in\mathcal X_{\rm tr}.
\]
\end{corollary}

\begin{proof}[Proof of Corollary~\ref{cor:ae2all}]
By Lemma~\ref{lemma_engression_prop1},
$\operatorname{Law}(\widetilde T(x,\varepsilon))=P_x$ for
$\mu_{\rm tr}$-almost every $x\in\mathcal X_{\rm tr}$. Suppose, for
contradiction, that the equality fails at some
$x_0\in\mathcal X_{\rm tr}$. Define
$d(x):=d_{\rm BL}\!\left(P_x,
\operatorname{Law}(\widetilde T(x,\varepsilon))\right)$, where
$d_{\rm BL}$ denotes the bounded-Lipschitz metric. Then $d(x_0)>0$.
By continuity of the two law-valued maps and the triangle inequality,
$d$ is continuous on $\mathcal X_{\rm tr}$. Hence there exists $r>0$
such that $d(x)>0$ for every
$x\in U:=B(x_0,r)\cap\mathcal X_{\rm tr}$.

Since $\operatorname{supp}(\mu_{\rm tr})=\mathcal X_{\rm tr}$,
$x_0\in\operatorname{supp}(\mu_{\rm tr})$, and therefore
$\mu_{\rm tr}(B(x_0,r))>0$. Moreover,
$\mu_{\rm tr}(\mathcal X_{\rm tr})=1$, so
\[
\mu_{\rm tr}(U)
=
\mu_{\rm tr}\!\left(B(x_0,r)\cap\mathcal X_{\rm tr}\right)
=
\mu_{\rm tr}(B(x_0,r))
>0.
\]
Thus the two laws differ on a set of positive
$\mu_{\rm tr}$-measure, contradicting their equality
$\mu_{\rm tr}$-almost everywhere.
\end{proof}

\begin{remark}
\label{rem:continuity-laws}
The continuity conditions in Corollary~\ref{cor:ae2all} are automatically
satisfied in our model class when the bounded--Lipschitz metric is used.
Indeed, let $T^\star\in\mathcal M$ denote the data-generating
transformation and choose the conditional law as
$P_x:=\operatorname{Law}(T^\star(x,\varepsilon))$. Since all component
functions defining $T^\star$ are continuous, $x_n\to x$ implies
$T^\star(x_n,\varepsilon)\to T^\star(x,\varepsilon)$ for every fixed
$\varepsilon$, and hence almost surely with respect to $\varepsilon$.
Therefore
$\operatorname{Law}(T^\star(x_n,\varepsilon))
$ converges to $\operatorname{Law}(T^\star(x,\varepsilon))$. Since
$d_{\rm BL}$ metrizes weak convergence, it follows that
$d_{\rm BL}(P_{x_n},P_x)\to0$.
Likewise, any population minimizer $\widetilde T$ considered in
Lemma~\ref{lemma_engression_prop1} belongs to $\mathcal M$. Hence the
same argument gives
\[
d_{\rm BL}\!\left(
\operatorname{Law}(\widetilde T(x_n,\varepsilon)),
\operatorname{Law}(\widetilde T(x,\varepsilon))
\right)\to0.
\]
Thus both law-valued maps appearing in
Corollary~\ref{cor:ae2all} are continuous on
$\mathcal X_{\rm tr}$ in $d_{\rm BL}$. In particular, the continuity
assumptions of the corollary do not impose an additional regularity
condition beyond membership of the data-generating transformation and
the population minimizer in $\mathcal M$.
\end{remark}

\begin{definition}[Warped ridge, foliation, cofoliarity]\label{def:warped}
A \emph{warped ridge} on a strip $U=(a,A_0)\times\R$ is a map
$(t,\eps)\mapsto\Th(ct+h(\eps))$ with $c\neq0$, $h$ strictly monotone, and profile $\Th$.
Its \emph{foliation} is the family of level curves $\{ct+h(\eps)=\text{const}\}$, recorded
by the normalised warp $h/c$. Two warped ridges with parameters $(c,h)$ and $(\tilde c,\tilde
h)$ are \emph{cofoliar} if their level curves coincide as subsets of $U$.
\end{definition}

The first structural lemma identifies exactly when two warped ridges are cofoliar. 

\begin{lemma}[Foliation rigidity]\label{lem:cofoliar}
Let $h,\th$ be strictly monotone and differentiable, $c,\tilde c\neq0$. The warped ridges
with parameters $(c,h)$ and $(\tilde c,\th)$ are cofoliar on $U$ if and only if there exist
$\alpha\neq0$ and $\delta\in\R$ with
\[
  \th=\alpha h+\delta,\qquad \tilde c=\alpha c.
\]
Under the normalization $h(0)=\th(0)=0$ this forces $\delta=0$, so cofoliarity is
equivalent to $\th=\alpha h$ and $\tilde c=\alpha c$ for a single $\alpha\neq0$.
\end{lemma}

\begin{proof}
\textit{Forward direction.}
Suppose the two foliations coincide on $U$. Let $F\colon U\to\R$,
$F(t,\eps):=ct+h(\eps)$: its gradient $(c,h'(\eps))$ is nowhere
zero on $U$ since $c\neq0$. Each level set $F^{-1}(s)$ is therefore a smooth,
connected curve, and the level sets partition $U$.

\textit{Existence and strict monotonicity of $\Lambda$.}
For each $s$ in $I:=F(U)\subset\R$, cofoliarity requires
$\tilde{c}t+\tilde{h}(\eps)$ to be constant on $F^{-1}(s)$; define $\Lambda(s)$
to be that constant value. This yields a well-defined function $\Lambda\colon I\to\R$
satisfying
\begin{equation}\label{eq:cofol-Lambda}
  \tilde{c}t+\th(\eps)=\Lambda\bigl(ct+h(\eps)\bigr)\qquad\text{on }U.
\end{equation}
Strict monotonicity of $\Lambda$: if $s<s'$ are distinct values in $I$, the
corresponding level curves $F^{-1}(s)$ and $F^{-1}(s')$ are disjoint, and by
cofoliarity the level sets of $\tilde{c}t+\th(\eps)$ with values $\Lambda(s)$ and
$\Lambda(s')$ are likewise disjoint, forcing $\Lambda(s)\neq\Lambda(s')$. Since $F$
and $\tilde{c}t+\th(\eps)$ are both continuous and $U$ is connected, the induced map
$\Lambda$ on $I$ is continuous and injective, hence strictly monotone.

\textit{Differentiability of $\Lambda$ and identification of its derivative.}
Fix any $\eps_0$ in the $\eps$-domain of $U$. As $t$ varies, $u:=ct+h(\eps_0)$
ranges over an open subinterval of $I$, and \eqref{eq:cofol-Lambda} reads
\[
  \Lambda(u) = \tilde{c}\,\frac{u-h(\eps_0)}{c}+\th(\eps_0),
\]
an affine function of $u$. In particular, $\Lambda$ is differentiable on this
subinterval with $\Lambda'(u)=\tilde{c}/c$. Since $\eps_0$ is arbitrary and the
corresponding subintervals cover $I$, $\Lambda$ is differentiable on all of $I$
with constant derivative $\Lambda'\equiv\alpha:=\tilde{c}/c\neq0$, hence
$\Lambda(u)=\alpha u+\delta$ for some $\delta\in\R$.

Substituting $\Lambda(u)=\alpha u+\delta$ into \eqref{eq:cofol-Lambda} gives
\[
  \tilde{c}t+\th(\eps)=\alpha\bigl(ct+h(\eps)\bigr)+\delta\qquad\text{on }U.
\]
Comparing the $t$-coefficients yields $\tilde{c}=\alpha c$, and comparing the
$\eps$-dependent parts yields $\th(\eps)=\alpha h(\eps)+\delta$ for all $\eps$.
Evaluating the latter at $\eps=0$, where $h(0)=\th(0)=0$ by hypothesis, gives
$\delta=0$.

\textit{Converse.}
If $\th=\alpha h+\delta$ and $\tilde{c}=\alpha c$, then
$\tilde{c}t+\th(\eps)=\alpha(ct+h(\eps))+\delta$, so \eqref{eq:cofol-Lambda}
holds with $\Lambda(u)=\alpha u+\delta$ strictly monotone, and the foliations
coincide. Evaluating at $\eps=0$ gives $\delta=0$ under the normalization.
\end{proof}

The normalization $h_k(0)=0$ is implied by assumption \ref{as:A2}: under
$\eps\sim \mathcal{N}(0,1)$ the median of $\eps$ is $0$, and a strictly monotone median--zero warp
satisfies $h_k(0)=0$. We use it throughout without further comment.

\begin{corollary}[Within--model non--cofoliarity]\label{cor:within}
Suppose the warps $\{h_k\}_{k=1}^K$ are pairwise non--proportional.
Then no two distinct warped ridges within the same parameterization are cofoliar on any axis on which both are active. 
Moreover, if the warps are linearly independent, then their derivatives $\{h_k'\}_{k=1}^K$
are linearly independent.
\end{corollary}

\begin{proof}
By Lemma~\ref{lem:cofoliar}, cofoliarity of interactions $k\neq k'$ forces
$h_{k'}=\alpha h_k$, i.e.\ proportional warps, contrary to hypothesis. For the second
claim, suppose $\sum_k\gamma_k h_k'\equiv0$. Integrating from $0$ to $\eps$ and using
$h_k(0)=0$ gives $\sum_k\gamma_k h_k\equiv0$; linear independence of $\{h_k\}$ then forces
$\gamma_k=0$ for all $k$.
\end{proof}

This section studies the case when the warps $h_k$ are non-affine, but strictly monotone and real analytic. Under growth conditions on the $h_k$, the model structure is identifiable. The affine regime is discussed in Appendix \ref{ss:affine-id}. 

\subsubsection*{Notation}
Two parameterizations in $\mathcal{M}$, with parameters
$(f_j)_{j=1}^p,(g_k,\beta_k,h_k)_{k=1}^K$ and $(\tf_j)_{j=1}^p,(\tg_l,\tb_l,\th_l)_{l=1}^{\tK}$ and $K$,
$\tK$ interactions, are considered. If they induce the same conditional distribution
of $T$ on $\Xtr$, then by Lemma~\ref{initial_step_lemma} they satisfy the pointwise
identity
\begin{equation}\label{eq:full}
  \sum_j f_j(x_j)+\sum_{k=1}^K g_k(\beta_k^\top x+h_k(\eps))
  =\sum_j\tf_j(x_j)+\sum_{l=1}^{\tK}\tg_l(\tb_l^\top x+\th_l(\eps)),
  \qquad x\in\Xtr,\ \eps\in\R.
\end{equation}
Write $\Btr_k=\supp\beta_k$. Setting $x=te_j$ in \eqref{eq:full} and collecting
$t$-independent terms into $B_j(\varepsilon)$ gives, on the strip $U_j=(a_j,A_j)\times\R$,
\begin{equation}\label{eq:axis}
  \sum_{k\in S_j} g_k(\beta_{kj}t+h_k(\eps))
  -\sum_{l\in\tS_j}\tg_l(\tb_{lj}t+\th_l(\eps))=A_j(t)+B_j(\eps),
\end{equation}
where $S_j=\{k:\beta_{kj}\neq0\}$, $\tS_j=\{l:\tb_{lj}\neq0\}$, $A_j=\tf_j-f_j$
up to a constant, and $B_j$ collects all $\eps$-only terms.

Cofoliarity (Definition~\ref{def:warped}) is an equivalence relation on the
active terms appearing in \eqref{eq:axis}: each term carries a foliation
$(\beta_{kj},h_k)$ or $(\tb_{lj},\th_l)$, and two terms are cofoliar iff their
foliations coincide up to a nonzero scalar, i.e.\ iff $(\tilde c,\th)=(\alpha
c,\alpha h)$ for some $\alpha\neq0$ (Lemma~\ref{lem:cofoliar}). By
Corollary~\ref{cor:within}(i), no two true terms and no two tilde terms share a
foliation class, so each class contains at most one term from each side. Within a
matched class $\{k,l\}$, Lemma~\ref{lem:cofoliar} gives $\th_l=\alpha_{kl}h_k$ and
$\tb_{lj}=\alpha_{kl}\beta_{kj}$, so $\tb_{lj}t+\th_l(\eps)=\alpha_{kl}(\beta_{kj}t+h_k(\eps))$
and the contribution of the pair to \eqref{eq:axis} is
\[
  g_k\!\bigl(\beta_{kj}t+h_k(\eps)\bigr)
  -\tg_l\!\bigl(\alpha_{kl}(\beta_{kj}t+h_k(\eps))\bigr)
  =\bigl(g_k-\tg_l(\alpha_{kl}\,\cdot)\bigr)\!\bigl(\beta_{kj}t+h_k(\eps)\bigr),
\]
a single warped ridge with foliation $(\beta_{kj},h_k)$. Unmatched true terms
$\{k\}$ contribute $g_k(\beta_{kj}t+h_k(\eps))$ and unmatched tilde terms $\{l\}$
contribute $-\tg_l(\tb_{lj}t+\th_l(\eps))$, each again a single warped ridge.
Grouping \eqref{eq:axis} by these $N\ge1$ foliation classes and writing $c_m$,
$h_m$, $\Th_m$ for the slope, warp, and profile of the $m$-th class, the axis
identity becomes
\begin{equation}\label{eq:collapse}
  \sum_{m=1}^{N}\Th_m(c_mt+h_m(\eps))=A_j(t)+B_j(\eps),
\end{equation}
in which the foliations $(c_m,h_m)$ are pairwise non-cofoliar by construction. 
\begin{definition}[Collapse]\label{def:collapse}
A \emph{collapse} is an identity of the form \eqref{eq:collapse} with pairwise
non-cofoliar foliations in which some $\Th_m$ is not a polynomial. A term in
\eqref{eq:collapse} is \emph{unmatched} if its foliation class contains only one
side (true or tilde), so that $\Th_m=g_k$ or $\Th_m=-\tg_l$.
\end{definition}

Let $\cP$ denote the polynomials and $\cP_0$ the constants. The \emph{redistribution space} is
\begin{equation}\label{eq:redist}
  \cR=\Bigl\{(q_k)_{k=1}^K \in\cP^K:\ \exists\,(r_j)_{j=1}^p\in\cP^p,\ \textstyle\sum_{k=1}^K q_k(\beta_k^\top x+h_k(\eps))
  =\sum_{j=1}^p r_j(x_j)\ \text{on }\Xtr\times\R\Bigr\}.
\end{equation}
The model is \emph{identifiable} (Definition \ref{DEF_model_is_identifiable}) if any second parameterization in $\mathcal{M}$ inducing the same
conditional law on $\Xtr$ has, after relabelling, the same $K$, supports, directions and warps up to
a common per--term scalar, and interaction functions equal up to that scalar and an additive
constant; the per--term rescaling $(\beta_k,h_k,g_k)\mapsto(\rho\beta_k,\rho h_k,g_k(\cdot/\rho))$ and
relabelling are the unavoidable indeterminacies.

\begin{definition}[Rigidity]\label{def:rigid}
A pair of parameterizations $((f_j)_{j=1}^p, (\beta_k, h_k, g_k)_{k=1}^K), ((f_j)_{j=1}^p, (\tilde{\beta}_l, \tilde{h}_l, \tilde{g}_l)_{l=1}^{\tilde{K}})$ is \emph{rigid} if, on
every axis $j$, after grouping \eqref{eq:axis} into its pairwise non--cofoliar foliation
classes, every class, singleton or matched pair, has its profile $\Theta_m$ in $\cP$.
Equivalently, no collapse \eqref{eq:collapse} arises from \eqref{eq:axis}: every identity of
the form \eqref{eq:collapse} obtained from \eqref{eq:axis} has all of its profiles
polynomial.
\end{definition}

For completeness, we repeat here the assumptions from Section \ref{subsec:identif}.
\subsection*{Assumptions}\label{sec:assumptions}

\begin{enumerate}[label=\textnormal{(A\arabic*)},ref=\textnormal{A\arabic*},leftmargin=3.2em]
    \item \textbf{Functions.}
    Each $g_k$ is continuous and not a polynomial; each $f_j$ is continuous.

    \item \textbf{Warp normalization.}
    Each $h_k$ is strictly monotone with median zero, so that $h_k(0)=0$,
    and
    \[
        h_k(\varepsilon)\to\pm\infty
        \qquad\text{as}\qquad
        \varepsilon\to\pm\infty.
    \]

    \item \textbf{Analyticity.}
    Each $g_k$, $f_j$, and $h_k$ is real analytic.

    \item \textbf{Supports.}
    $|\mathcal{B}_k|\geq 2$ for every $k$, and the supports
    $\{\mathcal{B}_k\}_{k=1}^K$ are pairwise distinct.

    \item \textbf{Non-proportional warps.} The warps $\{h_k\}_{k=1}^K$ are pairwise non--proportional.

    \item \textbf{Distinct growth orders.}
    Any two non-proportional warp functions that occur in parameterizations
belonging to $\mathcal M$ have distinct asymptotic growth orders:
for any such $h,\widetilde h$,
\[
    \left|\frac{h(\varepsilon)}
    {\widetilde h(\varepsilon)}\right|
    \longrightarrow 0
    \qquad\text{or}\qquad
    \left|\frac{h(\varepsilon)}
    {\widetilde h(\varepsilon)}\right|
    \longrightarrow \infty
    \qquad\text{as }\varepsilon\to+\infty.
\]

    \item \textbf{Exponential--polynomial functions.}
    Every $g_k$ is an exponential polynomial,
    \[
        g_k(u)=\sum_i p_i(u)e^{\mu_i u},
    \]
    with $p_i\in\mathbb{C}[u],$ where $\mathbb{C}[u]=\{\text{polynomials in }u\text{ with coefficients in }\mathbb{C}\}$,
    and distinct $\mu_i\in\mathbb{R}$.
\end{enumerate}

The following is imposed on the parameterizations. 

\begin{enumerate}[leftmargin=3.2em,start=1]
\Bitem{LI}\label{as:A3} \emph{Linear independence.} The warps $\{h_k\}_{k=1}^K$ are linearly independent.
\Bitem{EI}\label{as:EI} \emph{Exponential independence.} The only relation
  \[
    \sum_{m} P_m\!\big(h_m(\eps)\big)\,e^{\mu_m h_m(\eps)}\equiv\text{const},
    \qquad P_m\in\C[u],\ \mu_m\in\R\setminus\{0\},
  \]
  (finite, the $(m,\mu_m)$ distinct) is the one with every $P_m\equiv0$.
\Bitem{S}\label{as:Sing} \emph{Complex singularity.} Each $g_k$ (and each $\tg_l$) satisfies: there
  is an integer $r_k\ge0$ and a discrete set $\mathcal{J}_k\subset\C$ (closed, no accumulation point in
  $\C$), with $\mathcal{J}_k\neq\varnothing$, such that $g_k^{(r_k)}$ extends to a holomorphic function
  on $\C\setminus \mathcal{J}_k$ of which no point of $\mathcal{J}_k$ is a removable singularity (similarly for
  $\tg_l$). Moreover, for every $\rho\neq0$, the difference $g_k-\tg_l(\cdot/\rho)$ either is a
  polynomial or itself satisfies this property, for some order and some nonempty discrete set.
\Bitem{N2}\label{as:T2} \emph{Sparse co--activation.} For each axis $j\in\{1,\dots,p\}$, let $S_j=\{k:\beta_{kj}\neq0\}$. 
Define an equivalence relation $\simeq_j$ on $S_j$ by 
$k\simeq_j \ell$ iff $h_\ell=\alpha h_k$ for some $\alpha\in\R\setminus\{0\}$. 
Then
\[
|S_j/{\simeq_j}| \le 2 \qquad\text{for every }j.
\]
\end{enumerate}

Lemma~\ref{res:growth} shows that, in the presence of \ref{as:A2}, Assumption
\ref{as:G} implies \ref{as:A3}, and that the implication is strict. Assumption
\ref{as:A3} is therefore not an additional requirement in
Theorem~\ref{res:growthid}, where \ref{as:G} is assumed; it is stated separately
because Propositions~\ref{res:assemble} and~\ref{res:pin} and
Theorem~\ref{res:reduction} do not assume \ref{as:G}. The remaining combinations
imply nothing further: Remark~\ref{res:fail} satisfies
\ref{as:A1}--\ref{as:A3} together with \ref{as:E} and is not rigid, and
Remark~\ref{rem:polyG-A6-fails} satisfies \ref{as:G} with polynomial profiles, so
neither \ref{as:E} nor \ref{as:G} can be dispensed with in
Lemma~\ref{lem:Gprime-rigid}.

In the following, the statements below concern an \emph{arbitrary} observationally equivalent pair in $\mathcal{M}$; two parameterizations
$\bigl((f_j)_{j=1}^p,(\beta_k,h_k,g_k)_{k=1}^K\bigr)$ and
$\bigl((\tf_j)_{j=1}^p,(\tb_l,\th_l,\tg_l)_{l=1}^{\tK}\bigr)$ inducing the same conditional law on $\Xtr$. 

\subsection{Proofs of Theorem \ref{THM:identify_main} and Theorem \ref{thm:extrapolation}}\label{sec:proofs}

\begin{proof}[Proof of Theorem~\ref{res:growthid} (Theorem~\ref{THM:identify_main})]
By Lemma~\ref{res:growth} (using \ref{as:A2}), \ref{as:G} implies \ref{as:A3}. With \ref{as:E}
(real--frequency exponential--polynomial profiles), Lemma~\ref{lem:Gprime-rigid} gives
rigidity. Its hypotheses and \ref{as:A3} then hold, so Theorem~\ref{res:reduction}
yields the matching.
\end{proof}

\begin{proof}[Proof of Theorem \ref{thm:extrapolation}]
For a population minimizer, Lemma \ref{lemma_engression_prop1} together with Corollary \ref{cor:ae2all}, and Remark \ref{rem:continuity-laws} gives $T(x,\varepsilon)\overset{d}{=}\tilde T(x,\varepsilon)$ for all
$x\in\mathcal{X}_{\text{tr}}$. Lemma \ref{initial_step_lemma} therefore yields $T(x,\varepsilon)=\tilde T(x,\varepsilon)$  
for every $x\in \Xtr$, every $\varepsilon\in \R$. By identifiability
(Definition \ref{DEF_model_is_identifiable}), after relabelling there are $r_k\neq0$
and constants $b_i,\gamma_k$ with $\sum_i b_i+\sum_k\gamma_k=0$ such that
\begin{align}
 f_j(x_j)=\tilde f_j(x_j)+b_j\ \text{for}\ x_j\in[a_j,A_j],\quad\text{and}\quad
 g_k(\bar x_k)=\tilde g_k(\tfrac{1}{r_k}\bar x_k)+\gamma_k,\ \ x\in\mathcal{X}_{\text{tr}},\ \varepsilon\in\mathbb{R},
\end{align}
with $\bar x_k=\beta_k^\top x+h_k(\varepsilon)$.

\emph{Assertion 2.} Fix $k$ and any $i\in\mathcal B_k$. On the training axis
$\bar x_k=\beta_{ki}t+h_k(\varepsilon)$, and since $h_k$ is unbounded above and below
as $\varepsilon$ ranges over $\mathbb{R}$ (A2), $\bar x_k$ attains every value in
$\mathbb{R}$. Hence $g_k(u)=\tilde g_k(u/r_k)+\gamma_k$ holds for all $u\in\mathbb{R}$,
so for every $x\in\mathbb{R}^p$, $\varepsilon\in\mathbb{R}$,
$g_k(\beta_k^\top x+h_k(\varepsilon))=\tilde g_k(\tilde\beta_k^\top x+\tilde h_k(\varepsilon))+\gamma_k$
(using $\beta_k=r_k\tilde\beta_k$, $h_k=r_k\tilde h_k$). Summing over $k$ gives
$G(x,\varepsilon)=\tilde G(x,\varepsilon)+\sum_k\gamma_k$ on $\mathbb{R}^p$.

\emph{Assertion 1.} For $x\in\mathcal{X}_{\text{test}}$ every $x_i\in[a_i,A_i]$, so the
additive identity applies alongside the interaction identity above, and
\[
 T(x,\varepsilon)-\tilde T(x,\varepsilon)
 =\sum_{j=1}^p\big(f_j(x_j)-\tilde f_j(x_j)\big)+\sum_{k=1}^K\big(g_k(\bar x_k)-\tilde g_k(\tfrac1{r_k}\bar x_k)\big)
 =\sum_{j=1}^p b_j+\sum_{k=1}^K\gamma_k=0 .
\]
As $\mathcal{X}_{\text{tr}}\subseteq\mathcal{X}_{\text{test}}$, this holds on
$\mathcal{X}_{\text{tr}}\cup\mathcal{X}_{\text{test}}$.
\end{proof}

\section{Synthetic Experiments and Details on Real Data Analyses}
\label{app:synthetic_generation}

We complement the real-data results with a controlled study in fixed dimension
$p = 40$. The construction isolates combinatorial extrapolation: training
covariates lie on the coordinate axes (one active component at a time), while
test covariates fill the interior of the input box. A method can therefore only
succeed by recombining the marginal effects it observes into the joint effects
it never sees. We consider two interaction regimes, \emph{sparse} and
\emph{dense}, each crossed with three response types (polynomial, log,
softplus), and compare \algname\ against ERM, Engression, and a NAM.

\subsection{Data-generating process}

\textit{Training covariates.} For each coordinate $i \in \{1,\dots,p\}$ we draw
$n_{\text{axis}}$ values uniformly from $[a, A]$ and embed each as a point
$v\,e_i$, with $v$ in the $i$-th slot and zeros elsewhere. This yields
$p\cdot n_{\text{axis}}$ training points, each with a single nonzero coordinate.

\textit{Test covariates.} We sample $n_{\text{test}}$ points uniformly from
$[a, A']^{p}$, so test points have many simultaneously active coordinates and
lie almost entirely outside the training support.

\textit{Noise.} Before evaluating the response we perturb every coordinate by
independent Gaussian noise, $\tilde x_i = x_i + \eta_i$ with
$\eta_i \sim \mathcal{N}(0, \sigma^2)$, applied identically to training and test
inputs. A smaller $\sigma$ widens the gap between the axis-aligned training
inputs and the interior test inputs; $\sigma \to 0$ is the hardest regime.

\subsection{Response functions}

Each response adds a marginal part over the $p$ coordinates to a sum over
interaction groups $S_1,\dots,S_G \subseteq \{1,\dots,p\}$, with $\alpha \ge 0$
controlling interaction strength:
\begin{align}
\text{Polynomial:}\quad
  Y &= \sum_{i=1}^{p}\tilde x_i^{3}
     + \alpha \sum_{g=1}^{G}\;\prod_{i\in S_g}\tilde x_i + 1, \\[2pt]
\text{Log:}\quad
  Y &= \sum_{i=1}^{p}\log\!\bigl(1+|\tilde x_i|\bigr)
     + \alpha \sum_{g=1}^{G}\log\!\Bigl(1+\bigl|\textstyle\sum_{i\in S_g}\tilde x_i\bigr|\Bigr), \\[2pt]
\text{Softplus:}\quad
  Y &= \sum_{i=1}^{p}\log\!\bigl(1+e^{\tilde x_i}\bigr)
     + \alpha \sum_{g=1}^{G}\log\!\Bigl(1+\exp\!\bigl(\textstyle\sum_{i\in S_g}\tilde x_i\bigr)\Bigr).
\end{align}
For pairs ($|S_g|=2$) the polynomial interaction is a product
$\tilde x_i\tilde x_j$ and the log and softplus interactions are functions of
$\tilde x_i+\tilde x_j$. 

\subsection{Interaction regimes}

Both regimes use $G = \lfloor \log p \rfloor$ groups, so $G = 3$ for $p=40$, with
a fixed random assignment of coordinates to groups held constant
across runs.

\textit{Sparse.} Each group is a single coordinate pair ($|S_g|=2$), so only $G$
of the $\binom{p}{2}=780$ possible pairwise interactions are active. We scale
$\alpha$ up by a factor of $10$ to keep the interaction signal comparable to the
additive part.

\textit{Dense.} Each group contains $\lfloor p/2 \rfloor = 20$ coordinates and
contributes one high-order joint term: a degree-$20$ product for the polynomial
case, a function of the coordinate sum otherwise. Because these terms can be
very large, we scale $\alpha$ down by a factor of $10$.

\subsection{Settings}

We fix $p=40$ with $n_{\text{axis}}=100$ ($4000$ training points) and
$n_{\text{test}}=5000$. The training range is $[a,A]=[0.1,2.0]$ and the test
range $[a,A']=[0.1,2.1]$. The base interaction strength is $\alpha=0.2$ (before
the regime-specific scaling above) and the noise level is $\sigma=0.05$. Group
assignments and all model initializations are seeded through
\texttt{reproduce(seed)}, which fixes the Python, NumPy and PyTorch generators
and sets deterministic cuDNN behaviour. Results are averaged over $10$ seeds.

\subsection{Models and hyperparameters}\label{sec:models_hyperparameters}

\textit{\algname.} Each of the $p$ marginal functions is a $3$--$4$ layer MLP with
$100$ units per layer and LeakyReLU activations. The interaction part uses
$K=1$ interaction functions, each a $2$-layer, $100$-unit StoNet block from
Engression \cite{shen2024engression}.

\textit{Baselines.} ERM is a feedforward network with $100$ units per layer; its
depth is chosen to roughly match \algname's total depth,
$\lceil p\,d_{\text{marginal}} + K\,d_{\text{engressor}}\rceil$, which for
$p=40$ exceeds $200$ layers and is therefore capped at $20$--$25$ for
tractability. Engression (StoNet) uses the original implementation with an
ensemble of $2$-layer, $100$-unit engressors, under the same cap. NAM \cite{agarwal2021neural} is
included as an additive baseline that by construction cannot represent any of
the interaction terms above; it isolates how much of each method's accuracy
comes from the additive component alone.

\textit{Training.} All models use Adam with a constant learning rate of
$10^{-4}$, train for $300$ epochs with batch size $10^4$, and run on a single
NVIDIA T4 GPU.

\subsection{Results}

\begin{table}[ht]
\centering
\caption{%
  Test MSE (mean $\pm$ std, 10 seeds) for DExtrI, NAM, ERM, and Engression on the
  high-dimensional synthetic benchmark ($p = 40$).
  \emph{Sparse}: $\lfloor\log p\rfloor$ pairwise interaction groups;
  \emph{Dense}: $\lfloor\log p\rfloor$ groups of size $\lfloor p/2 \rfloor$.
}
\label{tab:synthetic}
\setlength{\tabcolsep}{7pt}
\renewcommand{\arraystretch}{1.2}
\begin{tabular}{llrrrr}
\toprule
\textbf{Function} & \textbf{Regime}
  & \textbf{DExtrI} & \textbf{NAM}
  & \textbf{ERM} & \textbf{Engression} \\
\midrule
\multirow{2}{*}{Log}
  & Sparse & $\mathbf{41.0} \pm 2.45$ & $47.6 \pm 2.07$ & $7180 \pm 15.7$ & $6012 \pm 547$ \\
  & Dense  & $\mathbf{0.366} \pm 0.037$ & $0.409 \pm 0.033$ & $676.9 \pm 0.30$ & $534.8 \pm 107$\\
\addlinespace[3pt]
\multirow{2}{*}{Polynomial}
  & Sparse & $\mathbf{6445} \pm 130$ & $6809 \pm 32.1$ & $30491 \pm 8$ & $33905 \pm 38849$ \\
  & Dense  & $\mathbf{7318} \pm 24.6$ & $7359 \pm 2.98$ & $19930 \pm 5$ & $42656 \pm 47920$ \\
\addlinespace[3pt]
\multirow{2}{*}{Softplus}
  & Sparse & $\mathbf{456} \pm 50.5$ & $592 \pm 52.9$ & $16364 \pm 44$ & $2481 \pm 889$\\
  & Dense  & $\mathbf{55.1} \pm 3.90$ & $71.0 \pm 3.76$ & $1683 \pm 1$ & $629 \pm 567$\\
\bottomrule
\end{tabular}
\end{table}

\algname\ attains the lowest test MSE in every cell of
Table~\ref{tab:synthetic}, frequently by one to several orders of magnitude. The
gap is largest in the sparse regime, where a handful of pairwise interactions
carry the entire off-axis signal and the baselines have no mechanism to recover
them from marginal data. The advantage holds across all three response types and
persists in the dense regime, where the joint terms involve half the
coordinates. ERM and Engression not only sit well above \algname\ in mean error
but also show large run-to-run variance, a sign of unstable extrapolation rather
than a near-miss. NAM, as a strictly additive model, fits the marginal part but
cannot express the interaction terms, so it trails \algname\ wherever those terms
contribute to the response.

\subsection{Sensitivity to $K$}

We varied the number of interaction engressors $K$ on this suite. The test-loss
curves show minima at intermediate $K$, but these minima do not
appear in the energy-score or training-MSE traces, indicating a test-specific
effect rather than a quantity that can be tuned on the training axes. Overall the
gains from varying $K$ are small and inconsistent as visible in Figure \ref{fig:K_sweep}.

\begin{figure}
    \centering
    \includegraphics[width=\linewidth, height=0.90\textheight]{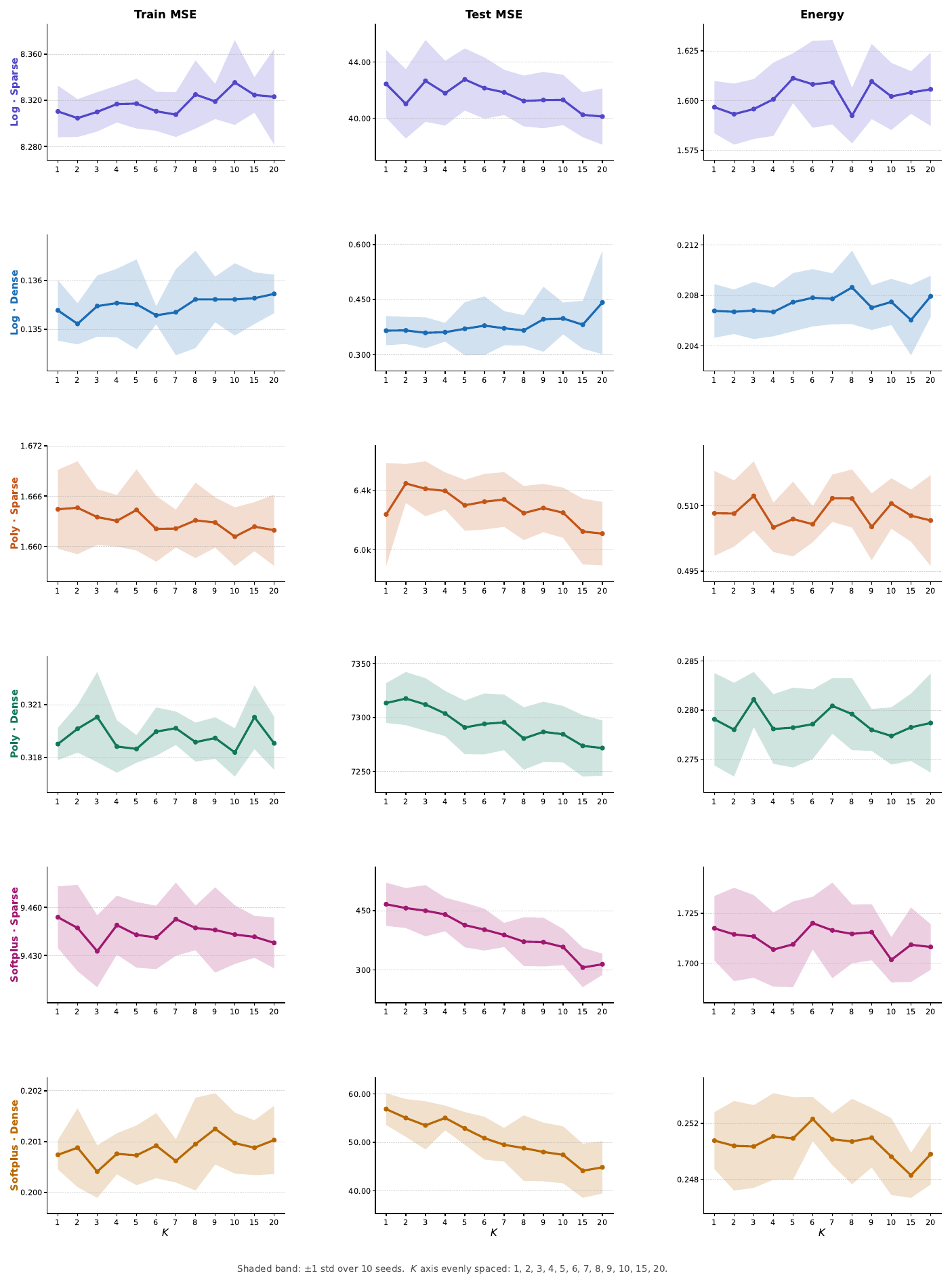}
    \caption{Test MSE performance of our method in dependence on $K$ for different configurations on the synthetic data generated as for Table \ref{tab:synthetic}.}
    \label{fig:K_sweep}
\end{figure}

\subsection{Hyperparameters}
In synthetic experiments (Table~\ref{tab:synthetic}), the \algname\ model was constructed using $p$ marginal MLPs, each consisting of three or four layers with 100 hidden units per layer and LeakyReLU activations. We set the interaction parameter to $K=1$ in all experiments. Each interaction function MLP was implemented using the 2-layer or 4-layer, 100-unit StoNet architecture from Engression \cite{shen2024engression}, available at \url{https://github.com/xwshen51/engression}. 

The ERM baseline uses a standard feedforward architecture with 100 units per layer. The depth is set to approximately match the total depth of \algname's networks: for \algname\ with $p$ marginal MLPs of depth $d_{\text{marginal}}$ and $K$ interaction functions of depth $d_{\text{interaction}}$, the ERM depth is $\lceil p \cdot d_{\text{marginal}} + K \cdot d_{\text{interaction}} \rceil$. When this sum exceeds 25 layers (occurring when $p > 30$), we cap both ERM and StoNet at 20--25 layers to maintain computational tractability. All networks were trained with the Adam optimizer using a constant learning rate of $10^{-4}$. All experiments were performed on a single NVIDIA T4 GPU for 300 epochs for synthetic and 1000--2000 epochs for real data.

\subsection{O'Neil Dataset and preprocessing details}\label{subseq:oneil_desc}

\paragraph{Raw data.}
The dataset is sourced from the preprocessed version of the Merck OncoPolyPharmacology Screen provided by Twarog et al. \cite{twarog2021data}, which supplies single-agent measurements in \texttt{AllSings.txt} ($11{,}856$ rows, $6$ viability replicates per triplet) and combination measurements in \texttt{AllCombs.txt} ($358{,}560$ rows, $4$ replicates). 
Viability is measured as fractional cell growth relative to a vehicle control and is clipped to $[0, 1.5]$ to remove assay outliers. Replicates are either averaged to obtain a single estimate of the underlying response $\mathbb{E}_\varepsilon[T(x_i, \varepsilon)]$, or retained individually as independent noise samples $\{T(x_i, \varepsilon_{ij})\}_{j=1}^{R}$ depending on the experimental setting, where $R \in \{4, 6\}$ for combination and single-agent measurements respectively.

\paragraph{Feature encoding.}
Concentrations are log-transformed as $\tilde{c} = \log_{10}(c) + 1$ to linearise the 
dose-response relationship and avoid a degeneracy at $c = 1\,\mu\text{M}$ where 
$\log_{10}(1) = 0$ would render the feature indistinguishable from a drug-absent slot. 
Each observation is encoded as a vector $x \in \mathbb{R}^{38}$, where entry $i$ 
equals $\tilde{c}_i$ if drug $i$ is present and $0$ otherwise. This encoding is shared 
across single-agent and combination observations, yielding identical feature 
dimensionality: single-agent vectors have exactly one nonzero entry, and combination 
vectors have exactly two. Cell line identity is optionally appended as a one-hot vector 
$e_k \in \{0,1\}^{39}$, giving $x \in \mathbb{R}^{77}$ in that setting.

\paragraph{Normalization and regression setup.}
All features are standardised by the per-feature standard deviation computed on the training set. The response $y$ is standardised by the training standard deviation.
Denoting the training set $\mathcal{D}_{\text{train}} = \{(x_i, y_i)\}_{i=1}^{N}$ with $N = 11{,}856$ (or $N = 71{,}136$ when replicates are retained) and the test set $\mathcal{D}_{\text{test}} = \{(x_j, y_j)\}_{j=1}^{M}$ with $M = 358{,}560$ (or $M =1{,}434{,}240$ with replicates), the regression task is to learn $f : \mathbb{R}^{38} \to \mathbb{R}$ from $\mathcal{D}_{\text{train}}$ and evaluate generalization on $\mathcal{D}_{\text{test}}$. Crucially, no combination observations appear in training, so test inputs lie strictly outside the support of the training distribution in the sense that $\|x\|_0 = 2$ for all test points while $\|x\|_0 = 1$ for all training points.

\subsection{HPOBench tabular MLP benchmark: dataset and preprocessing details}
\label{app:hpobench}

\paragraph{Benchmark.}
The \textsc{HPOBench} tabular \textsc{mlp} family \citep{eggensperger2021hpobench} provides exhaustively pre-evaluated performance tables for a two-hidden-layer multilayer perceptron trained on eight \textsc{OpenML} binary classification tasks \citep{vanschoren2013openml}.
For each task, the model is trained on the fixed \textsc{OpenML} training split, with 33\% of that split held out as a validation set (stratified, fixed seed), and performance is recorded at full training budget (maximum number of epochs). Each entry in the table is the average of five independent training runs with different random seeds. The \textsc{OpenML} task identifiers are listed in Table~\ref{tab:hpobench_tasks}.

\paragraph{Hyperparameter space.}
The five hyperparameters and their discrete grids are given in Table~\ref{tab:hpobench_hps}. All grids are log-uniformly spaced except \texttt{depth}, which takes integer values in $\{1, 2, 3\}$. Four hyperparameters (\texttt{alpha}, \texttt{batch\_size}, \texttt{learning\_rate\_init}, \texttt{width}) span several orders of magnitude and are therefore log$_{10}$-encoded prior to any further preprocessing; \texttt{depth} is left on its original integer scale.

\paragraph{Axis-aligned split.}
A baseline configuration is fixed at the componentwise grid midpoint (index $\lfloor n_j / 2 \rfloor$ for a grid of size $n_j$), giving $\boldsymbol{\lambda}^0 = (\alpha^0, b^0, d^0, \eta^0, w^0)$. The training set is axis-aligned: for each hyperparameter $j \in \{1,\ldots,5\}$ in turn, all $n_j$ grid values are evaluated while the remaining four hyperparameters are fixed at $\boldsymbol{\lambda}^0$, producing $\sum_{j=1}^{5} n_j = 43$ training observations. The test set consists of $5{,}000$ configurations drawn uniformly at random (without replacement) from the full combinatorial grid of $30{,}000$ joint configurations, in which all five hyperparameters vary simultaneously.

\paragraph{Preprocessing.}
Let $\boldsymbol{v}^0 \in \mathbb{R}^5$ denote the encoded baseline vector. Each input $\boldsymbol{x}$ is first centred by subtracting $\boldsymbol{v}^0$, giving $\tilde{\boldsymbol{x}} = \boldsymbol{x} - \boldsymbol{v}^0$, so that the baseline maps to the origin and the axis-aligned training vectors each have exactly one non-zero coordinate. The centred inputs are then standardised column-wise by the per-feature standard deviation $\hat{\sigma}_j$ computed on the training set alone, using a zero-mean scaler (i.e.\ the mean is not resubtracted, preserving the centering around the baseline):
\[
  x_j^{\text{scaled}} \;=\; \frac{\tilde{x}_j}{\hat{\sigma}_j}, \qquad
  j = 1,\ldots,5.
\]
The response $y$ (validation error rate) is standardised by the training standard deviation $\hat{\sigma}_y$, also computed on the 39 training observations only.
The same scaler is applied to the test set. All test \textsc{mse} values reported in Table~\ref{tab:hpobench} are on this normalised scale.

\begin{table}[h]
\centering
\caption{OpenML task identifiers and dataset statistics for the eight
         \textsc{HPOBench} \textsc{mlp} benchmark tasks used in
         Section~\ref{sec:hpobench}.}
\label{tab:hpobench_tasks}
\begin{tabular}{lrrr}
\toprule
\textbf{Dataset} & \textbf{OpenML task ID} & \textbf{Observations} & \textbf{Features} \\
\midrule
australian         & 146818 &   690 & 14 \\
car                & 146821 & 1728 & 21 \\
phoneme            & 9952   & 5404 &  5 \\
vehicle            & 53     &  846 & 18 \\
kc1                & 3917   & 2109 & 21 \\
segment            & 146822 & 2310 & 19 \\
blood\_transfusion & 10101  &  748 &  4 \\
credit\_g          & 31     & 1000 & 20 \\
\bottomrule
\end{tabular}
\end{table}

\begin{table}[h]
\centering
\caption{Hyperparameter search space for the \textsc{HPOBench} tabular
         \textsc{mlp} benchmark. All grids are log-uniformly spaced except
         \texttt{depth}. Hyperparameters marked $\dagger$ are
         log$_{10}$-encoded before preprocessing.}
\label{tab:hpobench_hps}
\begin{tabular}{llrcc}
\toprule
\textbf{Hyperparameter} & \textbf{Type} & \textbf{Grid pts} & \textbf{Range} & \textbf{Encoded} \\
\midrule
\texttt{alpha}              & float   & 10 & $[10^{-8},\; 1]$         & $\log_{10}$ $^\dagger$ \\
\texttt{batch\_size}        & integer & 10 & $[4,\; 256]$             & $\log_{10}$ $^\dagger$ \\
\texttt{depth}              & integer &  3 & $\{1, 2, 3\}$            & identity \\
\texttt{learning\_rate\_init} & float & 10 & $[10^{-5},\; 1]$         & $\log_{10}$ $^\dagger$ \\
\texttt{width}              & integer & 10 & $[16,\; 1024]$           & $\log_{10}$ $^\dagger$ \\
\midrule
\multicolumn{2}{l}{Full grid size} & $30{,}000$ & & \\
\multicolumn{2}{l}{Axis-aligned training points} & $43$ & & \\
\bottomrule
\end{tabular}
\end{table}

\section{Extension for Control Functions}

In many applications, only a subset of features exhibits interactions of interest, while the remaining variables act as control features. 
In such cases, the framework can be extended to model interactions among the relevant features without disregarding the influence of the controls. 
For example, in addition to measurements of single-drug effects, one may observe background variables such as cell line, measurement device, or patient characteristics, all of which may affect the response.

Formally, let $X \in \mathbb{R}^p$ denote the features of interest and $Z \in \mathbb{R}^q$ denote control variables. 
We assume that interactions occur only among the components of $X$, while $Z$ enters the model through a general (non-interacting) effect. 
Moreover, we assume that the support $\mathcal{Z}$ of $Z$ is fully observed.

The model then takes the form
\begin{align*}
    Y = \sum_{j=1}^p f_j(X_j) + \sum_{k=1}^K g_k(\beta_k^\top X + \eta_k) + \iota(Z),
\end{align*}
where $\iota$ is an arbitrary differentiable function, and $f_j, g_k$ are defined as in the previous sections.

The observed covariate support is given by the product of the axis-aligned training set for $X$ and the full support of $Z$, i.e.,
\[
\mathcal{X}_{\mathrm{tr}} \times \mathcal{Z}.
\]
Accordingly, the model class $\mathcal{M}$ is extended to include functions of the form above, incorporating the additional component $\iota$.
Identifiability and extrapolation guarantees over $\mathcal{X}_{\mathrm{test}} \times \mathcal{Z}$ remain unchanged under this extension. 
In particular, no extrapolation is required along the $Z$-coordinates, as their support is fully observed. 

\section{Additional Results}

\subsection{Auxiliary results}
\label{sec:results}

\begin{proposition}\label{res:assemble}
Under \ref{as:A1}, \ref{as:A2}, \ref{as:Supp} and \ref{as:C} if any pair of parametrizations, $((f_j)_{j=1}^p, (\beta_k, h_k, g_k)_{k=1}^K)$, $((\tilde{f}_j)_{j=1}^p, (\tilde{\beta}_l, \tilde{h}_l, \tilde{g}_l)_{l=1}^{\tilde{K}})$, in $\mathcal{M}$ is rigid, then on each axis $j$ there is a
bijection $\sigma_j:S_j\to\tS_j$ with
$\th_{\sigma_j(k)}=\alpha_{kj}h_k$, $\tb_{\sigma_j(k),j}=\alpha_{kj}\beta_{kj}$, and
$g_k-\tg_{\sigma_j(k)}(\alpha_{kj}\,\cdot)\in\cP$ for some $\alpha_{kj}\neq0$; these are the
restrictions of one global bijection $\pi$, and $K=\tK$ with
\[
  \Btr_k=\tB_{\pi(k)},\qquad \beta_k=\rho_k\tb_{\pi(k)},\qquad h_k=\rho_k\th_{\pi(k)},
  \qquad g_k-\tg_{\pi(k)}(\cdot/\rho_k)\in\cP\quad(\rho_k\neq0).
\]
\end{proposition}

\begin{proposition}\label{res:pin}
Under \ref{as:A3}, $\cR=\cP_0$. Consequently, under \ref{as:A1},\ref{as:A2}, \ref{as:Supp} and \ref{as:C}, if a pair of parametrizations $((f_j)_{j=1}^p, (\beta_k, h_k, g_k)_{k=1}^K)$, $((\tilde{f}_j)_{j=1}^p, (\tilde{\beta}_l, \tilde{h}_l, \tilde{g}_l)_{l=1}^{\tilde{K}})$, in $\mathcal{M}$, is rigid, then in addition to
Proposition~\ref{res:assemble} the functions satisfy $g_k=\tg_{\pi(k)}(\cdot/\rho_k)+c_k$ for
constants $c_k$, and the model is identifiable in the sense of Definition \ref{DEF_model_is_identifiable}.
\end{proposition}

\begin{theorem}\label{res:reduction}
Assume \ref{as:A1}, \ref{as:A2}, \ref{as:Supp}, \ref{as:C}, and \ref{as:A3}.
\begin{enumerate}[label=\textup{(\roman*)}]
  \item If a pair of parametrizations is rigid, then it is identifiable in the sense of
    Definition~\ref{DEF_model_is_identifiable}.
  \item Rigidity is also necessary: if the pair is not rigid, there exist two
    parameterizations satisfying \ref{as:A1}--\ref{as:C} that induce the same
    conditional law on $\Xtr$ but are not related by the indeterminacies of
    Proposition~\ref{res:assemble}; in particular their interaction supports ${\cal B}_k$ may differ (the following Remark).
\end{enumerate}
\end{theorem}

\begin{remark}\label{res:fail}
There exist warps $h_1,h_2,h_3$ satisfying \ref{as:A2}, \ref{as:AR}, \ref{as:C} and \ref{as:A3}, and functions
satisfying \ref{as:A1}, forming two parameterizations that induce the same conditional law on $\Xtr$
with distinct interaction supports; for them the pair is not rigid. Explicitly, $c_1=c_2=c_3=1$, 
\[
  h_1(\eps)=\eps,\quad h_2(\eps)=\eps^3,\quad h_3(\eps)=\log\!\big(e^{\eps}+e^{\eps^3}\big) -\log2,
  \qquad \Th_1=\Th_2=e^{(\cdot)},\quad \Th_3=-2e^{(\cdot)},
\]
satisfy $\sum_{m=1}^3\Th_m(t+h_m(\eps))\equiv0$.
\end{remark}

The following results discuss the case when $g_k, \tilde{g}_l$ are exponential polynomials, as stated in Assmption \ref{as:E}.

Let $\C[z]$ denote the set of polynomials in $z$, with coefficients in $\C$.
\begin{lemma}\label{res:dissect}
Let $\nu\in\R\setminus\{0\}$, let $M$ be finite and nonempty, and
for $m\in M$ let $c_m\in\R\setminus\{0\}$ and $P_m\in\C[z]$, with the foliations $\phi_m:=h_m/c_m$ (normalized warp) pairwise non--cofoliar (Assumption \ref{as:C}). Under Assumptions \ref{as:A2}, \ref{as:G}, if
\begin{equation}\label{eq:onefreq}
  \sum_{m\in M}P_m\!\big(h_m(\eps)\big)\,e^{(\nu/c_m)h_m(\eps)}=C
  \qquad(\text{constant, all }\eps\in\R),
\end{equation}
then $P_m\equiv0$ for every $m\in M$, and $C=0$.
\end{lemma}

 \begin{lemma}\label{lem:Gprime-rigid}
Assume \ref{as:A1}--\ref{as:AR} and \ref{as:C}--\ref{as:E}. Then no collapse \eqref{eq:collapse} with a non--polynomial $\Theta_m$ can occur on any
axis, and the pair is rigid.
\end{lemma}

\begin{lemma}\label{res:growth}
Under \ref{as:A2}, assumption \ref{as:G} implies \ref{as:A3}. The implication is strict:
\ref{as:A3} does not imply \ref{as:G}.
\end{lemma}

\begin{theorem}[Theorem \ref{THM:identify_main}]\label{res:growthid}
Under \ref{as:A1}--\ref{as:E}, the pair is rigid, and therefore $K=\tK$ and, after relabelling,
\[
  \Btr_k=\tB_{\pi(k)},\quad \beta_k=\rho_k\tb_{\pi(k)},\quad h_k=\rho_k\th_{\pi(k)},
  \quad g_k=\tg_{\pi(k)}(\cdot/\rho_k)+c_k .
\]
\end{theorem}

\begin{remark}[Pooled \ref{as:G} does not identify polynomial profiles]\label{rem:polyG-A6-fails}
Assumption \ref{as:G} controls ratios of the raw warps but not of their powers, and this does
not suffice for polynomial profiles. Set $z(\eps)=\eps+\eps^3$,
$q(s)=(1+s^2)^{1/4}+4s^2/(1+e^{s})$, $U(r)=\int_0^r q$, and
\[
  h_1=z,\qquad h_3=U(z),\qquad h_2=z^2+U(z),
\]
so that $h_2=h_1^2+h_3$ holds exactly. Each $h_k$ is real--analytic, strictly increasing,
normalised ($h_k(0)=0$), and divergent, with positive--tail orders $h_1\sim z$,
$h_3\sim\tfrac23z^{3/2}$, $h_2\sim z^2$; thus $h_1\ll h_3\ll h_2$ and the warps satisfy
\ref{as:G} and are pairwise non--proportional. Take $p=5$, $K=\tK=3$, identical warps and
profiles $g_1=u^3$, $g_2=g_3=u^2$ in both parameterizations, and directions
\[
  \beta_1=(1,1,1,0,0),\ \beta_2=(1,1,0,1,0),\ \beta_3=(4,4,0,0,1),
\]
\[
  \tb_1=(-1,-1,1,0,0),\ \tb_2=(4,4,0,1,0),\ \tb_3=(1,1,0,0,1),
\]
each of support size three, pairwise distinct within each model. On $\Xtr$ at most one
coordinate is nonzero; on coordinates $3,4,5$ a single interaction is active with equal
coefficient in both models, and on coordinates $1,2$ the responses differ by
\[
  (t{+}h_1)^3-({-}t{+}h_1)^3+(t{+}h_2)^2-(4t{+}h_2)^2+(4t{+}h_3)^2-(t{+}h_3)^2
  =2t^3+6t\bigl(h_1^2-h_2+h_3\bigr)=2t^3,
\]
using $h_2=h_1^2+h_3$: the cubic's cross term $6t\,h_1^2$ is cancelled by the quadratics'
$-6t\,h_2$ and $6t\,h_3$. Absorbing $2t^3$ into $\tf_1=\tf_2=2u^3$ (other $f_j,\tf_j$ zero)
gives $T=\widetilde T$ pointwise on $\Xtr\times\R$. The warps' distinct growth orders force any
relabelling to match $h_k$ with $\tilde h_k$, hence (per--term rescaling) $\rho_k=1$ and
$\beta_k=\tb_k$; but $\beta_1=(1,1,1,0,0)$ and $\tb_1=(-1,-1,1,0,0)$ are not proportional.
Thus \ref{as:G} does not identify polynomial--profile models.
\end{remark}

\subsection{Proofs of auxiliary results in Section \ref{sec:results}}

\begin{proof}[Proof of Proposition~\ref{res:assemble}]
Fix an axis~$j\in\{1,\dots,p\}$. Let
\[
S_j=\{k:\beta_{kj}\neq0\},\qquad \tS_j=\{l:\tb_{lj}\neq0\}
\]
denote the sets of interaction indices active on this axis in the two parameterizations. On the strip
$U_j=(a_j,A_j)\times\R$, the identity \eqref{eq:axis} reads
\begin{align}
\sum_{k\in S_j} g_k(\beta_{kj}t+h_k(\eps))-\sum_{l\in\tS_j}\tg_l(\tb_{lj}t+\th_l(\eps))
=A_j(t)+B_j(\eps).
\end{align}

Define an equivalence relation $\sim_j$ on the union $S_j\sqcup\tS_j$ by declaring two indices
equivalent iff their associated warped ridges are cofoliar; by Lemma \ref{lem:cofoliar}, this holds precisely
when the corresponding normalised warps coincide. Let $\mathcal{C}_j$ denote the set of equivalence
classes of $S_j\sqcup\tS_j$ under $\sim_j$. For a class $C\in\mathcal{C}_j$, set
\[
I_C:=C\cap S_j,\qquad J_C:=C\cap \tS_j,
\]
the active indices of the first and second parameterization, respectively, that fall into this class.

We first claim that $|I_C|\le 1$ and $|J_C|\le 1$ for every class $C$. Indeed, suppose
$k_1,k_2\in I_C$ with $k_1\neq k_2$. Then the warped ridges $( \beta_{k_1j}, h_{k_1})$ and
$( \beta_{k_2j}, h_{k_2})$ are cofoliar, so by Lemma \ref{lem:cofoliar} there exists $\alpha\neq0$ such that
$h_{k_2}=\alpha h_{k_1}$, contradicting Assumption~\ref{as:C} (pairwise non-proportionality of the
warps). The same argument applied to the second parameterization gives $|J_C|\le 1$, since
Assumption~\ref{as:C} also holds for $\{\th_l\}$. Thus each equivalence class has either
\begin{itemize}
  \item a single first-side index, $(I_C,J_C)=(\{k\},\varnothing)$,
  \item a single second-side index, $(I_C,J_C)=(\varnothing,\{l\})$,
  \item one of each, $(I_C,J_C)=(\{k\},\{l\})$,
  \item or is empty (which we ignore).
\end{itemize}
Classes of the first two types are called \emph{singletons}; classes of the third type are
\emph{matched pairs}. Distinct classes have distinct foliations, since each class is defined by
a unique normalised warp.

Now invoke rigidity (Definition~\ref{def:rigid}). It requires that, for every class $C\in\mathcal{C}_j$,
the combined class function
\[
F_C(t,\eps):=
\sum_{k\in I_C} g_k(\beta_{kj}t+h_k(\eps))
-
\sum_{l\in J_C} \tg_l(\tb_{lj}t+\th_l(\eps))
\]
belongs to the polynomial space $\cP$ (viewed as a function of $t$ and $\eps$ on $U_j$). Consider a
singleton of the first type, $C=(\{k\},\varnothing)$. Then
\[
F_C(t,\eps)=g_k(\beta_{kj}t+h_k(\eps)).
\]
If $F_C\in\cP$, then $g_k$ is a polynomial in its argument, since the map $u=\beta_{kj}t+h_k(\eps)$
is surjective onto $\R$ (because $h_k$ is strictly monotone by Assumption~\ref{as:A2} and $t$ varies
over an interval). This contradicts Assumption~\ref{as:A1}, which requires each $g_k$ to be
non-polynomial. The same reasoning excludes singletons $(\varnothing,\{l\})$, as they would force
$\tg_l\in\cP$, contradicting \ref{as:A1} for the second parameterization. Therefore, every nontrivial
equivalence class is a matched pair: for every $k\in S_j$, there exists a unique $l\in\tS_j$ such that
the ridges are cofoliar. This defines a map $\sigma_j:S_j\to\tS_j$, which is bijective by applying
the same argument in reverse.

For each $k\in S_j$, set $l=\sigma_j(k)$. Since $k$ and $l$ lie in the same matched pair, their ridges are cofoliar. By Lemma \ref{lem:cofoliar}, there exists a unique scalar $\alpha_{kj}\neq0$ such
that
\begin{align}\label{eq:mathcing_hbeta}
\th_{\sigma_j(k)}=\alpha_{kj}\,h_k,\qquad
\tb_{\sigma_j(k),j}=\alpha_{kj}\,\beta_{kj}.
\end{align}
The class function for this matched pair is
\[
F_{\{k,\sigma_j(k)\}}(t,\eps)
=
g_k(\beta_{kj}t+h_k(\eps))
-
\tg_{\sigma_j(k)}(\tb_{\sigma_j(k),j}t+\th_{\sigma_j(k)}(\eps)).
\]
Using \eqref{eq:mathcing_hbeta} and setting $u:=\beta_{kj}t+h_k(\eps)$, we obtain
\[
F_{\{k,\sigma_j(k)\}}(t,\eps)
=
g_k(u)-\tg_{\sigma_j(k)}(\alpha_{kj}u).
\]
Rigidity forces $F_{\{k,\sigma_j(k)\}}\in\cP$; hence
\begin{align}\label{eq:gtildeg_poly}
g_k-\tg_{\sigma_j(k)}(\alpha_{kj}\,\cdot)\in\cP.
\end{align}

It remains to show that the per-axis bijections $\sigma_j$ are compatible across different axes.
Take an interaction index $k$ and two axes $j,j'\in\Btr_k=\supp\beta_k$. From \eqref{eq:mathcing_hbeta},
\[
\th_{\sigma_j(k)}=\alpha_{kj}h_k,\qquad
\th_{\sigma_{j'}(k)}=\alpha_{kj'}h_k.
\]
Thus $\th_{\sigma_j(k)}$ and $\th_{\sigma_{j'}(k)}$ are proportional. By Assumption~\ref{as:C}
applied to the second parameterization, the warps $\{\th_l\}$ are pairwise non-proportional, so
these two warps must be identical: $\th_{\sigma_j(k)}=\th_{\sigma_{j'}(k)}$. Since the warps in the
second parameterization are indexed distinctly, this forces $\sigma_j(k)=\sigma_{j'}(k)$. Therefore
$\sigma_j(k)$ is independent of $j\in\Btr_k$; denote this common value by $\pi(k)$. Applying the
same argument to $\pi^{-1}$ shows that $\pi$ is a bijection $\{1,\dots,K\}\to\{1,\dots,\tK\}$, so
$K=\tK$.

Moreover, for each fixed $k$, the scalar $\alpha_{kj}$ in \eqref{eq:mathcing_hbeta} is the same for all $j\in\Btr_k$:
if $j,j'\in\Btr_k$, then $\alpha_{kj}h_k=\th_{\pi(k)}=\alpha_{kj'}h_k$, and since $h_k\not\equiv0$,
we get $\alpha_{kj}=\alpha_{kj'}=:\alpha_k$. Let $\rho_k:=1/\alpha_k$. Then \eqref{eq:mathcing_hbeta} becomes
\begin{align}\label{eq:hbeta_ratio}
\th_{\pi(k)}=\frac{1}{\rho_k}h_k,\qquad
\tb_{\pi(k),j}=\frac{1}{\rho_k}\beta_{kj}\quad\forall j\in\Btr_k.
\end{align}
The first relation in \eqref{eq:hbeta_ratio} gives $h_k=\rho_k\th_{\pi(k)}$. The second gives, for each component
$j\in\Btr_k$, $\beta_{kj}=\rho_k\tb_{\pi(k),j}$, and for $j\notin\Btr_k$, both $\beta_{kj}=0$ and
$\tb_{\pi(k),j}=0$ by definition of the supports. Hence $\beta_k=\rho_k\tb_{\pi(k)}$ as vectors.
Reading the support relations in \eqref{eq:hbeta_ratio} over all $j$ yields $\Btr_k\subseteq\tB_{\pi(k)}$; the symmetric
argument using $\pi^{-1}$ gives the reverse inclusion. Thus $\Btr_k=\tB_{\pi(k)}$. Combining this
with the vector equality, the warp equality, and the polynomial difference gives exactly the conclusion of the proposition. \qedhere
\end{proof}

\begin{proof}[Proof of Proposition~\ref{res:pin}]
Let $(q_k)\in\cR$. Fix an axis $j$. Substituting $x=te_j$ into the defining
identity~\eqref{eq:redist} and collecting onto the right-hand side every term
$q_k(h_k(\eps))$ with $\beta_{kj}=0$ (a function of $\eps$ alone) together with the
separable part of \eqref{eq:redist} gives, on the strip $U_j=(a_j,A_j)\times\R$,
\begin{equation}\label{eq:res-j}
  \sum_{k\in S_j}q_k\!\bigl(\beta_{kj}t+h_k(\eps)\bigr)=A_j(t)+B_j(\eps),
\end{equation} where $S_j=\{k:\beta_{kj}\neq0\}$.
Applying $\partial_t\partial_\eps$ to \eqref{eq:res-j} annihilates $A_j(t)$ and
$B_j(\eps)$ and yields
\begin{equation}\label{eq:res-de}
  \sum_{k\in S_j}\beta_{kj}\,h_k'(\eps)\,q_k''\!\bigl(\beta_{kj}t+h_k(\eps)\bigr)=0
  \qquad\text{on }U_j.
\end{equation}

Fix $\eps$ and interpret \eqref{eq:res-de} as a polynomial identity in $t$ (each
summand is a polynomial in $t$ with $\eps$-dependent coefficients, since $q_k''$ is a
polynomial). Set $D:=\max_{k\in S_j}\deg q_k''$ and suppose for contradiction that
$D\ge0$, i.e.\ that some $q_k$ has degree at least $2$. The coefficient of $t^D$
in \eqref{eq:res-de} is
\[
  \sum_{\substack{k\in S_j\\\deg q_k''=D}}\beta_{kj}^{\,D+1}\lambda_k\,h_k'(\eps)=0
  \qquad\forall\,\eps\in\R,
\]
where $\lambda_k\neq0$ denotes the leading coefficient of $q_k''$. Since $\beta_{kj}\neq0$
for $k\in S_j$, the scalars $\beta_{kj}^{\,D+1}\lambda_k$ are all nonzero, so this is a
nontrivial linear combination of $\{h_k'\}_{k\in S_j,\,\deg q_k''=D}$ vanishing
identically. However, the derivatives $\{h_k'\}_{k=1}^K$ are linearly independent by Corollary \ref{cor:within}. 
This contradiction shows $D<0$, i.e.\ $\deg q_k''\le -1$, and
hence every $q_k$ has degree at most $1$:
\[
  q_k(u)=c_ku+d_k,\qquad c_k,d_k\in\R.
\]

It remains to show $c_k=0$ for all $k$. Substituting $q_k(u)=c_ku+d_k$ into
\eqref{eq:res-j} and differentiating in $t$ yields
\[
  \sum_{k\in S_j}\beta_{kj}c_k = -A_j'(t)
\]
for all $t\in(a_j,A_j)$; the left-hand side is
constant, so $A_j'$ is constant. Hence $A_j(t)=e_jt+b_j$ for some $e_j,b_j\in\R$.
Substituting the affine forms back into the full residual identity \eqref{eq:redist},
\[
  \sum_j(e_jx_j+b_j)+\sum_k c_k\!\bigl(\beta_k^\top x+h_k(\eps)\bigr)+\sum_k d_k=0
  \qquad\text{on }\Xtr\times\R.
\]
Separating the $\eps$-dependent part, the identity requires $\sum_k c_kh_k(\eps)=0$
for all $\eps\in\R$; linear independence of $\{h_k\}$ \ref{as:A3} forces $c_k=0$ for all
$k$. The remaining $x$-dependent part then reduces to $\sum_j e_jx_j+\text{const}=0$
on $\Xtr$; since $\Xtr$ contains segments in each coordinate direction, this forces
$e_j=0$ for all $j$. Hence $q_k\equiv d_k$ is a constant for every $k$,
i.e.\ $(q_k)\in\cP_0$.
\end{proof}

\begin{remark}
By Proposition~\ref{res:assemble}, after relabelling by the bijection $\pi$
and setting $r_k=\alpha_k^{-1}$, the matched class function
$q_k:=g_k-\tg_{\pi(k)}(\alpha_k\,\cdot)$ is a polynomial (Step~2 of
Proposition~\ref{res:assemble}) and its tuple lies in $\cR$ after compensation
(substituting $\tg_{\pi(k)}(\alpha_k\bar x_k)=g_k(\bar x_k)-q_k(\bar x_k)$ into
\eqref{eq:full} and using $\beta_k=r_k\tb_k$, $h_k=r_k\th_k$ converts the full
identity to \eqref{eq:redist} for the tuple $(q_k)$). By $\cR=\cP_0$ each $q_k$ is a
constant, $q_k\equiv\gamma_k$, yielding
\[
  g_k(\bar x_k)=\tg_k\!\bigl(\tfrac{1}{r_k}\bar x_k\bigr)+\gamma_k,
\]
and the residual constraint $\sum_j b_j+\sum_k\gamma_k=0$ follows from evaluating
\eqref{eq:full} at $x=0$, $\eps=0$.
\end{remark}

\begin{proof}[Proof of Theorem~\ref{res:reduction}]
(i) Proposition~\ref{res:assemble} (using \ref{as:A1}, \ref{as:A2}, \ref{as:C}) produces
the bijection $\pi$ and the structural equalities; Proposition~\ref{res:pin} (using
\ref{as:A3}) reduces the residual polynomial freedom to constants.

(ii) Failure of rigidity means there exists an axis $j$ and a grouped identity
\eqref{eq:collapse} containing a foliation class with function
$\Th_m\notin\cP$. The singleton is either $\Th_m=g_k$ for some $k\in S_j$ (no term
of the second parameterization active on axis $j$ is cofoliar with the $k$-th
interaction) or $\Th_m=-\tg_l$ (symmetric case). In the first case,
if the two parametrizations were related by the admissible indeterminacies as in Proposition \ref{res:assemble}, then the interaction $k$ would have a globally matched interaction $l=\pi(k)$ satisfying $\tilde{h}_l = \alpha h_k$ and $\tilde{\beta}_{lj} = \alpha \beta_{kj}$ with the same $\alpha\neq 0 $. By
Lemma~\ref{lem:cofoliar}, the two would then be cofoliar on axis $j$, which would contradict the fact that the class containing $k$ is a singleton. Hence, the two parametrizations are not related by the admissible indeterminacies. Remark~\ref{res:fail} provides an explicit instance.

If instead the non-polynomial foliation class is matched, then for some
matched pair $(k,l)$ its combined profile is
$\Th_m=g_k-\tg_l(\alpha\,\cdot)\notin\cP$. By
Proposition~\ref{res:assemble}, two parameterizations that differ only by
the admissible indeterminacies must have a polynomial combined profile on
every matched class. Since $\Th_m$ is non-polynomial, the two
parameterizations cannot be related by those indeterminacies. As they
nevertheless induce the same conditional law on $\Xtr$, they constitute a
failure of identifiability.
\end{proof}

\begin{proof}[Proof of Lemma~\ref{res:dissect}]
Write $\phi_m=h_m/c_m$, so \eqref{eq:onefreq} is
$\sum_{m\in M}P_m(h_m)e^{\nu\phi_m}=C$. Fix distinct $m,m'$.
If $h_m$ and $h_{m'}$ are non-proportional, then by \ref{as:G},
$|\phi_m/\phi_{m'}|\to0$ or $\infty$, so exactly one of
$|\phi_m|,|\phi_{m'}|$ dominates. If
$|\phi_m/\phi_{m'}|\to\infty$, then
$\phi_{m'}/\phi_m\to0$ and
$\phi_m-\phi_{m'}=\phi_m(1-\phi_{m'}/\phi_m)\sim\phi_m$, which
tends to $\pm\infty$ since $|\phi_m|\to\infty$; the case
$|\phi_m/\phi_{m'}|\to0$ is symmetric.
If instead $h_{m'}=\alpha h_m$ for some $\alpha\neq0$, then
non--cofoliarity implies $c_{m'}\neq\alpha c_m$, and hence
\[
  \phi_m-\phi_{m'}
  =\left(\frac{1}{c_m}-\frac{\alpha}{c_{m'}}\right)h_m
  \longrightarrow \pm\infty.
\]
Hence for every distinct pair
\[
  \nu(\phi_m-\phi_{m'})\to+\infty\quad\text{or}\quad\to-\infty,
\] the case $|\phi_m/\phi_{m'}|\to0$ is symmetric. Hence for every distinct pair
\[
  \nu(\phi_m-\phi_{m'})\to+\infty\quad\text{or}\quad\to-\infty,
\]
these two alternatives being \emph{exhaustive and mutually exclusive} (the difference cannot
be bounded). Define
\[
  m\succ m'\quad\Longleftrightarrow\quad\nu(\phi_m-\phi_{m'})\to+\infty.
\]
Exclusivity makes $\succ$ a well--defined relation, exhaustivity makes it total on distinct
indices, and it is transitive because
$\nu(\phi_m-\phi_r)=\nu(\phi_m-\phi_{m'})+\nu(\phi_{m'}-\phi_r)$ and the sum of two quantities
tending to $+\infty$ tends to $+\infty$. A total order on the finite set $M$ has a maximum;
let $m^\ast$ be it, so $\nu(\phi_{m^\ast}-\phi_m)\to+\infty$ for every $m\neq m^\ast$.

Fix $m\neq m^\ast$; if $S_m\equiv0$ there is nothing
to prove, so let $d=\deg S_m$. Since $h_m=c_m\phi_m$, for large $\eps$
$|P_m(h_m)|\le C'(1+|\phi_m|)^{d}=\exp\!\big(d\log(1+|\phi_m|)\big)$. We bound the term
$P_m(h_m)e^{\nu(\phi_m-\phi_{m^\ast})}$ using \ref{as:G}.
 
\emph{Case A: $|\phi_m/\phi_{m^\ast}|\to0$.} Then $\phi_m=o(\phi_{m^\ast})$, so the prefactor
is $\exp\!\big(o(|\phi_{m^\ast}|)\big)$. Since $m^\ast\succ m$, the exponent
$\nu(\phi_m-\phi_{m^\ast})\to-\infty$ at rate of magnitude
$|\nu\phi_{m^\ast}|(1+o(1))$. Exponential decay $e^{-|\nu\phi_{m^\ast}|(1+o(1))}$ dominates
$\exp\!\big(o(|\phi_{m^\ast}|)\big)$, so the term $\to0$.
 
\emph{Case B: $|\phi_m/\phi_{m^\ast}|\to\infty$.} Then $\phi_{m^\ast}/\phi_m\to0$, so
\[
  \phi_m-\phi_{m^\ast}=\phi_m\left(1-\frac{\phi_{m^\ast}}{\phi_m}\right)=\phi_m\,(1+o(1)).
\]
Since $m^\ast\succ m$, we have $\nu(\phi_m-\phi_{m^\ast})\to-\infty$, and combining this gives $\nu\phi_m(1+o(1))\to-\infty$. Hence, for all sufficiently large $\eps$,
\[
  \nu(\phi_m-\phi_{m^\ast})\le-\tfrac{|\nu|}{2}\,|\phi_m|.
\]
As $h_m=c_m\phi_m$ and $S_m$ is a polynomial, $|P_m(h_m)|\le C'(1+|\phi_m|)^{d}$ for some
$C'>0$, $d\ge0$ and all large $\eps$. Therefore
\[
  \left|P_m(h_m)\,e^{\nu(\phi_m-\phi_{m^\ast})}\right|
  \le C'(1+|\phi_m|)^{d}\,e^{-(|\nu|/2)|\phi_m|}\longrightarrow0,
\]
exponential decay dominating polynomial growth.

\emph{Case C: $h_m$ and $h_{m^\ast}$ are proportional.}
Write $h_m=\alpha h_{m^\ast}$. Then
$\phi_m=r\phi_{m^\ast}$ for the constant
$r=\alpha c_{m^\ast}/c_m\neq1$, where $r\neq1$ follows from
non--cofoliarity. Thus
\(
  \phi_m-\phi_{m^\ast}
  =\left(1-\frac{1}{r}\right)\phi_m.
\)
Since $m^\ast\succ m$, we have
$\nu(\phi_m-\phi_{m^\ast})\to-\infty$. Hence, for some
$\kappa>0$ and all sufficiently large $\eps$,
\(
  \nu(\phi_m-\phi_{m^\ast})\le -\kappa|\phi_m|.
\)
Therefore, exactly as in Case B,
\[
  \left|P_m(h_m)e^{\nu(\phi_m-\phi_{m^\ast})}\right|
  \le C'(1+|\phi_m|)^d e^{-\kappa|\phi_m|}
  \longrightarrow0.
\]
 
Divide \eqref{eq:onefreq} by
$e^{\nu\phi_{m^\ast}}$:
\[
  P_{m^\ast}(h_{m^\ast})+\sum_{m\neq m^\ast}P_m(h_m)e^{\nu(\phi_m-\phi_{m^\ast})}
  =C\,e^{-\nu\phi_{m^\ast}},
\]
so by the previous calculation, $P_{m^\ast}(h_{m^\ast})+o(1)=C\,e^{-\nu\phi_{m^\ast}}$.
\begin{itemize}
\item If $\nu\phi_{m^\ast}\to+\infty$, the right side $\to0$, so $P_{m^\ast}(h_{m^\ast})\to0$;
  as $|h_{m^\ast}|\to\infty$, a nonzero polynomial cannot tend to $0$ along $\eps\to+\infty$,
  so $P_{m^\ast}\equiv0$.
\item If $\nu\phi_{m^\ast}\to-\infty$, first note $C=0$. Indeed, suppose $C\neq0$. The
  identity $P_{m^\ast}(h_{m^\ast})+o(1)=C\,e^{-\nu\phi_{m^\ast}}$ has, on the right, a term
  whose modulus $|C|\,e^{-\nu\phi_{m^\ast}}\to+\infty$,
  so the right side grows exponentially in $|\phi_{m^\ast}|$. The left side, however, is
  $P_{m^\ast}(c_{m^\ast}\phi_{m^\ast})+o(1)$, a polynomial in $\phi_{m^\ast}$ up to $o(1)$,
  hence bounded by $C'(1+|\phi_{m^\ast}|)^{\deg P_{m^\ast}}$ for large $\eps$. A quantity of at
  most polynomial growth cannot equal one of exponential growth as $\eps\to+\infty$; this
  contradiction forces $C=0$. With $C=0$ the identity reads
  $P_{m^\ast}(h_{m^\ast})+o(1)=0$, so $P_{m^\ast}(h_{m^\ast})\to0$; as $|h_{m^\ast}|\to\infty$,
  a nonzero polynomial cannot tend to $0$ along $\eps\to+\infty$, so $P_{m^\ast}\equiv0$.
\end{itemize}
Remove $m^\ast$ and repeat on $M\setminus\{m^\ast\}$. Induction gives $P_m\equiv0$ for all
$m$; the empty relation gives $C=0$.
\end{proof}

\begin{proof}[Proof of Lemma~\ref{lem:Gprime-rigid}]
Suppose a collapse arises on some axis. After merging exact cofoliar classes, index them
$m=1,\dots,N$, with pairwise non--cofoliar foliations $h_m/c_m$ ($c_m\neq0$), and suppose
some $\Th_{m_0}$ is not a polynomial.
 
Apply $\partial_t\partial_\eps$ to \eqref{eq:collapse}; the separable
side vanishes and each class contributes $c_m h_m'(\eps)\Th_m''(c_mt+h_m(\eps))$, so
\begin{equation}\label{eq:Gp-diff}
  \sum_{m=1}^N c_m\,h_m'(\eps)\,\Th_m''\!\big(c_mt+h_m(\eps)\big)=0 .
\end{equation}
By hypothesis $\Th_m''(u)=\sum_i\tilde p_{mi}(u)e^{\lambda_{mi}u}$ with
$\tilde p_{mi}=p_{mi}''+2\lambda_{mi}p_{mi}'+\lambda_{mi}^2 p_{mi}$; if $\lambda_{mi}\neq0$ and
$p_{mi}\not\equiv0$ then $\tilde p_{mi}\not\equiv0$, its leading coefficient being
$\lambda_{mi}^2$ times that of $p_{mi}$.

Substitute $u=c_mt+h_m$ and factor
$e^{\lambda_{mi}(c_mt+h_m)}=e^{\nu t}e^{\lambda_{mi}h_m}$ with $\nu:=\lambda_{mi}c_m$. Expand
$\tilde p_{mi}(c_mt+h_m)=\sum_a R_{mia}(h_m)\,t^a$ with $R_{mia}$ polynomial. Then
\eqref{eq:Gp-diff} is a finite combination of $\{t^a e^{\nu t}\}$, which are linearly
independent on the axis interval (distinct $\nu$ give distinct exponential rates in $t$; for
fixed $\nu$ the powers $t^a$ are independent). Hence each coefficient vanishes: for every
$\nu$ and $a$,
\begin{equation}\label{eq:Gp-coeff}
  \sum_{(m,i):\,\lambda_{mi}c_m=\nu} c_m\,R_{mia}\!\big(h_m(\eps)\big)\,h_m'(\eps)\,
  e^{\lambda_{mi}h_m(\eps)}=0 .
\end{equation}
For fixed $m$ the $\lambda_{mi}$ are distinct, so at most one $i=i(m)$ has
$\lambda_{mi}c_m=\nu$; set $\lambda_m:=\lambda_{m,i(m)}=\nu/c_m$ and let $G_\nu$ be the classes
contributing to $\nu$.

Distinct merged classes are non--cofoliar by
construction, so for $m\neq m'$ in $G_\nu$ the foliations $h_m/c_m$ and $h_{m'}/c_{m'}$ are
non--cofoliar. Hence Lemma~\ref{res:dissect} is applicable at each $\nu$.

Fix $\nu\neq0$; then each $\lambda_m=\nu/c_m\neq0$.
For each $m\in G_\nu$ and each $a$, consider the equation
\begin{equation}\label{eq:aux-ode}
  V_{ma}'+\lambda_m V_{ma}=R_{m,i(m),a}
\end{equation}
for an unknown polynomial $V_{ma}\in\C[z]$, where $R_{m,i(m),a}$ is the known polynomial from
\eqref{eq:Gp-coeff}. Write $R:=R_{m,i(m),a}$, $\lambda:=\lambda_m$, and seek
$V=\sum_{s=0}^{D}\sigma_s z^{s}$ with $D=\deg R$. In \eqref{eq:aux-ode} the top--degree
coefficient of $V'+\lambda V$ is $\lambda\sigma_{D}$ (the derivative $V'$ has degree $D-1$ and
does not reach degree $D$), so matching the coefficient of $z^{D}$ in $R$ gives
$\sigma_{D}=(\text{coeff of }z^{D}\text{ in }R)/\lambda$; this division is legitimate because
$\lambda\neq0$. Descending one degree at a time, the coefficient of $z^{s}$ gives
$\lambda\sigma_{s}=(\text{coeff of }z^{s}\text{ in }R)-(s+1)\sigma_{s+1}$, determining
$\sigma_{s}$ from the already--found $\sigma_{s+1}$, again by division by $\lambda$. Hence a
polynomial solution $S_{ma}$ exists (and has degree $\deg R$).

With $V_{ma}$ so chosen, the chain rule gives
\[
  \frac{d}{d\eps}\big[V_{ma}(h_m)e^{\lambda_m h_m}\big]
  =h_m'\,\big(V_{ma}'+\lambda_m V_{ma}\big)(h_m)\,e^{\lambda_m h_m}
  =h_m'\,R_{m,i(m),a}(h_m)\,e^{\lambda_m h_m},
\]
the last equality by \eqref{eq:aux-ode}. Thus the $m$--th summand of \eqref{eq:Gp-coeff}
(without its constant factor $c_m$) is exactly the $\eps$--derivative of
$V_{ma}(h_m)e^{\lambda_m h_m}$.

Summing over $m\in G_\nu$ with the factors $c_m$,
\eqref{eq:Gp-coeff} states that
\[
  \frac{d}{d\eps}\left[\sum_{m\in G_\nu}c_m\,V_{ma}(h_m)e^{\lambda_m h_m}\right]=0
  \qquad\text{on }\R.
\]
A function with vanishing derivative on the connected set $\R$ is constant, so
\[
  \sum_{m\in G_\nu}c_m\, V_{ma}\!\big(h_m(\eps)\big)\,e^{(\nu/c_m)h_m(\eps)}=K_{\nu,a}
\]
for some constant $K_{\nu,a}\in\C$.

This is now of the form \eqref{eq:onefreq} at frequency
$\nu$, with polynomial coefficients $S_m:=c_m V_{ma}$; the foliations $\{h_m/c_m\}_{m\in
G_\nu}$ are pairwise non--cofoliar, as noted above. Lemma~\ref{res:dissect} therefore gives
$c_m V_{ma}\equiv0$ for every $m\in G_\nu$, hence $V_{ma}\equiv0$ (as $c_m\neq0$). Feeding this
back into \eqref{eq:aux-ode},
\[
  R_{m,i(m),a}=V_{ma}'+\lambda_m V_{ma}\equiv0
\]
for every $m\in G_\nu$ and every $a$.

Since $\Th_{m_0}$ is not a polynomial, some $\lambda_{m_0i_0}\neq0$ with
$p_{m_0i_0}\not\equiv0$, so $\tilde p_{m_0i_0}\not\equiv0$. Set
$\nu=\lambda_{m_0i_0}c_{m_0}\neq0$. The $t^0$ coefficient of
$\tilde p_{m_0i_0}(c_{m_0}t+h_{m_0})$ is $\tilde p_{m_0i_0}(h_{m_0})$, so
$R_{m_0i_0 0}=\tilde p_{m_0i_0}\not\equiv0$. Contradiction.

\end{proof}

\begin{proof}[Proof of Lemma~\ref{res:growth}]
\emph{Linear independence.} Suppose $\sum_m\lambda_m h_m(\eps)\equiv\text{const}$ with not all
$\lambda_m$ zero, and let $M_0=\{m:\lambda_m\neq0\}\neq\varnothing$. The warps
$\{h_m\}_{m\in M_0}$ lie in the pooled family, and \ref{as:G} strictly orders them by growth
as $\eps\to+\infty$ (within a cofoliar class the ratio tends to a finite nonzero limit, but a
nontrivial linear relation among proportional warps already forces their coefficients to
vanish, so no two indices in $M_0$ are cofoliar and the strict order applies across all of
$M_0$). Let $m^\ast\in M_0$ be the fastest--growing. Dividing the relation by $h_{m^\ast}$,
\[
  \lambda_{m^\ast}+\sum_{m\in M_0\setminus\{m^\ast\}}\lambda_m\,\frac{h_m}{h_{m^\ast}}
  =\frac{\text{const}}{h_{m^\ast}} .
\]
As $\eps\to+\infty$, each ratio $h_m/h_{m^\ast}\to0$ (by the strict order) and
$h_{m^\ast}\to+\infty$ (\ref{as:A2}), so the right side and every sum term tend to $0$,
forcing $\lambda_{m^\ast}=0$, contradicting $m^\ast\in M_0$. Hence all $\lambda_m=0$, i.e.
\ref{as:A3} holds.

\emph{Strictness.} The warps of Remark~\ref{res:fail} are linearly independent yet do not
have pairwise distinct growth orders, so \ref{as:A3} does not imply \ref{as:G}.
\end{proof}



\subsection{Further results with exponential polynomials in the setting of non-affine $h_k$}
\begin{proposition}\label{res:N2}
Under \ref{as:A2}, \ref{as:AR}, \ref{as:C} and \ref{as:T2}, every function appearing in a collapse \eqref{eq:collapse} arising from \eqref{eq:axis} is an exponential polynomial.
\end{proposition}

The remaining result discusses the case of complex singularity.
\begin{theorem}\label{res:sing}
Under \ref{as:A2}, \ref{as:AR}, \ref{as:C}, and \ref{as:Sing}, the pair of parametrizations inducing the same distribution on $\mathcal{X}_{\mathrm{tr}}$ is rigid. 
\end{theorem}

\begin{proof}[Proof of Proposition~\ref{res:N2}]
By \ref{as:T2}, $N\le2$. We must
show that every $\Th_m$, $m=1,2$, is an exponential polynomial (as in \ref{as:E}).

We use repeatedly the following elementary fact: a real--analytic
function on $\R$ (or on an interval) that is not identically zero has only isolated zeros. This
is immediate from the identity theorem, since a real--analytic function extends locally to a
holomorphic function, and the zero set of a holomorphic function that is not identically zero on
a connected domain has no accumulation point in that domain. We call a point of the domain
\emph{good} for a given real--analytic function if the function is nonzero there; the set of
good points is then the complement of a discrete (in particular countable) set.

By \ref{as:A2}, each $h_m$ is strictly monotone, hence $h_m'\ge0$ (or $\le0$) throughout its
domain; since $h_m$ is non--constant (indeed unbounded, by \ref{as:A2}), $h_m'\not\equiv0$, so
$h_m'$ vanishes only on a discrete set, and is nonzero off it.

Apply $\partial_t\partial_\eps$ to \eqref{eq:collapse} to make the right-hand side vanish, let $\Psi_m:=\Th_m''$, then
\begin{equation}\label{eq:N2-reduced}
  \sum_{m=1}^N c_m\,h_m'(\eps)\,\Psi_m\big(c_mt+h_m(\eps)\big)=0\qquad\text{on }U_j.
\end{equation}

\textit{The case $N=1$.}
Equation \eqref{eq:N2-reduced} reads $c_1h_1'(\eps)\Psi_1(c_1t+h_1(\eps))=0$. Fix $\eps$ good
for $h_1'$ (i.e.\ $h_1'(\eps)\neq0$); since $c_1\neq0$, this forces
$\Psi_1(c_1t+h_1(\eps))=0$ for every $t\in(a_j,A_j)$, i.e.\ $\Psi_1$ vanishes on the interval
$I_1(\eps):=c_1(a_j,A_j)+h_1(\eps)$. As $\eps$ ranges over the good set (dense, and cofinal in both directions since $h_1(\eps)\to\pm\infty$ by \ref{as:A2}), the
intervals $I_1(\eps)$ have fixed length $|c_1|(A_j-a_j)>0$ and slide continuously with
$h_1(\eps)$; consecutive good $\eps$'s give overlapping intervals (by continuity of $h_1$, since
the excluded bad set is discrete, hence contains, near any point, arbitrarily close good points
on both sides), so their union is a single unbounded interval, in fact all of $\R$. Thus
$\Psi_1\equiv0$ on a set with nonempty interior, and since $\Psi_1=\Th_1''$ is real--analytic on
$\R$ (as $\Th_1$ is, by \ref{as:A3}), the identity theorem gives $\Psi_1\equiv0$ on all of $\R$.
Hence $\Th_1$ is affine, $\Th_1(u)=pu+q$, an exponential polynomial with sole frequency $0$.

\textit{The case $N=2$, reduction to a scalar ODE.}
Write $b_m(\eps):=c_mh_m'(\eps)$ and $u_m:=c_mt+h_m(\eps)$, $m=1,2$; \eqref{eq:N2-reduced} reads
\begin{equation}\label{eq:N2-two}
  b_1(\eps)\Psi_1(u_1)+b_2(\eps)\Psi_2(u_2)=0\qquad\text{on }U_j.
\end{equation}
If $\Psi_2\equiv0$, then \eqref{eq:N2-two} gives $b_1(\eps)\Psi_1(u_1)=0$ for all $(t,\eps)$;
since $b_1\not\equiv0$ and $u_1$ ranges over an open interval as $t$ varies (for $\eps$ with
$b_1(\eps)\neq0$), the same argument as in the case $N=1$ gives $\Psi_1\equiv0$ on $\R$. Both $\Th_1,\Th_2$
are then affine, and the case is settled as in $N=1$. Assume henceforth $\Psi_1,\Psi_2\not\equiv0$.

Since the foliations of the two classes are non--cofoliar, Lemma~\ref{lem:cofoliar} rules out
$h_2=\alpha h_1$, $c_2=\alpha c_1$ for any $\alpha\neq0$; equivalently, writing
\[
  W_{12}(\eps):=c_2h_1'(\eps)-c_1h_2'(\eps),
\]
$W_{12}\not\equiv0$ (a relation $W_{12}\equiv0$ together with strict monotonicity of $h_1,h_2$
would integrate, via $h_1(0)=h_2(0)=0$ under \ref{as:A2}, to $h_2=(c_2/c_1)h_1$, cofoliar with
$\alpha=c_2/c_1$, excluded). As $W_{12}$ is real--analytic and not identically zero, it is
nonzero off a discrete set. Call $\eps$ \emph{admissible} if $b_1(\eps),b_2(\eps),W_{12}(\eps)$
are all nonzero; the inadmissible set is a finite union of discrete sets, hence itself discrete.

Fix an admissible $\eps$. The map $(t,\eps)\mapsto(u_1,u_2)$ has Jacobian
$c_1h_2'(\eps)-c_2h_1'(\eps)=-W_{12}(\eps)\neq0$, so on a neighbourhood of $(t,\eps)$ (for any $t$)
$(u_1,u_2)$ form local coordinates. Consider the vector field
\[
  L:=h_1'(\eps)\,\partial_t-c_1\,\partial_\eps .
\]
Directly, $Lu_1=h_1'(\eps)c_1-c_1h_1'(\eps)=0$ and $Lu_2=h_1'(\eps)c_2-c_1h_2'(\eps)=W_{12}(\eps)$;
and, since $b_m$ depends on $\eps$ alone, $Lb_m=-c_1b_m'(\eps)$.

Applying $L$ to the identity \eqref{eq:N2-two} (valid on all of $U_j$, so $L$ applied to it is
again $\equiv0$ on $U_j$) gives, by the product and chain rules,
\[
  -c_1b_1'(\eps)\Psi_1(u_1)+b_1(\eps)\Psi_1'(u_1)\cdot0
  -c_1b_2'(\eps)\Psi_2(u_2)+b_2(\eps)\Psi_2'(u_2)\,W_{12}(\eps)=0.
\]
Using \eqref{eq:N2-two} to substitute $\Psi_1(u_1)=-(b_2(\eps)/b_1(\eps))\Psi_2(u_2)$ (legitimate
since $b_1(\eps)\neq0$) and simplifying,
\[
  c_1\Big(\frac{b_1'(\eps)}{b_1(\eps)}-\frac{b_2'(\eps)}{b_2(\eps)}\Big)\Psi_2(u_2)
  +W_{12}(\eps)\,\Psi_2'(u_2)=0,
\]
so that, dividing by $W_{12}(\eps)\neq0$,
\begin{equation}\label{eq:N2-ode}
  \Psi_2'(u_2)=\kappa(\eps)\,\Psi_2(u_2),\qquad
  \kappa(\eps):=\frac{c_1\big(b_2'(\eps)b_1(\eps)-b_1'(\eps)b_2(\eps)\big)}{b_1(\eps)b_2(\eps)W_{12}(\eps)}.
\end{equation}

Fix an admissible $\eps_0$. As $t$ ranges over $(a_j,A_j)$, $u_2=c_2t+h_2(\eps_0)$ ranges over
the open interval $I_2(\eps_0):=c_2(a_j,A_j)+h_2(\eps_0)$, and \eqref{eq:N2-ode} says $\Psi_2$
solves the linear ODE $y'=\kappa(\eps_0)y$ (constant coefficient, since $\kappa(\eps_0)$ does not
depend on $t$) on $I_2(\eps_0)$. Since $\Psi_2\not\equiv0$ and is real--analytic, it has at most
isolated zeros on $I_2(\eps_0)$, so there is $u_*\in I_2(\eps_0)$ with $\Psi_2(u_*)\neq0$; the
unique solution of the ODE through $(u_*,\Psi_2(u_*))$ is $Ce^{\kappa(\eps_0)u}$ with
$C=\Psi_2(u_*)e^{-\kappa(\eps_0)u_*}$, and by uniqueness for linear ODEs this solution agrees
with $\Psi_2$ on the whole interval $I_2(\eps_0)$:
\begin{equation}\label{eq:N2-local}
  \Psi_2(u)=C(\eps_0)\,e^{\kappa(\eps_0)u}\qquad(u\in I_2(\eps_0)).
\end{equation}

Let $\eps_0,\eps_1$ be admissible with overlapping intervals, $I_2(\eps_0)\cap I_2(\eps_1)\neq
\varnothing$; this holds whenever $|h_2(\eps_0)-h_2(\eps_1)|<|c_2|(A_j-a_j)$, in particular for
$\eps_1$ close enough to $\eps_0$, by continuity of $h_2$. On the (nonempty, open) overlap,
\eqref{eq:N2-local} gives $C(\eps_0)e^{\kappa(\eps_0)u}=C(\eps_1)e^{\kappa(\eps_1)u}$ for all $u$
in an interval; dividing, $C(\eps_0)/C(\eps_1)=e^{(\kappa(\eps_1)-\kappa(\eps_0))u}$ for all such
$u$, and the left side is constant in $u$ while the right side is constant in $u$ only if
$\kappa(\eps_0)=\kappa(\eps_1)$ (a nonconstant exponential cannot equal a constant on an
interval with more than one point). Hence  also
$C(\eps_0)=C(\eps_1)$.

The inadmissible set is discrete, hence has no accumulation point, so it meets any bounded
interval in only finitely many points. Consequently, for any two admissible $\eps,\eps'$, finitely
many applications of the overlap argument above (stepping across each inadmissible point in
$[\eps,\eps']$ by choosing admissible values arbitrarily close on either side, which exist since
the admissible set is dense) show $\kappa(\eps)=\kappa(\eps')$. Thus $\kappa(\eps_0)\equiv
\lambda$ and $C(\eps_0)\equiv C$ are the same constants for every admissible $\eps_0$.

By \eqref{eq:N2-local}, $\Psi_2(u)=Ce^{\lambda u}$ for every $u\in I_2(\eps_0)$, over all
admissible $\eps_0$. As in Step~2, the admissible set is co--countable and $h_2(\eps)\to\pm\infty$
as $\eps\to\pm\infty$ (\ref{as:A2}), so the intervals $I_2(\eps_0)$, of fixed length
$|c_2|(A_j-a_j)$ and sliding continuously and unboundedly with $h_2(\eps_0)$, have union equal to
all of $\R$. Hence $\Psi_2(u)=Ce^{\lambda u}$ for every $u\in\R$; equivalently
$\Th_2''\equiv Ce^{\lambda u}$ on $\R$.

If $\lambda\neq0$, integrating twice gives
\[
  \Th_2(u)=\frac{C}{\lambda^2}e^{\lambda u}+d_1u+d_2
\]
for constants $d_1,d_2$, an exponential polynomial with frequencies $\{0,\lambda\}$. If $\lambda=0$, then $\Th_2''\equiv C$, so
$\Th_2(u)=\tfrac{C}{2}u^2+d_1u+d_2$, a polynomial of degree at most $2$, an exponential
polynomial with sole frequency $0$. In either case $\Th_2$ is an exponential polynomial.

The equation \eqref{eq:N2-two} is symmetric under exchanging the indices $1,2$ (equivalently,
one may run Steps~3--7 with the vector field $L':=h_2'(\eps)\partial_t-c_2\partial_\eps$, which
satisfies $L'u_2=0$), so the same argument shows $\Th_1$ is an exponential polynomial. Hence, every profile $\Th_m$ appearing in the is an exponential
polynomial. 
\end{proof}

\begin{lemma}\label{lem:persist}
Let $\varphi$ be holomorphic on $\C\setminus S$ with $S$ discrete and no point of $S$
a removable singularity of $\varphi$. Then $\varphi'$ is holomorphic on $\C\setminus S$
(the derivative of a holomorphic function is again holomorphic on the same domain),
and no point of $S$ is a removable singularity of $\varphi'$.
\end{lemma}

\begin{proof}[Proof of Lemma~\ref{lem:persist}]
Only the second assertion requires an argument: we show that if
$\varphi'$ were removable at a point, then $\varphi$ would have been removable there too.
Suppose $\varphi'$ extends holomorphically across some $\zeta_0\in S$. Take a disc $\Delta$ at
$\zeta_0$ with $\Delta\cap S=\{\zeta_0\}$, and let $H$ be a primitive on $\Delta$ of that
extension --- available because $\Delta$ is simply connected. On the punctured disc
$\Delta\setminus\{\zeta_0\}$, which is connected, $(\varphi-H)'=\varphi'-H'\equiv0$, so
$\varphi=H+\mathrm{const}$ there. But $H+\mathrm{const}$ is holomorphic on all of $\Delta$,
so this exhibits a holomorphic extension of $\varphi$ across $\zeta_0$, contradicting that
$\zeta_0$ is non--removable for $\varphi$.
\end{proof}

\begin{proof}[Proof of Theorem~\ref{res:sing}]
    suppose some $\mathcal{J}_{m_0}\neq\varnothing$. Then \eqref{eq:collapse} has no solution. Equivalently, if \eqref{eq:collapse} holds and every $\Th_m$ satisfies \ref{as:Sing}, then every $\Th_m$ is entire. Relabel $m_0=1$ and fix $\zeta_0\in \mathcal{J}_1$.

    With $r:=\max\{1,r_1,\dots,r_N\}$, apply Lemma~\ref{lem:persist} $r-r_m$ times to each
    $\Th_m^{(r_m)}$: by induction, each $\Th_m^{(r)}$ is holomorphic on $\C\setminus \mathcal{J}_m$ and no
    point of $\mathcal{J}_m$ is removable for it. All profiles are now differentiated to the same order,
    with their singular sets unchanged.
    
    Two properties of $\mathcal{J}_m$ are needed. (1) $\mathcal{J}_m\cap\R=\varnothing$: the singularities are complex. Indeed $\Th_m^{(r)}$ is real--analytic on $\R$ by \ref{as:AR} (derivatives of
    real--analytic functions are real--analytic), so a real point $x_0\in \mathcal{J}_m$ would be removable.
    To see this, real--analyticity at $x_0$ means the Taylor series of $\Th_m^{(r)}$ at $x_0$
    converges on a real interval around $x_0$, and the same series, viewed as a function of a
    complex variable, converges on some complex disc about $x_0$ and defines a holomorphic
    function there. This new function agrees with $\Th_m^{(r)}$ itself (as already defined,
    holomorphically, on the punctured disc) along that real interval without $x_0$, a set with
    an accumulation point, so by the identity theorem the two agree on the whole punctured
    disc. Hence $\Th_m^{(r)}$, originally only known to be holomorphic near $x_0$ with $x_0$
    removed, in fact extends holomorphically across $x_0$, this contradicting non--removability.
    (2) $\mathcal{J}_m$ is countable. Being discrete and closed, $\mathcal{J}_m$ can only meet any given closed
    disc in finitely many points, if infinitely many points of $\mathcal{J}_m$ lay in one disc, that disc
    being compact would force those points to accumulate somewhere in it, contradicting
    discreteness. Now cover $\C$ by the closed discs of integer radius centred at the origin,
    $\{|z|\le n\}$ for $n=1,2,3,\dots$: every point of $\C$ lies in one of them, so this is a
    countable collection of discs whose union is all of $\C$. Since $\mathcal{J}_m$ meets each such disc in
    only finitely many points, $\mathcal{J}_m$ itself is a countable union of finite sets, hence countable.
    Consequently, for every $\eps$ the set of $t\in\C$ at which \emph{some} summand of the
    differentiated identity below is singular, namely
    \[
      D(\eps):=\bigcup_{m=1}^N c_m^{-1}\big(\mathcal{J}_m-h_m(\eps)\big),
    \]
    is discrete, closed, countable, and disjoint from $\R$: each constituent is the image of
    $\mathcal{J}_m$ under the affine map $z\mapsto(z-h_m(\eps))/c_m$, which has real
    coefficients and therefore preserves both discreteness and the property of avoiding $\R$;
    and a finite union of closed discrete sets is closed and discrete.
        
    Applying $\partial_\eps\partial_t^{\,r-1}$ to \eqref{eq:collapse}; gives
    \begin{equation}\label{eq:Rr}
      \sum_{m=1}^N c_m^{\,r-1}h_m'(\eps)\,\Th_m^{(r)}\!\big(c_mt+h_m(\eps)\big)
      =\beta(\eps):=\begin{cases} B'(\eps), & r=1,\\[1pt] 0, & r\ge2,\end{cases}
      \qquad\text{on }U,
      \tag{R$^{(r)}$}
    \end{equation}
    constant in $t$ in either case. The single
    $\eps$--derivative is what eliminates $A(t)$, and with it any need to track the singularities
    of the main effect; the $r-1$ derivatives in $t$ raise the profiles to the order at which
    \ref{as:Sing} applies to all of them simultaneously.

    Fix $\eps$. Read \eqref{eq:Rr} as an identity in $t$. Its left side is holomorphic on
    $\Omega_\eps:=\C\setminus D(\eps)$, by construction of $D(\eps)$, and $\Omega_\eps\supset
    (a,A_0)$ because $D(\eps)$ avoids $\R$. Moreover $\Omega_\eps$ is connected: given two of its
    points, consider the uncountable family of circular arcs joining them: distinct arcs meet only
    at the two endpoints, so each point of the countable set $D(\eps)$ obstructs at most one arc,
    and an unobstructed arc remains, giving a path in $\Omega_\eps$. Both sides of \eqref{eq:Rr}
    are therefore holomorphic on the connected open set $\Omega_\eps$ and agree on the interval
    $(a,A_0)$, which has accumulation points; by the identity theorem they agree on all of
    $\Omega_\eps$. The identity, originally valid only on a real strip, now holds for the
    punctured plane.

    The map $t\mapsto c_1t+h_1(\eps)$ is an affine map of $\C$ carrying
    \[
      t_1(\eps):=\frac{\zeta_0-h_1(\eps)}{c_1}
    \]
    to $\zeta_0\in \mathcal{J}_1$; the first summand of \eqref{eq:Rr} has a non--removable singularity at $t_1(\eps)$,
    its scalar prefactor $c_1^{\,r-1}h_1'(\eps)$ being nonzero. As $\eps$
    ranges over $\R$, the point $t_1(\eps)$ traces a real--analytic arc in the complex
    $t$--plane.

    We claim that at each $\eps$ some other summand must be singular at precisely the same point
    $t_1(\eps)$. Suppose not, i.e.\ that for some $\eps$ we have $c_mt_1(\eps)+h_m(\eps)\notin
    \mathcal{J}_m$ for all $m\ge2$. Choose a disc $\Delta$ at $t_1(\eps)$ with $\Delta\cap
    D(\eps)=\{t_1(\eps)\}$, possible since $D(\eps)$ is discrete. Every summand with $m\ge2$ is
    then holomorphic on all of $\Delta$: its singularities lie in $c_m^{-1}(\mathcal{J}_m-h_m(\eps))\subset
    D(\eps)$, and the only point of $D(\eps)$ in $\Delta$ is $t_1(\eps)$, which by assumption is
    not one of them. On $\Delta\setminus\{t_1(\eps)\}\subset\Omega_\eps$ we may therefore
    rearrange \eqref{eq:Rr} to isolate the first summand:
    \[
      c_1^{\,r-1}h_1'(\eps)\,\Th_1^{(r)}\!\big(c_1t+h_1(\eps)\big)
      =\beta(\eps)-\sum_{m\ge2}c_m^{\,r-1}h_m'(\eps)\,\Th_m^{(r)}\!\big(c_mt+h_m(\eps)\big).
    \]
    Every term on the right is holomorphic on the \emph{full} disc $\Delta$, so the left side
    extends holomorphically across $t_1(\eps)$; dividing by the nonzero constant
    $c_1^{\,r-1}h_1'(\eps)$, so does $\Th_1^{(r)}(c_1\cdot+h_1(\eps))$, contradicting
    non--removability. Hence for \emph{every} $\eps\in\R$ there is some $m\ge2$ with
    $c_mt_1(\eps)+h_m(\eps)\in \mathcal{J}_m$. (For $N=1$ there is no such $m$ and the proof already ends
    here.)

    It remains to show that this cannot hold for all $\eps$. Put $w_m(\eps):=c_mt_1(\eps)+h_m(\eps)$
    and $Z_m:=\{\eps\in\R:w_m(\eps)\in \mathcal{J}_m\}$ for $m\ge2$; the previous step says exactly
    $\R=\bigcup_{m=2}^N Z_m$, a finite union. We show each $Z_m$ is countable.
    
    Substituting $t_1(\eps)$,
    \[
      w_m=\tfrac{c_m}{c_1}\zeta_0+h_m-\tfrac{c_m}{c_1}h_1 .
    \]
    Since $h_1,h_m$ are real--valued, the imaginary part is frozen,
    $\Im w_m\equiv\tfrac{c_m}{c_1}\Im\zeta_0$, so $w_m$ moves along a fixed horizontal line;
    whereas the real part moves nondegenerately,
    \[
      (\Re w_m)'=h_m'-\tfrac{c_m}{c_1}h_1'=-\tfrac1{c_1}W_{1m}\not\equiv0
    \]
    by non--cofoliarity of $(c_1,h_1)$ and $(c_m,h_m)$ (Lemma~\ref{lem:cofoliar}). Now $S_m$ is countable, so it meets that horizontal line in a countable set $\{s_{m,k}\}_k$,
    and $w_m(\eps)\in S_m$ forces $\Re w_m(\eps)=\Re s_{m,k}$ for some $k$. Hence
    \[
      Z_m=\bigcup_k\{\eps\in\R:\Re w_m(\eps)=\Re s_{m,k}\},
    \]
    a countable union of zero sets of real--analytic functions, none identically zero because
    $\Re w_m$ is non--constant. The zeros of a non--trivial real--analytic function are isolated,
    hence countably many, so each such set is countable and therefore so is $Z_m$. Thus $\R$ is a
    finite union of countable sets, contradiction.
    
\end{proof}

\subsection{Identification with polynomial functions \(g_k\)}
\label{sec:polynomial-identification}

This section replaces the standing non-polynomial-profile assumption by
the following polynomial assumption; the two assumptions are not imposed
simultaneously. Fix an integer \(D\geq 2\) and set
\[
    D_\star:=\max\{2,D-1\}.
\]

\begin{enumerate}[leftmargin=3.2em]
\Bitem{P}\label{as:polynomial}
\emph{Polynomial profiles.}
For both parameterizations,
\begin{equation*}
    g_k(u)=\sum_{s=0}^{d_k}a_{k,s}u^s,
    \qquad
    \tg_\ell(u)=\sum_{s=0}^{\widetilde d_\ell}
        \widetilde a_{\ell,s}u^s,
\end{equation*}
where
\(
    2\leq d_k,\widetilde d_\ell\leq D,
  \enspace
    a_{k,d_k}\neq0,
\enspace
    \widetilde a_{\ell,\widetilde d_\ell}\neq0.
\)
Coefficients above the degree of a profile are
defined to be zero.

\Bitem{LF}\label{as:LF}
\emph{Pooled finite-power independence.}
For the two parameterizations under comparison, form the pooled warp family
\begin{equation*}
    \mathcal H
    :=
    \{h_k:1\leq k\leq K\}
    \cup
    \{\th_\ell:1\leq \ell\leq\tK\}.
\end{equation*}
Let \(\mathfrak C\) be the set of proportionality classes in
\(\mathcal H\), and choose one representative \(H_C\) from each class.
Then
\begin{equation*}
    \{1\}
    \cup
    \{H_C^s:C\in\mathfrak C,\ 1\leq s\leq D_\star\}
\end{equation*}
is linearly independent as a family of functions of \(\eps\).
\end{enumerate}

The validity of \ref{as:LF} is independent of the representatives \(H_C\):
replacing \(H_C\) by \(\alpha_C H_C\), with \(\alpha_C\neq0\), merely
multiplies the column \(H_C^s\) by the nonzero number \(\alpha_C^s\).
Since the warps are real analytic, \ref{as:LF} is equivalently the
Wronskian condition
\begin{equation}\label{eq:pid-Wronskian}
    W\!\left(
        1,
        (H_C^s)_{\substack{C\in\mathfrak C\\1\leq s\leq D_\star}}
    \right)
    \not\equiv0.
\end{equation}

Assumption \ref{as:LF} is a condition on a \emph{pooled comparison} of two
parameterizations. Accordingly, the theorem below is first stated
pairwise. It implies identifiability of a model class when
\ref{as:LF} is required for every observationally equivalent pair in a distributional sense on $\Xtr$ in
that class.

\begin{lemma}[Differentiated power independence]
\label{lem:pid-degree-shift}
Let \(J\subseteq\R\) be a connected interval, for an integer \(Q\geq0\), and let
\(\{H_C\}_{C\in\mathfrak C}\subset C^1(J)\) be finite. Then
\begin{equation}\label{eq:pid-degree-shift-equivalence}
    \{1\}\cup
    \{H_C^r:C\in\mathfrak C,\ 1\leq r\leq Q+1\}
    \text{ is linearly independent}
\end{equation}
if and only if
\begin{equation*}
    \{H_C^qH_C':
        C\in\mathfrak C,\ 0\leq q\leq Q\}
    \text{ is linearly independent}.
\end{equation*}
If the functions are real analytic, and the pairs \((C,r)\) have the
same ordering on both sides, then
\begin{equation}\label{eq:pid-Wronskian-shift}
\begin{split}
    &W\!\left(
        1,
        (H_C^r)_{\substack{C\in\mathfrak C\\1\leq r\leq Q+1}}
    \right)
    =
    \bigl((Q+1)!\bigr)^{|\mathfrak C|}
    W\!\left(
        (H_C^{r-1}H_C')_
        {\substack{C\in\mathfrak C\\1\leq r\leq Q+1}}
    \right).
\end{split}
\end{equation}
\end{lemma}

\begin{proof}
Suppose first that the family in
\eqref{eq:pid-degree-shift-equivalence} is linearly independent and
\(
    \sum_{C\in\mathfrak C}\sum_{q=0}^{Q}
        b_{C,q}H_C^qH_C'=0.
\)
Since
\(
    H_C^qH_C'
    =\frac{1}{q+1}(H_C^{q+1})',
\)
integration on the connected interval \(J\) gives
\(
    b_0+
    \sum_{C\in\mathfrak C}\sum_{q=0}^{Q}
        \frac{b_{C,q}}{q+1}H_C^{q+1}=0
\)
for some constant \(b_0\). Linear independence forces every
\(b_{C,q}\) and \(b_0\) to vanish.

Conversely, differentiating a relation
\(
    b_0+
    \sum_{C\in\mathfrak C}\sum_{r=1}^{Q+1}
        b_{C,r}H_C^r=0
\)
gives
\(
    \sum_{C\in\mathfrak C}\sum_{r=1}^{Q+1}
        r b_{C,r}H_C^{r-1}H_C'=0.
\)
Independence of the differentiated family gives
\(b_{C,r}=0\) for every \(C,r\), and the original relation then gives
\(b_0=0\).

For \eqref{eq:pid-Wronskian-shift}, expand the Wronskian on the
left along its first column, which is
\((1,0,\ldots,0)^\top\). This gives
\[
    W(1,f_1,\ldots,f_M)=W(f_1',\ldots,f_M'),
\]
where the \(f_a\)'s enumerate the functions \(H_C^r\).
Because
\(
    (H_C^r)'=rH_C^{r-1}H_C',
\)
factoring \(r\) from the corresponding column and multiplying over
all \(C\) and \(1\leq r\leq Q+1\) gives
\eqref{eq:pid-Wronskian-shift}.
\end{proof}

\begin{theorem}[Polynomial-profile identification]
\label{res:polynomial-id}
Let two parameterizations satisfy the standing warp-normalization,
analyticity, support, and within-parameterization non-proportionality
assumptions
\ref{as:A2}--\ref{as:C}.
Assume \ref{as:polynomial}, and suppose that \(\Xtr\) contains, on
each coordinate axis, a nondegenerate segment
\(
    \{t e_j:t\in I_j\},
     0\in I_j.
\)
If the two parameterizations induce the same conditional law on
\(\Xtr\) and their pooled warp family satisfies \ref{as:LF}, then
the model is identifiable in the sense of Definition \ref{DEF_model_is_identifiable}.
\end{theorem}

\begin{proof}
By Lemma~\ref{initial_step_lemma}, equality of the conditional laws
implies the pointwise identity
\begin{equation}\label{eq:pid-pointwise}
\begin{split}
    &\sum_{j=1}^{p}f_j(x_j)
    +\sum_{k=1}^{K}
        g_k\!\left(\beta_k^\top x+h_k(\eps)\right)
        =
    \sum_{j=1}^{p}\tf_j(x_j)
    +\sum_{\ell=1}^{\tK}
        \tg_\ell\!\left(
            \tb_\ell^\top x+\th_\ell(\eps)
        \right),
    \qquad (x,\eps)\in\Xtr\times\R.
\end{split}
\end{equation}

For each pooled class \(C\in\mathfrak C\), let \(k(C)\) be the
untilded index whose warp belongs to \(C\), if such an index exists,
and let \(\ell(C)\) be the corresponding tilded index, if it exists.
These indices are unique by the within-parameterization
non-proportionality assumption. Choose nonzero class scales such that
\begin{equation}\label{eq:pid-class-scales}
    h_{k(C)}=\lambda_{k(C)}H_C,
    \qquad
    \th_{\ell(C)}=\widetilde\lambda_{\ell(C)}H_C
\end{equation}
whenever the relevant indices exist.

Fix a coordinate \(j\) and set \(x=t e_j\) in
\eqref{eq:pid-pointwise}. Applying
\(\partial_t\partial_\eps\) eliminates every term depending only on
\(t\) or only on \(\eps\) and gives
\begin{equation}\label{eq:pid-mixed}
\begin{split}
    &\sum_{k=1}^{K}
        \beta_{kj}h_k'(\eps)
        g_k''\!\left(\beta_{kj}t+h_k(\eps)\right)
    -
    \sum_{\ell=1}^{\tK}
        \tb_{\ell j}\th_\ell'(\eps)
        \tg_\ell''\!\left(
            \tb_{\ell j}t+\th_\ell(\eps)
        \right)=0.
\end{split}
\end{equation}

For one untilded term, the expansion gives
\begin{equation*}
\begin{split}
    &\beta_{kj}h_k'
        g_k''(\beta_{kj}t+h_k)
    =
    \sum_{s=2}^{d_k}\sum_{i=1}^{s-1}
        a_{k,s}s(s-1)\binom{s-2}{i-1}
        \beta_{kj}^{\,i}t^{i-1}
        h_k^{\,s-i-1}h_k'.
\end{split}
\end{equation*}
If \(h_k=\lambda_kH_C\), then
\begin{equation*}
    h_k^{\,s-i-1}h_k'
    =
    \lambda_k^{\,s-i}H_C^{\,s-i-1}H_C'.
\end{equation*}
The analogous formula holds for a tilded term. For each fixed \(\eps\), \eqref{eq:pid-mixed} is a polynomial
identity in \(t\) on the nondegenerate interval \(I_j\), so every
coefficient of \(t^{i-1}\) vanishes. For fixed \(i\), its
\(\eps\)-dependent terms are linear combinations of
\[
    H_C^qH_C',
    \qquad q=s-i-1,\qquad 0\leq q\leq D-2.
\]
Since \(D_\star\geq D-1\), \ref{as:LF} implies independence of
\(\{1,H_C^r:1\leq r\leq D-1\}\); hence
Lemma~\ref{lem:pid-degree-shift}, with \(Q=D-2\), implies
independence of all the displayed functions.

Separating first the powers of \(t\) and then the functions
\(H_C^qH_C'\), and cancelling the common nonzero combinatorial
factor
\(
    s(s-1)\binom{s-2}{i-1},
\)
gives, for every class \(C\), every coordinate \(j\), every
\(2\leq s\leq D\), and every \(1\leq i\leq s-1\),
\begin{equation}\label{eq:pid-mixed-coefficients}
\begin{split}
    &\mathbf 1_{\{k(C)\ {\rm exists}\}}\,
      a_{k(C),s}\,
      \beta_{k(C)j}^{\,i}\lambda_{k(C)}^{\,s-i}
      =
      \mathbf 1_{\{\ell(C)\ {\rm exists}\}}\,
      \widetilde a_{\ell(C),s}\,
      \tb_{\ell(C)j}^{\,i}
      \widetilde\lambda_{\ell(C)}^{\,s-i}.
\end{split}
\end{equation}

Suppose \(C\) contains an untilded index \(k=k(C)\), and let
\(d=d_k\). Choose \(j\in\Btr_k\). In \eqref{eq:pid-mixed-coefficients}, take \(s=d\) and
\(i=1\). Its left-hand side is
\(
    a_{k,d}\beta_{kj}\lambda_k^{d-1}\neq0.
\)
Therefore \(C\) contains a tilded index \(\ell(C)\),
\(\widetilde a_{\ell(C),d}\neq0\), and
\(\tb_{\ell(C)j}\neq0\). In particular,
\(\widetilde d_{\ell(C)}\geq d\). Applying the same argument with
the two parameterizations interchanged gives
\(d\geq\widetilde d_{\ell(C)}\). Hence
\begin{equation*}
    d_k=\widetilde d_{\ell(C)}.
\end{equation*}
The symmetric argument starting from a tilded class shows that every
pooled class contains one index from each parameterization.
Consequently, the class correspondence defines a bijection
\(
    \pi:\{1,\ldots,K\}\rightarrow\{1,\ldots,\tK\},
\)
and \(K=\tK\).

Fix a matched pair \(k\) and \(\ell=\pi(k)\), and denote their common
degree by \(d\). Equation
\eqref{eq:pid-mixed-coefficients} with \(s=d\) and \(i=1\) reads
\begin{equation}\label{eq:pid-leading-mixed}
    a_{k,d}\beta_{kj}\lambda_k^{d-1}
    =
    \widetilde a_{\ell,d}
    \tb_{\ell j}\widetilde\lambda_\ell^{d-1}.
\end{equation}

All factors in \eqref{eq:pid-leading-mixed}, except possibly the coordinate
coefficients, are nonzero. Hence
\(\beta_{kj}=0\) if and only if \(\tb_{\ell j}=0\) for every \(j\), and
therefore \(\Btr_k=\tB_\ell\).

Suppose first that \(d\geq3\), and choose
\(j_0\in\Btr_k=\tB_\ell\). Taking \(s=d\) and \(i=1,2\) in
\eqref{eq:pid-mixed-coefficients} gives
\begin{align}
    a_{k,d}\beta_{kj_0}\lambda_k^{d-1}
    &=
    \widetilde a_{\ell,d}
    \tb_{\ell j_0}\widetilde\lambda_\ell^{d-1},
    \label{eq:pid-leading-i1}\\
    a_{k,d}\beta_{kj_0}^{2}\lambda_k^{d-2}
    &=
    \widetilde a_{\ell,d}
    \tb_{\ell j_0}^{2}\widetilde\lambda_\ell^{d-2}.
    \label{eq:pid-leading-i2}
\end{align}
The quantities divided out below are nonzero because \(j_0\) belongs to the
common support, the class scales are nonzero, and the degree-\(d\)
coefficients are leading coefficients. Dividing
\eqref{eq:pid-leading-i2} by \eqref{eq:pid-leading-i1} yields
\(\beta_{kj_0}/\lambda_k
=\tb_{\ell j_0}/\widetilde\lambda_\ell\).
Substitution into \eqref{eq:pid-leading-i1} then gives
\begin{equation}\label{eq:pid-leading-normalisation}
    a_{k,d}\lambda_k^d
    =
    \widetilde a_{\ell,d}\widetilde\lambda_\ell^d.
\end{equation}
For an arbitrary coordinate \(j\), dividing
\eqref{eq:pid-leading-mixed} by
\eqref{eq:pid-leading-normalisation} gives
\begin{equation}\label{eq:pid-all-coordinate-ratios}
    \frac{\beta_{kj}}{\lambda_k}
    =
    \frac{\tb_{\ell j}}{\widetilde\lambda_\ell},
    \qquad j=1,\ldots,p.
\end{equation}
Thus, with
\(\rho_k:=\lambda_k/\widetilde\lambda_\ell\neq0\),
\eqref{eq:pid-class-scales} and
\eqref{eq:pid-all-coordinate-ratios} imply
\begin{equation}\label{eq:pid-high-degree-structure}
    h_k=\rho_k\th_\ell,
    \qquad
    \beta_k=\rho_k\tb_\ell.
\end{equation}

It remains to compare the nonlinear profile coefficients of this pair. For
\(2\leq s\leq d\), take \(i=1\) and \(j=j_0\) in
\eqref{eq:pid-mixed-coefficients}. Substituting
\(\beta_{kj_0}=\rho_k\tb_{\ell j_0}\) and
\(\lambda_k=\rho_k\widetilde\lambda_\ell\), and then cancelling the
nonzero factor
\(\tb_{\ell j_0}\widetilde\lambda_\ell^{s-1}\), gives
\begin{equation}\label{eq:pid-high-degree-coefficients}
    a_{k,s}
    =
    \widetilde a_{\ell,s}\rho_k^{-s},
    \qquad 2\leq s\leq d.
\end{equation}

Set \(x=0\) in \eqref{eq:pid-pointwise}. The main effects then contribute
only a constant, so, for some \(C_0\in\R\),
\begin{equation}\label{eq:pid-origin}
    \sum_{k=1}^{K}g_k(h_k(\eps))
    -
    \sum_{\ell=1}^{\tK}\tg_\ell(\th_\ell(\eps))
    \equiv C_0.
\end{equation}
For every matched pair of degree at least three,
\eqref{eq:pid-high-degree-structure} and
\eqref{eq:pid-high-degree-coefficients} show that all terms of degree
\(s\geq2\) cancel in \eqref{eq:pid-origin}. Consequently, after collecting
the remaining terms by pooled warp class, the origin identity becomes
\begin{equation}\label{eq:pid-reduced-origin}
\begin{split}
    0={}&C_1
    +\sum_{C\in\mathfrak C}
       \left(
          a_{k(C),1}\lambda_{k(C)}
          -
          \widetilde a_{\ell(C),1}
          \widetilde\lambda_{\ell(C)}
       \right)H_C
    +
    \sum_{\substack{C\in\mathfrak C\\ d_{k(C)}=2}}
       \left(
          a_{k(C),2}\lambda_{k(C)}^2
          -
          \widetilde a_{\ell(C),2}
          \widetilde\lambda_{\ell(C)}^2
       \right)H_C^2
\end{split}
\end{equation}
for another constant \(C_1\). Since \(D_\star\geq2\),
assumption~\ref{as:LF} implies linear independence of
\(\{1\}\cup\{H_C,H_C^2:C\in\mathfrak C\}\). Hence
\eqref{eq:pid-reduced-origin} yields
\begin{align}
    a_{k,1}\lambda_k
    &=
    \widetilde a_{\ell,1}\widetilde\lambda_\ell
    &&\text{for every matched pair},
    \label{eq:pid-linear-coefficient-origin}\\
    a_{k,2}\lambda_k^2
    &=
    \widetilde a_{\ell,2}\widetilde\lambda_\ell^2
    &&\text{for every quadratic matched pair}.
    \label{eq:pid-quadratic-origin}
\end{align}

Now consider a matched pair of degree \(d=2\). At \(s=2\), the mixed
coefficient system \eqref{eq:pid-mixed-coefficients} contains only the
equation corresponding to \(i=1\):
\begin{equation}\label{eq:pid-quadratic-mixed}
    a_{k,2}\beta_{kj}\lambda_k
    =
    \widetilde a_{\ell,2}
    \tb_{\ell j}\widetilde\lambda_\ell.
\end{equation}
Dividing \eqref{eq:pid-quadratic-mixed} by
\eqref{eq:pid-quadratic-origin} gives
\(\beta_{kj}/\lambda_k
=\tb_{\ell j}/\widetilde\lambda_\ell\) for every \(j\). Therefore, upon
setting \(\rho_k:=\lambda_k/\widetilde\lambda_\ell\), the quadratic pair
satisfies
\begin{equation}\label{eq:pid-quadratic-structure}
    h_k=\rho_k\th_\ell,
    \qquad
    \beta_k=\rho_k\tb_\ell,
    \qquad
    a_{k,2}=\widetilde a_{\ell,2}\rho_k^{-2}.
\end{equation}
This is the only point at which the quadratic case requires a separate
argument: the mixed system supplies only
\eqref{eq:pid-quadratic-mixed}, while the coefficient of \(H_C^2\) in
the reduced origin identity supplies the additional relation
\eqref{eq:pid-quadratic-origin}.

For every matched pair, including the quadratic pairs,
\eqref{eq:pid-linear-coefficient-origin} and
\(\lambda_k=\rho_k\widetilde\lambda_\ell\) give
\begin{equation*}
    a_{k,1}
    =
    \widetilde a_{\ell,1}\rho_k^{-1}.
\end{equation*}
Together with \eqref{eq:pid-high-degree-coefficients} for \(d\geq3\) and
\eqref{eq:pid-quadratic-structure} for \(d=2\), this proves
\(a_{k,s}=\widetilde a_{\ell,s}\rho_k^{-s}\) for every
\(1\leq s\leq d_k\). Thus \(g_k\) and
\(\tg_\ell(\,\cdot\,/\rho_k)\) have identical nonconstant
coefficients. Their difference is therefore the constant
\(c_k:=a_{k,0}-\widetilde a_{\ell,0}\), so
\begin{equation}\label{eq:pid-profile-final}
    g_k(u)=\tg_\ell(u/\rho_k)+c_k.
\end{equation}

For \(\ell=\pi(k)\), the structural relations imply
\((\beta_k^\top x+h_k(\eps))/\rho_k
=\tb_{\pi(k)}^\top x+\th_{\pi(k)}(\eps)\).
Substitution of \eqref{eq:pid-profile-final} into
\eqref{eq:pid-pointwise} therefore cancels the matched interaction terms
and leaves
\begin{equation}\label{eq:pid-main-residual}
    \sum_{j=1}^{p}
        \bigl(f_j(x_j)-\tf_j(x_j)\bigr)
    =
    -\sum_{k=1}^{K}c_k,
    \qquad x\in\Xtr.
\end{equation}
On the \(j\)th axis, set \(x=t e_j\) and vary \(t\in I_j\). Every term in
\eqref{eq:pid-main-residual} except the \(j\)th is then constant, so
\(f_j(t)-\tf_j(t)\) is constant on \(I_j\); denote this constant by
\(b_j\). Evaluating the same identity at the origin gives
\(\sum_{j=1}^{p}b_j+\sum_{k=1}^{K}c_k=0\).
\end{proof}

\begin{remark}
\label{rem:LF-versus-growth}
Assumption~\ref{as:LF} is a finite-dimensional linear-independence
condition on the powers of the pooled warp representatives. It is implied by the stronger power-separated growth
condition
\[
    \left|\frac{H_C(\eps)^r}{H_{C'}(\eps)^s}\right|
    \longrightarrow 0
    \quad\text{or}\quad
    \infty
\]
for every two distinct pairs \((C,r)\neq(C',s)\) with
\(1\leq r,s\leq D_\star\): in a nontrivial linear relation, division by
the fastest-growing term gives an immediate contradiction.

The ordinary growth assumption \ref{as:G} compares only the raw warps \(H_C\)
and \(H_{C'}\), and therefore does not by itself imply
\ref{as:LF}. Indeed, raw warps may have distinct growth orders while
their powers satisfy a relation such as in remark \ref{rem:polyG-A6-fails}. Conversely, \ref{as:LF} does not imply \ref{as:G}:
finite-power linear independence can hold even when
\(H_C/H_{C'}\) has a finite nonzero limit. Thus \ref{as:G} and
\ref{as:LF} address different settings, whereas power-separated
growth is a convenient sufficient condition for \ref{as:LF}.

Finally, \ref{as:LF} does not replace the standing normalization,
regularity, support, or within-parameterization non-proportionality
assumptions. \ref{as:LF} is used only to separate the finitely many pooled
warp powers arising from polynomial profiles.
\end{remark}

\providecommand{\R}{\mathbb{R}}
\providecommand{\Z}{\mathbb{Z}}
\providecommand{\eps}{\varepsilon}
\providecommand{\tf}{\tilde f}
\providecommand{\tg}{\tilde g}
\providecommand{\tb}{\tilde\beta}
\providecommand{\tB}{\widetilde{\mathcal B}}
\providecommand{\tK}{\widetilde K}
\providecommand{\Btr}{\mathcal B}
\providecommand{\Xtr}{X_{\mathrm{tr}}}
\providecommand{\supp}{\operatorname{supp}}
\providecommand{\ta}{\tilde a}
\providecommand{\tS}{\widetilde S}
\providecommand{\cP}{\mathcal P}
\providecommand{\cQ}{\mathcal Q}
\providecommand{\cR}{\mathcal R}
\providecommand{\fl}{\ensuremath{{}^{\flat}}}

\subsection{Identifiability for affine $h_k$}\label{ss:affine-id}

This section discusses the case when warps $h_k$ are affine. Every interaction then uses the same noise, up to a scalar (and an intercept) and the warps can no longer be used to differentiate between the interactions. 

The result has two parts. First, two models with different
interaction supports $\mathcal{B}_k$ can give the same conditional law on $\Xtr$. Second, identification is possible when the profiles $\Theta_m$ are independent enough to label the interactions across axes. Independent profiles $\Theta_m$ do here what distinct warps do in the non-affine regime. 

One thing is always recovered. Fix a coordinate $j$, group the interactions active on axis $j$ by their slope. For each slope, the data fix the sum of the corresponding profile classes, taken in $\cQ$ (profiles modulo polynomials). No further resolution is possible from the axis-aligned design alone.

\subsubsection*{Notation}

Write $\cP\subset\mathrm{Meas}(\R)$ for the real polynomials and
$\cQ:=\mathrm{Meas}(\R)/\cP$ for the quotient vector space, with $[g]\in\cQ$ denoting
the class of $g$. For $\rho\neq0$, set $D_\rho[g]:=[g(\rho\,\cdot)]$; this is a well-defined
automorphism of $\cQ$ since dilation preserves $\cP$. A function $g$ is \emph{non-polynomial} when $[g]\neq0$.

Set $b_k:=\beta_k/a_k\in\R^p$ and $g_k^\sharp(s):=g_k(a_k s)$, so that
$g_k(\beta_k^\top x+a_k\eps)=g_k^\sharp(b_k^\top x+\eps)$.
Recall, we write $\Btr_k:=\supp\beta_k$ and $S_j:=\{k:b_{kj}\neq0\}$ for the set of
interactions active on coordinate $j$.

\subsubsection*{Assumptions in the affine regime}
\begin{enumerate}[label=\textnormal{(F\arabic*)},ref=\textnormal{F\arabic*},leftmargin=3.2em]
\item\label{affas:A1} Each $g_k$ is continuous and non-polynomial; each $f_j$ is continuous.
\item\label{affas:A2} The warps are affine and median-zero: $h_k(\eps)=a_k\eps$,
  $a_k\neq0$. 
\item\label{affas:A3} The augmented directions $v_k:=(\beta_k,a_k)\in\R^{p+1}$ are
  pairwise linearly independent; equivalently, the normalised directions $b_k$ are pairwise
  distinct.
\item\label{affas:supp} $|\Btr_k|\ge2$ for every $k$, and $\Btr_k\neq\Btr_l$ for $k\neq l$.
\item\label{affas:Lprofiles} Within each model the classes $[g_k^\sharp]$ are pairwise distinct.
  The distinct classes in
  $\{[g_1^\sharp],\dots,[g_K^\sharp]\}\cup\{[\tg_1^\sharp],\dots,[\tg_{\tK}^\sharp]\}$
  are linearly independent in $\cQ$.
\end{enumerate}

As in the non-affine case, the axis-aligned design is assumed \emph{nondegenerate}: for each coordinate $j$ the
training interval $(a_j,A_j)$ has nonempty interior, so that the strip
$U_j:=(a_j,A_j)\times\R\subseteq\R^2$ is open and nonempty. By Lemma~\ref{initial_step_lemma}, observational equivalence of two parameterizations
on $\Xtr$ reduces to the pointwise identity
\begin{equation}\label{eq:aff-id}
  \sum_{j=1}^p f_j(x_j)+\sum_{k=1}^{K} g_k^\sharp(b_k^\top x+\eps)
  =\sum_{j=1}^p \tf_j(x_j)+\sum_{l=1}^{\tK}\tg_l^\sharp(\tb_l^\top x+\eps)
  \qquad\text{on }\Xtr\times\R.
\end{equation}

\subsubsection*{Results in the affine regime}

\begin{theorem}[Identification in the affine regime]\label{thm:aff-id}
Assume \ref{affas:A1}--\ref{affas:Lprofiles}. Then for any two parametrizations observationally equivalent on $\mathcal{X}_{\mathrm{tr}}$ in distribution have
$K=\tK$, and there exist a permutation $\pi$ and nonzero scalars $\rho_1,\dots,\rho_K$
such that, after relabelling by $\pi$,
\[
  \beta_k=\rho_k\tb_k,\quad a_k=\rho_k\ta_k,\quad \Btr_k=\tB_k,\quad
  g_k(u)=\tg_k\!\bigl(u/\rho_k\bigr)+q_k(u)\quad(q_k\in\cP);
\]
equivalently, $b_k=\tb_{\pi(k)}$ and $[g_k^\sharp]=[\tg_{\pi(k)}^\sharp]$ in $\cQ$. The
number of interactions, the augmented directions up to scale, the supports, and the
$\cQ$-class of each profile are identified; the polynomial parts of the profiles are
determined only up to the residual freedom $\cR$ of Lemma~\ref{lem:aff-indet}.
\end{theorem}

\begin{remark}[On \ref{affas:A1}: non-polynomiality]\label{rem:A1}
The exclusion of polynomial profiles is necessary: a polynomial $g_k$ can be absorbed into the additive part via the polynomial redistribution \ref{ind:poly}, leaving the interaction structure undetermined.
\end{remark}

\begin{remark}[On \ref{affas:A2}: median-zero affine warps]\label{rem:A2}
The median-zero normalization $h_k(\eps)=a_k\eps$ is the affine analogue of the
curved-regime condition and is imposed to fix the location freedom in $h_k$. 
\end{remark}

\begin{remark}[On \ref{affas:A3}: pairwise non-proportional augmented directions]%
\label{rem:A3}
After noise-normalization, \ref{affas:A3} is equivalent to pairwise distinctness of the $b_k$. It ensures that no two interactions are cofoliar across the full design, this makes per-axis slope profiles $\Pi_j$ well defined as $\cQ$-valued invariants
(Proposition~\ref{prop:aff-invariant}) and making the cross-axis matching in the proof
of Theorem~\ref{thm:aff-id} well posed. It does not exclude cofoliarity on a single axis (two terms with $b_{kj}=b_{k'j}$ for some fixed $j$ but $b_k\neq b_{k'}$).
\end{remark}

\begin{remark}[On \ref{affas:supp}: support size and distinctness]\label{rem:S}
The condition $|\Btr_k|\ge2$ ensures that each true index $k$ contributes a nonzero slope on at least two axes, which is used in Step~1 of the proof of Theorem~\ref{thm:aff-id} to certify that each profile class is witnessed. Support distinctness $\Btr_k\neq\Btr_l$ is part of the model identifiability convention and is preserved by all indeterminacies of Lemma~\ref{lem:aff-indet}.
\end{remark}

\begin{remark}[On \ref{affas:Lprofiles}: linear independence of $g$ classes]\label{rem:L}
Condition \ref{affas:Lprofiles} is imposed on the pair of compared parameterizations. Its role is to convert the per-axis equality of $\cQ$-valued sums \eqref{eq:aff-star} into equality of individual terms. The counterexamples in Proposition~\ref{prop:aff-cycle} satisfy \ref{affas:A1}--\ref{affas:supp} but violate \ref{affas:Lprofiles}. The condition is undesirable in that it constrains the two parameterizations jointly rather than each on its own: its second clause refers to the pooled profile classes $[g_k^\sharp]$ and $[\tg_l^\sharp]$ of both models at once, so it cannot be checked against the true parameterization alone.
\end{remark}

\begin{remark}[Residual polynomial indeterminacy]\label{rem:aff-poly}
Under the hypotheses of Theorem~\ref{thm:aff-id}, the $\cQ$-classes $[g_k^\sharp]$ are identified. The polynomial parts are determined only modulo $\cR$. One always has $\dim\cR\ge 2K-1$ (Lemma~\ref{lem:aff-indet}), with equality
when the directions are in \emph{polynomial general position}: for every $d\ge2$, the only
$(c_k)\in\R^K$ with $\sum_k c_k(\beta_k^\top x+a_k\eps)^d$ simultaneously $\eps$-free and
additive in $x$ on the slab is $c=0$. 
\end{remark}

\begin{lemma}[Indeterminacies]\label{lem:aff-indet}
Each of the following transformations transforms a parameterization to an observationally equivalent one on $\Xtr\times\R$, i.e., preserve the equality in \eqref{eq:aff-id}:
\begin{enumerate}[label=\textup{(\roman*)},leftmargin=2.4em,itemsep=1pt]
\item\label{ind:perm} permutation of the $K$ interactions;
\item\label{ind:scale} per-term rescaling
  $(\beta_k,a_k,g_k)\mapsto(\rho_k\beta_k,\rho_k a_k,g_k(\cdot/\rho_k))$, $\rho_k\neq0$;
\item\label{ind:poly} polynomial redistribution $g_k\mapsto g_k+q_k$,
  $f_j\mapsto f_j-r_j$, for any $(q_k)\in\cR$ and the induced $r_j\in\cP$, where
  \[
    \cR:=\Bigl\{(q_k)\in\cP^{K}:\
      \textstyle\sum_{k}q_k(\beta_k^\top x+a_k\eps)
      =\sum_j r_j(x_j)
      \ \text{on }\Xtr\times\R
      \ \text{for some }(r_j)\in\cP^p
    \Bigr\}.
  \]
\end{enumerate}
The space $\cR$ contains the constants $(q_k\equiv c_k)$ and the linear tuples
$(q_k(s)=\lambda_k s)$ with $\sum_k\lambda_k a_k=0$; these span a subspace of dimension
at least $2K-1$.
\end{lemma}

\begin{lemma}[Planar uniqueness modulo polynomials]\label{lem:aff-planar}
Let $U\subseteq\R^2$ be open and connected, let $\xi_1,\dots,\xi_r\in\R\setminus\{0\}$ be
distinct, and let $\Psi_1,\dots,\Psi_r,\Psi_0,\Psi_\infty$ be measurable with
\[
  \sum_{i=1}^{r}\Psi_i(\xi_i t+\eps)=\Psi_0(t)+\Psi_\infty(\eps)
  \qquad\text{on }U.
\]
Then each $\Psi_i$ coincides on the interval $I_i:=\{\xi_i t+\eps:(t,\eps)\in U\}$ with a
polynomial of degree at most $r$. On a strip $U=(a,A)\times\R$ one has $I_i=\R$,
so each $\Psi_i$ is (globally) a polynomial; equivalently, the classes $[\Psi_i]\in\cQ$
are uniquely determined.
\end{lemma}

\begin{lemma}[Master relation]\label{lem:aff-master}
Assume \ref{affas:A1}--\ref{affas:supp}.
For every coordinate $j$ and every $\xi\in\R\setminus\{0\}$,
\begin{equation}\label{eq:aff-star}
  \sum_{\substack{k\in S_j\\ b_{kj}=\xi}}[\,g_k^\sharp\,]
  \;=\;
  \sum_{\substack{l\in\tS_j\\ \tb_{lj}=\xi}}[\,\tg_l^\sharp\,]
  \qquad\text{in }\cQ. \tag{$\star_{j,\xi}$}
\end{equation}
\end{lemma}

\begin{proposition}[Identifiable invariants]\label{prop:aff-invariant}
Assume \ref{affas:A1}--\ref{affas:supp}.
For each coordinate $j$, define the slope profile
\[
  \Pi_j:\R\setminus\{0\}\to\cQ,\qquad
  \Pi_j(\xi):=\sum_{\substack{k\in S_j\\ b_{kj}=\xi}}[\,g_k^\sharp\,].
\]
Any two observationally equivalent (reproducing the same law on $\mathcal{X}_{\mathrm{tr}}$) parameterizations satisfy $\Pi_j=\widetilde\Pi_j$ for
all $j$; each $\Pi_j$ is determined by the conditional law of the response on the strip
$U_j$ alone. 
\end{proposition}

\begin{example}[Non-identifiability]\label{ex:aff-fail}
There exist two parameterizations satisfying \ref{affas:A1}--\ref{affas:supp} that are observationally equivalent on $\Xtr$
yet have distinct interaction supports. In particular they are not related by the
indeterminacies \textup{\ref{ind:perm}}--\textup{\ref{ind:poly}} of
Lemma~\ref{lem:aff-indet}.
\end{example}

\begin{proposition}[Cyclic non-identifiability family]\label{prop:aff-cycle}
For every $m\ge2$, setting $p=2m$ and $K=\tK=m$, there exist two parameterizations
satisfying \ref{affas:A1}--\ref{affas:supp},
with equal additive parts and a common profile $\phi$, that are observationally equivalent
on $\Xtr$ and whose interaction supports are the two cyclic perfect matchings of
$\{1,\dots,2m\}$: $\{\{1,2\},\{3,4\},\dots,\{2m-1,2m\}\}$ and
$\{\{2,3\},\{4,5\},\dots,\{2m,1\}\}$.
\end{proposition}

\subsubsection*{Proof Theorem \ref{thm:aff-id} and with the auxiliary Lemmas}

\begin{proof}[Proof of Lemma~\ref{lem:aff-indet}]
Parts \ref{ind:perm} and \ref{ind:scale} are immediate from the definitions. For
\ref{ind:poly}: adding $q_k$ to $g_k$ shifts the interaction sum by
$\sum_k q_k(\beta_k^\top x+a_k\eps)$, which equals $\sum_j r_j(x_j)$ on $\Xtr\times\R$
by the definition of $\cR$; the compensating substitution $f_j\mapsto f_j-r_j$ restores
\eqref{eq:aff-id}. For the two explicit families in $\cR$: the constants $(q_k\equiv c_k)$
contribute $\sum_k c_k$ (a constant), so $r_j\equiv0$. For a linear tuple with
$\sum_k\lambda_k a_k=0$, evaluating on $x=te_j$ gives
\[
  \sum_k\lambda_k(\beta_{kj}t+a_k\eps)
  =\Bigl(\sum_k\lambda_k\beta_{kj}\Bigr)t+\Bigl(\sum_k\lambda_k a_k\Bigr)\eps,
\]
which is $\eps$-free and equals $\sum_j(\sum_k\lambda_k\beta_{kj})x_j$ across the axes;
hence $(\lambda_k s)\in\cR$ with $r_j(t)=(\sum_k\lambda_k\beta_{kj})t$. Constants occupy
degree $0$ and admissible linear forms degree $1$, so the two families span $\R^K\oplus\R^{K-1}$,
of dimension $2K-1$.
\end{proof}

\begin{proof}[Proof of Lemma~\ref{lem:aff-planar}]
Put $a^i:=(\xi_i,1)$ for $i=1,\dots,r$, $a^{r+1}:=(1,0)$, $a^{r+2}:=(0,1)$, and
$\Psi_{r+1}:=-\Psi_0$, $\Psi_{r+2}:=-\Psi_\infty$, so that the hypothesis reads
$\sum_{m=1}^{r+2}\Psi_m(\langle a^m,z\rangle)=0$ for $z=(t,\eps)\in U$. The directions
$a^1,\dots,a^{r+2}$ are pairwise linearly independent: $a^i,a^{i'}$ ($i\neq i'\le r$) by
$\xi_i\neq\xi_{i'}$; $a^i,a^{r+1}$ since their second coordinates are $1$ and $0$;
$a^i,a^{r+2}$ since $\xi_i\neq0$; and $a^{r+1},a^{r+2}$ trivially.

Fix $i\in\{1,\dots,r\}$. Let $\delta_v F(z):=F(z+v)-F(z)$. For each $m\neq i$ fix
$n_m\perp a^m$ (in the sense $\langle a^m,n_m\rangle=0$), $n_m\neq0$; then
$\theta_m:=\langle a^i,n_m\rangle\neq0$ by independence of $a^i$ and $a^m$. Given $h\in\R$
set $v_m:=(h/\theta_m)n_m$, so $\langle a^m,v_m\rangle=0$ and $\langle a^i,v_m\rangle=h$.
The composite $\prod_{m\neq i}\delta_{v_m}$ annihilates $\Psi_m$ for each $m\neq i$ and
reduces the $i$-th term to $(\Delta_h^{r+1}\Psi_i)(\langle a^i,\cdot\rangle)$; applying it
to $\sum_m\Psi_m(\langle a^m,\cdot\rangle)=0$ yields
\[
  (\Delta_h^{r+1}\Psi_i)(\langle a^i,z\rangle)=0.
\]
This holds for every $z$ such that the $2^{r+1}$ translates $z+\sum_{m\in I}v_m$,
$I\subseteq\{m\neq i\}$, lie in $U$. By openness there exists $\eta>0$ such that this
holds for all $z$ in a neighbourhood of any fixed $z_0\in U$ and all $|h|<\eta$; as $z$
ranges over this neighbourhood, $u=\langle a^i,z\rangle$ covers an open subinterval of
$I_i$. Hence $\Delta_h^{r+1}\Psi_i\equiv0$ on an open subinterval of $I_i$ for all
$|h|<\eta$. By Fr\'echet's theorem (measurable functions annihilated by all sufficiently
small $\Delta_h^{r+1}$ on an interval coincide there with a polynomial of degree $\le r$;
see \cite{aczel1966}, Ch.~3, and \cite{szekelyhidi1991}, Thm.~2.2), $\Psi_i$ is locally
polynomial on $I_i$. Since $I_i$ is an interval (image of the connected set $U$ under a
linear functional), these local polynomials patch to a single polynomial of degree $\le r$
on $I_i$. The uniqueness statement follows by applying the first part to
$\Psi_i-\widetilde\Psi_i$.
\end{proof}

\begin{proof}[Proof of Lemma~\ref{lem:aff-master}]
Substitute $x=te_j$ in \eqref{eq:aff-id}. Terms with $b_{kj}=0$ (resp.\ $\tb_{lj}=0$)
become functions of $\eps$ alone and are absorbed into $B_j(\eps)$; the remaining terms
yield, on the strip $U_j=(a_j,A_j)\times\R$,
\begin{equation}\label{eq:aff-axis}
  \sum_{k\in S_j}g_k^\sharp(b_{kj}t+\eps)-\sum_{l\in\tS_j}\tg_l^\sharp(\tb_{lj}t+\eps)
  =A_j(t)+B_j(\eps).
\end{equation}
Group the left side by slope: for each $\xi$ in the finite set $\Xi_j$ of nonzero slopes
appearing, set $\Phi_{j,\xi}:=\sum_{k\in S_j,\,b_{kj}=\xi}g_k^\sharp
-\sum_{l\in\tS_j,\,\tb_{lj}=\xi}\tg_l^\sharp$. Then
$\sum_{\xi\in\Xi_j}\Phi_{j,\xi}(\xi t+\eps)=A_j(t)+B_j(\eps)$ on $U_j$, with the $\xi$
distinct and nonzero. Lemma~\ref{lem:aff-planar} (strip case) forces each $\Phi_{j,\xi}$
to be a polynomial, i.e.\ $[\Phi_{j,\xi}]=0$, which is \eqref{eq:aff-star}.
\end{proof}

\begin{proof}[Proof of Proposition~\ref{prop:aff-invariant}]
\emph{Reduction to a per-slope identity.} By Lemma~\ref{initial_step_lemma}, observational
equivalence gives \eqref{eq:aff-id} on $\Xtr\times\R$. Put $x=te_j$, so $b_k^\top x=b_{kj}t$,
and let $\Xi_j:=\{b_{kj}:k\in S_j\}\cup\{\tb_{\ell j}:\ell\in\tS_j\}$ be the pooled set of
active slopes. Grouping both models' interactions by slope, the difference of the two responses
on $U_j$ is
\[
  \sum_{\xi\in\Xi_j}\Bigl(\Phi_\xi-\widetilde\Phi_\xi\Bigr)(\xi t+\eps)
  =A_j(t)+B_j(\eps),
  \qquad
  \Phi_\xi:=\!\!\sum_{k\in S_j:\,b_{kj}=\xi}\!\!g_k^\sharp,\quad
  \widetilde\Phi_\xi:=\!\!\sum_{\ell\in\tS_j:\,\tb_{\ell j}=\xi}\!\!\tg_\ell^\sharp,
\]
where an empty sum is $0$; the slopes in $\Xi_j$ are distinct and nonzero, and the noise
coefficient is $1$. The slopes in $\Xi_j$ are distinct and nonzero, and each ridge carries noise coefficient $1$,
so Lemma~\ref{lem:aff-planar} applies to the identity on the strip $U_j$: each summand
$\Phi_\xi-\widetilde\Phi_\xi$ agrees, on the line $I_\xi=\{\xi t+\eps:(t,\eps)\in U_j\}=\R$,
with a polynomial. Being a polynomial on $\R$ is exactly membership in $\cP$, so
$\Phi_\xi-\widetilde\Phi_\xi\in\cP$, i.e.\ $[\Phi_\xi]=[\widetilde\Phi_\xi]$ in
$\cQ=\mathrm{Meas}(\R)/\cP$, for every $\xi\in\Xi_j$.
Writing $\Pi_j(\xi):=[\Phi_\xi]$ (and $0$ off
$\Xi_j$), this is $\Pi_j(\xi)=\widetilde\Pi_j(\xi)$ for every $\xi$. In particular, at a slope
active in only one parameterization the other's class is $0$, so that model's non-polynomial
content at that slope vanishes; the active slopes need not coincide, but they may differ only
where the interaction content is polynomial. As \eqref{eq:aff-id} follows from the conditional
law, so does $\Pi_j=\widetilde\Pi_j$, and each $\Pi_j(\xi)$ is recovered from the law on $U_j$
alone.
\end{proof}

\begin{proof}[Proof of Proposition~\ref{prop:aff-cycle}]
Set $w_r:=e_r+e_{r+1}$ with indices modulo $2m$. The true model uses the $m$ odd directions
$w_1,w_3,\dots,w_{2m-1}$, the alternative uses the $m$ even directions
$w_2,w_4,\dots,w_{2m}$, all with entries $1$ and profile $\phi$. Within each model the
supports $\{r,r+1\}$ are pairwise disjoint and distinct, and the augmented vectors
$(w_r,1)$ have distinct supports, hence are pairwise non-proportional; so \ref{affas:A1}--\ref{affas:supp} hold. Coordinate $j$ belongs to exactly $w_{j-1}$ and $w_j$; on axis $j$
exactly one direction from each model is active (with slope $1$ and argument $t+\eps$)
while the other $m-1$ terms contribute $\phi(\eps)$. Both sides of \eqref{eq:aff-id}
therefore restrict to $\phi(t+\eps)+(m-1)\phi(\eps)$ on each axis and to $m\phi(\eps)$
at the origin, establishing observational equivalence. For $m=2$ this recovers the
construction of Examples in Proposition~\ref{prop:aff-cycle}.
\end{proof}

\begin{proof}[Proof of Theorem~\ref{thm:aff-id}]
Work in normalized coordinates. Each interaction $k$ (of either model) has a profile class
$[g_k^\sharp]\in\cQ$; let $Q_1,\dots,Q_M$ be the distinct values these classes take across both
models, so that every $[g_k^\sharp]$ and every $[\tg_l^\sharp]$ equals some $Q_m$. By
\ref{affas:Lprofiles} the $Q_1,\dots,Q_M$ are linearly independent in $\cQ$.

Fix $j$ and $\xi\neq0$. By Lemma~\ref{lem:aff-master}, \eqref{eq:aff-star}
reads, in $\cQ$,
\[
  \sum_{k\in S_j:\,b_{kj}=\xi}[g_k^\sharp]
  =\sum_{l\in\tS_j:\,\tb_{lj}=\xi}[\tg_l^\sharp].
\]
Collecting equal classes, each side is a combination $\sum_{m}n_m Q_m$ in which the coefficient
$n_m$ counts the interactions of that model with slope $\xi$ whose class is $Q_m$; explicitly,
$n_m=\#\{k\in S_j:b_{kj}=\xi,\,[g_k^\sharp]=Q_m\}$ on the left, and likewise on the right. By
\ref{affas:Lprofiles} the classes within a single model are distinct, so no two interactions
share a $Q_m$ and each $n_m\in\{0,1\}$. Linear independence of $Q_1,\dots,Q_M$ makes the
representation $\sum_m n_m Q_m$ unique, so the two sides are equal only if their coefficients
agree termwise: for every class $Q$,
\begin{equation}\label{eq:aff-count}
  \#\{k\in S_j: b_{kj}=\xi,\,[g_k^\sharp]=Q\}
  =\#\{l\in\tS_j: \tb_{lj}=\xi,\,[\tg_l^\sharp]=Q\}.
\end{equation}

Fix a true index $k$ and set
$Q=[g_k^\sharp]$. Choose $j\in\Btr_k$ and
$\xi=b_{kj}\neq0$. The left count in \eqref{eq:aff-count} is $1$, so the right count is
$1$: some $\tilde\imath$ satisfies $[\tg_{\tilde\imath}^\sharp]=Q$. By symmetry the two class sets are equal.
Since neither model repeats a class, $K=M=\tK$, and
\[
  \pi(k):=\text{the unique $\tilde{\imath}$ with }[\tg_{\pi(k)}^\sharp]=[g_k^\sharp]
\]
is a well-defined bijection.

Fix $k$ and $Q=[g_k^\sharp]=[\tg_{\pi(k)}^\sharp]$.
For every $j\in\Btr_k$, set $\xi=b_{kj}\neq0$; the left count in \eqref{eq:aff-count} is
$1$, so the right count is $1$, realised by $\pi(k)$; hence $\tb_{\pi(k)j}=b_{kj}$. As
this holds for all $j\in\Btr_k$, we have $\Btr_k\subseteq\tB_{\pi(k)}$ and slopes agree
on $\Btr_k$. Applying the argument to $\pi^{-1}$ at $\pi(k)$ gives
$\tB_{\pi(k)}\subseteq\Btr_k$; hence $\Btr_k=\tB_{\pi(k)}$ and $b_k=\tb_{\pi(k)}$.

From $b_k=\tb_{\pi(k)}$, i.e.\
$\beta_k/a_k=\tb_{\pi(k)}/\ta_{\pi(k)}$, set $\rho_k:=a_k/\ta_{\pi(k)}\neq0$; then
$\beta_k=\rho_k\tb_{\pi(k)}$ and $a_k=\rho_k\ta_{\pi(k)}$. From
$[g_k^\sharp]=[\tg_{\pi(k)}^\sharp]$ there exists $p_k\in\cP$ with
$g_k^\sharp=\tg_{\pi(k)}^\sharp+p_k$; unfolding via $g_k^\sharp(s)=g_k(a_ks)$ gives
$g_k(u)=\tg_{\pi(k)}(u/\rho_k)+q_k(u)$ with $q_k(u)=p_k(u/a_k)\in\cP$.
\end{proof}




\end{document}